\documentclass{article}

\usepackage{arxiv}
\usepackage[utf8]{inputenc} 
\usepackage[T1]{fontenc}    
\usepackage{hyperref}       
\usepackage{url}            
\usepackage{booktabs}       
\usepackage{amsfonts}       
\usepackage{amsmath}        
\usepackage{amsthm}         
\usepackage{microtype}  
\usepackage{graphicx}
\usepackage{natbib}
\usepackage{doi}
\usepackage{setspace}
\usepackage{multirow}
\usepackage[title]{appendix}
\usepackage{needspace}

\title{Causally-interpretable meta-analysis using aggregate data}

\author{Qingyang Shi \\
        Faculty of Science and Engineering, \\
        University of Groningen \\
	    \And
        Wouter A.C. van Amsterdam \\
        Julius Center for Health Sciences and Primary Care, \\
        University Medical Center Utrecht \\
        \And
        Sacha la Bastide-van Gemert \\
        Department of Epidemiology, \\
        University Medical Center Groningen \\
        \And
        Talitha Feenstra \\
        Groningen Research Institute of Pharmacy, \\
        University of Groningen \\
        Dutch National Institute of Public Health and the Environment (RIVM) \\     
        \And
        Issa J. Dahabreh \\
        CAUSALab, Department of Epidemiology, \\
        Department of Biostatistics, \\
        Harvard T.H. Chan School of Public Health \\
}

\date{\today}

\renewcommand{\headeright}{Version Nov 2026}
\renewcommand{\undertitle}{First Version: Feb 2026; Current Version: Nov 2026}
\renewcommand{\shorttitle}{CIMAgD}

\newtheorem{Proposition}{Proposition}

\newtheorem{Lemma}{Lemma}
\newtheorem{Corollary}{Corollary}

\hypersetup{
pdftitle={Causally-interpretable meta-analysis using aggregate data},
pdfsubject={stat.ME},
pdfauthor={Qingyang Shi, Wouter van Amsterdam, Sacha la Bastide-van Gemert, Talitha Feenstra, and Issa J. Dahabreh},
pdfkeywords={Meta-analysis, Aggregate data, Summary statistics, Causal inference, Transportability, Conditional average treatment effects}
}

\begin{document}
\begin{doublespace}

\maketitle

\begin{abstract}
Evidence syntheses and meta-analyses are used to inform clinical practice guidelines and health economic evaluations. However, heterogeneity of treatment effects poses a significant challenge. Conventional meta-analysis addresses heterogeneity through random-effects assumptions, which are not supported by design and lead to estimates that may not apply to any real-world population. Causally-interpretable meta-analysis (CIMA) offers a rigorous framework for specification, identification, and estimation of causal effects when combining information from multiple randomized trials. Initial development of CIMA focused on using individual data from randomized trials, but such data are often unavailable in practice. Here, we propose a new version of CIMA that requires only aggregate data from trials, while addressing the limitations of traditional meta-analysis methods. The method leverages the trials' reported estimates of marginal and one-at-a-time subgroup treatment effects and descriptive statistics for baseline covariates to build moment equations for identifying and estimating a parametric conditional average treatment effect (CATE) function. The average treatment effect in a new target population is obtained by marginalizing the CATE function over the individual covariate data that define the target population. The method can also be used to obtain causally-interpretable indirect treatment comparisons in the target population. We establish the asymptotic properties of the method, assess its finite-sample performance in simulation studies, and illustrate the application of the method by re-analyzing a published meta-analysis of SGLT2 inhibitors in patients with heart failure.
\end{abstract}


\section{INTRODUCTION}

Evidence syntheses and meta-analyses are conducted to inform clinical practice guidelines and health economic evaluations by estimating summary results that integrate treatment effect estimates from multiple randomized trials. Heterogeneity of treatment effects poses a major challenge for evidence synthesis. Conventional random-effects meta-analysis attempts to account for heterogeneity through a random-effects assumption — most often by imposing some distributional assumption about the "population" of effects across trials \citep{riley2011interpretation}. However, these assumptions are not supported by design, and various sources of heterogeneity explain the variation in treatment effects \citep{dahabreh2020toward, dahabreh2021study}. As a result, the assumption obscures the actual drivers of heterogeneity and yields estimates that may not apply to any real-world population, rendering them inadequate for context-specific decision-making \citep{dahabreh2020toward}.

Causally-interpretable meta-analysis (CIMA) has been proposed as an approach for clearly defining causal estimands in a target population of substantive interest, for expressing different identification assumptions that link the causal estimand with the observed data, and for using the observed data to estimate target population estimands, including using robust and efficient methods \citep{dahabreh2020toward, dahabreh2023efficient}. Importantly, CIMA allows for the often unknown and often ill-defined populations underlying randomized trials and focuses on estimands that apply to the target population of interest. Initial development of CIMA focused on using individual data collected from all trials. However, individual patient data are often unavailable or inaccessible due to privacy concerns, limiting the scope of applications of CIMA. Aggregate (summary) data extracted from published trial reports, on the other hand, have been used in conventional meta-analysis for decades due to their wide availability, facilitating the widespread adoption of conventional meta-analysis methods, despite their well-known limitations.

Here, we propose a new approach to estimate the average treatment effect in the target population by transporting treatment effects from trials through a parametric approximation to the conditional average treatment effect (CATE) function learned from aggregate data. To identify the CATE function using aggregate data, we leverage information from both marginal and one-at-a-time subgroup treatment effects from trials, and descriptive summary statistics of baseline covariates -- pieces of information that are routinely available in published reports of randomized trial results. For identifiability and stable estimation, we assume a parametric specification for the CATE function and construct moment equations to estimate its parameters via the generalized method of moments (GMM). We then obtain the average treatment effect in the target population by marginalizing the CATE function over the target population covariate distribution. We also show that the method can be immediately extended to causally-interpretable indirect treatment comparisons in the target population. We establish the asymptotic properties of this estimation procedure and conduct a simulation study to assess its finite-sample performance. We also illustrate the application of our methods to a published conventional meta-analysis of SGLT2 inhibitors in patients with heart failure \citep{vaduganathan2022sglt2}. Last, we implement the methods in a new \texttt{R} package \texttt{CIMAgD}, available at: \url{https://github.com/qingyshi/CIMAgD}.

\section{METHODS}

\subsection{Setup}

As in CIMA using individual participant data \citep{dahabreh2023efficient}, we consider a collection of separately conducted randomized trials, indexed by $s \in \mathcal{S}_R = \{1, \ldots, M\}$. In each trial with $S = s$, data are collected from $n_s$ participants on a binary treatment $A_i = a \in \{0, 1\}$, a vector of the baseline covariates $X_i$ representing the patient's characteristics (e.g., age, sex, body mass index, disease severity, etc.), and an outcome of interest $Y_i$. Conditional on trial $S=s$, we model the data as independent and identically distributed draws, $O_{i,s} = (Y_i, A_i, X_i, S_i = s) \sim P_s$, where the only restriction on $P_s$ is the random assignment of treatment. In addition, we consider a target population of interest, indexed by $S = 0$, from which we observe independent and identically distributed draws $O_{i,0} = (X_i, S_i = 0) \sim P_0$, from an unspecified distribution $P_0$.

Let $Y_i(a)$ be the potential outcome of intervention to set treatment $A$ to $a \in \{ 0, 1 \}$. The estimand of interest is the average treatment effect (ATE) in the target population: $\psi = \mathbb{E}[Y(1) - Y(0) | S = 0]$. We may also be interested in subgroup-specific ATEs in the target population, for which we present results in Appendix \ref{app:subgroups}.

Let $g_s(x) \equiv \mathbb{E}[Y(1) - Y(0)|X = x, S = s]$ be the conditional average treatment effect (CATE) functions for each $s \in \{0, \ldots, M\}$. Denote the source support as $\mathcal X_s=\{x:p(x\mid S=s)>0\}$ for $s=0,\ldots,M$. Support requirements are understood almost surely under the indicated covariate laws throughout the paper. We proceed with two assumptions:

\textit{A1. Transportability of CATE function across the trials and the target population}: $g_s(x) = g(x) \equiv \mathbb{E} [Y(1) - Y(0) | X = x]$ for every $s = 0, \ldots, M$ and every $x \in \mathcal X_s$.

\textit{A2. Population overlap}: $\mathcal X_0 \subseteq \cup_{s \in \mathcal S_R} \mathcal{X}_s$.

Under the assumptions \textit{A1} and \textit{A2}, we express the ATE in the target population as:
\[
\psi = \mathbb{E} [Y(1) - Y(0) | S = 0] = \mathbb{E} [g(X) | S = 0],
\]
where $g(x)$ is the CATE function.

It is well-known that the common CATE function can be identified from randomized trials and estimated using individual participant data. We address the setting in which only aggregate data are available from the trials. Let \(X = (X_1, \ldots, X_K)\), where \(X_k\) denotes the \(k\)-th covariate for \(k = 1, \ldots, K\). For the subgroup treatment effects reported by trials, we treat each subgroup-defining covariate as discrete. Thus, if a baseline covariate is intrinsically continuous, we assume that the reported subgroup results are based on a discretized version with finitely many categories (e.g., trial-specific subgroup definitions). This does not preclude using descriptive moments of the underlying continuous covariates, such as means and variances, in the moment-matching step.

Let $X_{k,l}$ be the $l$-th level/stratum of the $k$-th covariate. For each trial \(s\), we observe aggregate data consisting of:
\begin{enumerate}
    \item \textit{Estimated treatment effects}: the trial-level marginal average treatment effect \(\hat \tau_{s,0,0}\), where \(k=0\) and \(l=0\) index the overall (marginal) effect, and one-at-a-time subgroup average treatment effects \(\hat \tau_{s,k,l}\) for each stratum \(X_{k,l}\) of each covariate \(X_k\), with \(k=1, \ldots, K_s\) and \(l=1, \ldots, L_{s,k}\).
    \item \textit{Estimated covariate summary statistics}: descriptive summaries $\hat\mu_s=(\hat\mu_{s,1},\ldots,\hat\mu_{s,R_s})$ estimating $\mu_{s,r}=\mathbb E_{P_s}[h_{s,r}(X)]$ for fixed functions $h_{s,r}$. Examples include means, other moments, and category proportions.
\end{enumerate}

For each reported subgroup let $I_{s,j}(X)$ be its fixed indicator, allowing trial-specific cutpoints and definitions; set $I_{s,j}\equiv1$ for the overall effect. This notation preserves the underlying continuous covariates when a reported subgroup is a bin. Require $p_{s,j} = \mathbb E_{P_s}[I_{s,j}(X)]>0$, and the vector $\tau_{s,0}$ contains
\[
\tau_{s,j,0}=\mathbb E[Y(1)-Y(0)\mid I_{s,j}(X)=1,S=s],
\]
and $\mu_{s,0}$ contains the reported covariate moments. Missing subgroup reports are simply omitted from these vectors, and the identification remains valid under the assumptions stated below. 

\textit{C1. Joint asymptotic linearity of estimators generating reported summaries}: for each trial $s$,
\[
\begin{pmatrix}\hat\tau_s-\tau_{s,0} \\ \hat\mu_s-\mu_{s,0}\end{pmatrix}
=\frac1{n_s}\sum_{i\in \mathcal I_s}
\begin{pmatrix}\phi_s^\tau(O_i) \\ \phi_s^\mu(O_i)\end{pmatrix}
+o_p(n_s^{-1/2}),
\]
where the influence functions have mean zero and finite second moments. For empirical covariate moments, $\phi_s^\mu(X) = h_s(X) - \mu_{s,0}$. More general summary estimators may be used with their corresponding influence functions. We write $\phi_s^{\mu+}=(0, (\phi_s^\mu)^\top)^\top$ for the augmented vector including the fixed normalization moment.

We base our analyses on estimated treatment effects rather than group means, which also facilitates the incorporation of covariate-adjusted effect estimates when available \citep{van2024covariate}. When a trial does not explicitly report treatment effects, we note that sufficient summary information (e.g., event counts or arm-specific means) is typically provided to compute the corresponding effect estimates. Additive treatment effects are the main aggregate outcomes used to identify the additive CATE function. Nevertheless, relative treatment effects, e.g., risk ratios, can also contribute when compatible baseline-risk information is available. We leave the details of using relative treatment effects to Appendix \ref{app:relative-effects}. Appendix \ref{app:arm-means} illustrates, under stronger assumptions, directly modeling the conditional outcome mean functions using group means instead of treatment effects. 

Condition \textit{C1} is mild and is typically satisfied by estimators used in randomized trials, including the difference in means and covariate-adjusted estimators based on estimating equations \citep{tsiatis2008covariate}. Nonetheless, our framework relies on standard causal assumptions holding within each trial---or holding after appropriate adjustments for missing data and/or non-adherence \citep{hernan2020causal}---including consistency, no interference, positivity of treatment assignment, and (conditional) exchangeability. The plausibility of these assumptions must be assessed by the researchers on a trial-by-trial basis for the set of trials included in the analysis. 

Condition \textit{C1} gives a joint central limit theorem for the reported summaries. Define
\[
\Sigma_s=\mathbb E_{P_s}\left[
\begin{pmatrix}\phi_s^\tau \\ \phi_s^\mu\end{pmatrix}
\begin{pmatrix}\phi_s^\tau \\ \phi_s^\mu\end{pmatrix}^\top\right]
=\begin{pmatrix}
\Sigma_s^\tau&\Sigma_s^{\tau\mu}\\
(\Sigma_s^{\tau\mu})^\top &\Sigma_s^\mu
\end{pmatrix}.
\]
The asymptotic sampling variance of $(\hat\tau_s,\hat\mu_s)$ is $\Sigma_s/n_s$.

\textit{C1s. Consistent joint covariance estimation}: $\hat\Sigma_s=\Sigma_s+o_p(1)$ for every trial, including the treatment-effect/covariate-summary cross-covariances. Augmenting $\phi_s^\mu$ by its zero normalization coordinate augments the corresponding covariance blocks by zeros.

Condition \textit{C1s} is technically mild. However, published standard errors or confidence intervals often identify only part of $\Sigma_s$. Missing covariance information does not preclude identification or consistent estimation, but it limits inference. Appendix \ref{app:variance} shows the working approximation under different types of missing reports.

\subsection{Identification}

By assumption \textit{A1}, the treatment effects reported in each trial $s=1,\ldots,M$ can be written as functionals of the CATE function $g$:
\[
(Tg)_{s,j} \equiv \mathbb{E}[ g(X) \mid I_{s,j}(X) = 1, S = s ].
\]
Let $J_s$ denote the number of treatment-effect estimands available in trial $s$, and define
\[
(Tg)_s \equiv \bigl((Tg)_{s,1}, \ldots, (Tg)_{s,J_s} \bigr)^\top\in\mathbb{R}^{J_s}.
\]

To represent the functionals uniformly, we impose domination and integrability conditions on the trial covariate laws:

\textit{A3. Base distribution}: $P_s(X) \ll Q$ with $\frac{dP_s(X)}{dQ(X)} \in L^2(Q)$ for every $s=1,\ldots,M$.

Let the CATE function $g$ be square-integrable, $g \in L^2(Q)$. Then, for each trial $s=1,\ldots,M$, there exist representers $\alpha_{s,0}(x)=(\alpha_{s,1,0}(x),\ldots,\alpha_{s,J_s,0}(x))^\top$ with each $\alpha_{s,j,0}\in L^2(Q)$ such that
\[
(Tg)_{s,j} = \int \alpha_{s,j,0}(x)\,g(x)\,dQ(x), \qquad \alpha_{s,j,0}(x) = \frac{I_{s,j}(x)}{p_{s,j}} \frac{dP_s}{dQ}(x)
\]
for any $g \in L^2(Q)$. Denote $\alpha_{s,0}=(\alpha_{s,j,0})_{j=1}^{J_s}$ and stack them into $\alpha_0 \in L^2(Q)^J$. We use $g(\theta)$ as shorthand for the map $x \mapsto g(x; \theta)$.

Because the aggregate data do not identify the full trial-specific covariate distribution $P_s(X)$ without further structure, we introduce a working family of representer functions through exponential tilting and moment matching using the reported descriptive moments of the covariates \citep{cover2006maximum}. Let $\mu_{s,0}=(\mu_{s,1},\ldots,\mu_{s,R_s})$ be the finite linear moments defined by $\mu_{s,0}=\int h_s(x)\,dP_s(x)<\infty$, where $h_s=(h_{s,1},\ldots,h_{s,R_s})$ is a vector of fixed functions for covariate summary statistics. The moment-matching equations are
\[
\mathbb{E}_Q\bigl[w_s(X;\eta_{s,0})\,h_s^+(X)\bigr]-\mu_{s,0}^+=0_{(R_s+1)},
\]
where $w_s(x;\eta_s)=\exp\{\eta_s^\top h_s^+(x)\}$, $\eta_s\in\mathbb{R}^{R_s+1}$, $h_s^+(x)\equiv(1,h_s(x)^\top)^\top$, and $\mu_{s,0}^+\equiv(1,\mu_{s,0}^\top)^\top$.

This yields the working representer functions
\[
\alpha_{s,j}(x;\eta_s)=\frac{w_s(x;\eta_s)I_{s,j}(x)}{c_{s,j}(\eta_s)}, \qquad
c_{s,j}(\eta_s)=\mathbb E_Q[w_s(X;\eta_s)I_{s,j}(X)],
\]
for each overall or subgroup indicator $I_{s,j}$ for $s = 1, ..., M$ and $j = 1, ..., J_s$.

Exponential tilting is a well-established method for moment matching, and we assume the mild regularity conditions of exponential tilting (i.e., feasibility and finite solution; see details in Appendix \ref{app:asymptotic-proof}). Exponential tilting minimizes the Kullback–Leibler divergence from $Q$ subject to the moment matching. The conditional law of covariates given the moments is inherited from $Q$. Importantly, we do not require the exponential-tilting \(\alpha_s(\eta_s)\) to be correctly specified; that is, the true function \(\alpha_{s,0}\) need not lie in the parametric family \(\{\alpha_s(\eta_s): \eta_s \in \mathbb{R}^{(R_s+1)}\}\). Rather, we use exponential tilting as a convenient device to aggregate information from the reported descriptive moments.

For each trial $s$, let $\tau_{s,0} = (\tau_{s,1,0}, \ldots, \tau_{s,J_s,0})^\top$ denote the vector of true treatment effects. Stacking across trials $s=1,\ldots,M$ yields
\[
\tau_0 = (\tau_{1,0}^\top, \ldots, \tau_{M,0}^\top)^\top \in \mathbb{R}^J,
\qquad
J = \sum_{s=1}^M J_s.
\]
Using information from all trials, we form the stacked moment conditions for $g$:
\[
\int \alpha_0(x)\,g(x)\,dQ(x) - \tau_0 = 0_J,
\]
where the integral is understood componentwise.

A finite set of moments does not generally identify an unrestricted infinite-dimensional CATE. We therefore impose a parametric model \(g(\cdot; \theta)\), indexed by a $d$-dimensional parameter $\theta \in \Theta \subset \mathbb{R}^d$, and assume that:

\textit{A4. Parametric CATE function}: There exists \(\theta_0 \in \Theta\) such that the true CATE function satisfies \(g_0(\cdot) = g(\cdot; \theta_0)\).

\textit{A5. Uniqueness of the CATE parameters}: For every \(\theta \neq \theta_0\),
\[
g(\cdot;\theta) - g(\cdot;\theta_0) \notin \mathcal{A}^\perp, \qquad \mathcal{A} = \mathrm{span}\{\alpha_{s,j,0}:s = 1, \ldots, M, j = 1, \ldots, J_s\}.
\]

A4 specifies the true CATE parametrically; A5 requires its parameters to be distinguishable through their projections onto $\mathcal{A}$. When at least one subgroup effect is reported for a covariate, $\mathcal A$ contains information
about that covariate, so the CATE model may include it. Not every trial needs to report subgroup effects: informally, it typically suffices that one trial reports enough subgroup effects to identify the CATE, or that several trials report complementary subgroups, so that each covariate in the CATE model is addressed by at least one reported subgroup. A linear CATE parametrization can be a useful approximation. Appendix \ref{app:misspecification} studies the linear CATE approximation under misspecification and gives conditions under which its aggregate-data pseudo-true limit agrees with that of an individual data pseudo-outcome regression under the same misspecification. We show the discrepancy error bound for the pseudo-true limits and establish the asymptotic linearity and inference around these limits.

\textit{A6. Moment invariance under exponential tilting}: The population tilting parameters \(\eta_0\) satisfy
\[
\int \bigl\{\alpha_0(x) - \alpha(x;\eta_0)\bigr\}\,g(x;\theta)\,dQ(x)=0_J,
\]
for all $\theta \in \Theta$, where the integrals are understood componentwise.

For the linear CATE $g(x;\theta)=b(x)^\top\theta$ with $b(x)\in\mathbb R^d$, A6 is equivalent to preserving every subgroup basis mean. A sufficient condition for A6 is that the moment matching basis span, $\mathcal H_s=\mathrm{span}\{1,h_{s,1},\ldots,h_{s,R_s}\}$, contains the subgroup basis functions $I_{s,j}$ and $I_{s,j}b_k$ for every relevant $j,k$, that is, the reported covariate summaries are sufficient to support the linear CATE parametrization. Whenever the covariate summaries are insufficient, A6 requires that the remaining subgroup basis means supplied by $Q$ match those of the underlying trial. Appendix \ref{app:local-invariance} provides a diagnostic for this sufficient span condition and explains how it informs the choice of $Q$ for the remaining subgroup basis means. For a general CATE function, A6 states $(T-\tilde T)g(\theta)=0$ for all $\theta$, where $(\tilde Tg)_{s,j}=\mathbb E_{\tilde P_s}[g(X)\mid I_{s,j}=1]$ under the tilted working distribution $\tilde P_s$, that is, the difference between the true and working representers is orthogonal to the span of the candidate CATE functions. Appendix \ref{app:local-invariance} replaces global A6 by a weaker version with invariance on a neighborhood of $\theta_0$ (A6loc), paired with identification in the working span (A5w). Nevertheless, for a linear CATE, this local A6loc implies global A6; we keep this global A6 in the main text and leave the other details to Appendix \ref{app:local-invariance}. 

\begin{Proposition}[Identification of CATE parameters]
\label{prop:cate-identification}
Under assumptions \textit{A1} and \textit{A3} through \textit{A6}, the true parameter $\theta_0$ is the unique element of $\Theta$ satisfying
\[
\int \alpha(x;\eta_0)\,g(x;\theta)\,dQ(x) - \tau_0 = 0_J,
\]
where the integral is understood componentwise.
\end{Proposition}

The base distribution \(Q\) serves as a device for aggregating and transmitting information across sources to facilitate identification. Under assumptions \textit{A3} through \textit{A6}, $Q$ can in principle be chosen flexibly (also see Appendix \ref{app:local-invariance}). Taking $Q$ from the target population data source can be a convenient choice; it does not necessarily mean $Q = P_0(X)$. An enlarged version obtained by relaxing the eligibility criteria may also be useful, as it can provide the dominance at the limit required by A3. A fitted distribution with suitable support is another option, with its estimation uncertainty propagated as described below.

\textit{C2. Regularity conditions for identification}: Let $\theta_0 \in \mathrm{int}(\Theta)$ denote the true parameter, and define the stacked moment function
\[
M(\theta) \equiv \int \alpha(x;\eta_0)\,g(x;\theta)\,dQ(x) - \tau_0 \in \mathbb{R}^J,
\]
where the integral is understood componentwise. We assume that $M:\Theta \to \mathbb{R}^J$ is continuously differentiable on an open neighborhood $\mathcal{N}(\theta_0)$, and that the Jacobian
\[
J^m_\theta(\theta_0) \equiv \partial_\theta M(\theta)\big|_{\theta=\theta_0}
= \int \alpha(x;\eta_0)\,\partial_\theta g(x;\theta_0)\,dQ(x)
= \mathbb{E}_Q\!\left[\alpha(X;\eta_0)\,\partial_\theta g(X;\theta_0)\right]
\in \mathbb{R}^{J\times d}
\]
has full column rank, i.e., $\mathrm{rank}\!\bigl(J^m_\theta(\theta_0)\bigr)=d$.

These conditions imply \(J \ge d\), so the total number of treatment-effect moments across all trials must be at least as large as the number of parameters. This places a practical constraint on the complexity of the chosen parameterization given the available aggregate data. Moreover, full column rank typically requires that the moment functions provide non-redundant information about \(\theta\); heuristically, this is facilitated when the components of \(\alpha = (\alpha_1, \ldots, \alpha_J)\) are sufficiently linearly independent under \(Q\). Appendix \ref{app:asymptotic-proof} shows other sufficient regularity conditions C3 in addition to C2.

\subsection{Estimation}
\label{estimation}

Let $O_s = (O_{1,s}, \ldots, O_{n_s,s})$ be mutually independent samples, with observations drawn i.i.d. from $P_s$ for $s=0,\ldots,M$, and let $X_q = (X_{1,q}, ..., X_{n_q, q}) \sim Q$ be an independent covariate sample of size $n_q$ from the fixed population $Q$. Write $\mathbb Q_{n_q}f=n_q^{-1}\sum_i f(X_{i,q})$ and $\mathbb P_{n_s}f=n_s^{-1}\sum_i f(O_{i,s})$ for integration against empirical measures, and let $n=\sum_{s=0}^M n_s$ be total sample size. Draws from $Q$ may reuse the target samples, or be Monte Carlo draws from a parametric $Q$. Below, we assume independent sampling and address the joint contribution from reusing the target samples in Appendix \ref{app:variance}. 

For any candidate $\eta_s$, define the empirical moment as
\[
\hat m_{s,j}(\theta, \eta_s) = \mathbb Q_{n_q}[\hat\alpha_{s,j}(\eta_s) g(\theta)], \qquad \hat\alpha_{s,j}(x; \eta_s) = \frac{w_s(x;\eta_s) I_{s,j}(x)}{\mathbb Q_{n_q}[w_s(\eta_s)I_{s,j}]}.
\]
Stack $\hat m_{s,j}(\theta,\hat\eta_s)$ over all retained $(s,j)$ and set
\[
\hat M(\theta,\hat\tau,\hat\eta)
=\bigl(\hat m_{s,j}(\theta,\hat\eta_s)\bigr)_{s,j}-\hat\tau,
\qquad \hat\tau=(\hat\tau_s)_{s=1}^M.
\]

The CATE parameter vector $\theta$ is estimated by the generalized method of moments (GMM):
\[
\hat\theta=\arg\min_{\theta\in\Theta}
\hat M(\theta,\hat\tau,\hat\eta)^\top \hat W \hat M(\theta,\hat\tau,\hat\eta),
\]
where $\hat W \to_p W$ is positive-definite, and each $\eta_s$ is solved by the moment matching
\[
\mathbb Q_{n_q}[w_s(\hat\eta_s)h_s^+]-\hat\mu_s^+=0, \qquad 
\hat\mu_s^+=(1,\hat\mu_s^\top)^\top.
\]
The target estimand is estimated by the plug-in estimator:
\[
\hat\psi=\mathbb P_{n_0}[g(\hat\theta)],
\]
that is, by averaging the estimated CATE over the target-population samples. A two-step GMM weighting matrix may use the inverse covariance of the sample moments when that covariance is consistently estimable and nonsingular \citep{hansen1982large}. A fixed positive-definite weighting matrix is also permissible.

Below, we define all Jacobians for the estimation, and all derivatives are evaluated at $(\theta_0,\eta_0)$. Define GMM Jacobians as
\[
J_\theta^m = \mathbb E_Q[\alpha(X;\eta_0)\partial_\theta g(X;\theta_0)]\in\mathbb R^{J\times d}, \qquad
J_m^\theta = -\{(J_\theta^m)^\top WJ_\theta^m\}^{-1}(J_\theta^m)^\top W\in\mathbb R^{d\times J}.
\]
Define aggregate data Jacobians as
\[
J_{\eta_s}^m\in\mathbb R^{J\times(R_s+1)}, \qquad
(J_{\eta_s}^m)_{s,j,\cdot}=\mathbb E_Q[\alpha_{s,j}(X;\eta_{s,0}) \{g(X;\theta_0)-\tau_{s,j,0}\}h_s^+(X)^\top],
\]
with other trials' rows equal to zero, and
\[
J_{\tau_s}^m=-E_s\in\mathbb R^{J\times J_s},
\]
where $E_s$ is the block of columns of the $J\times J$ identity matrix that inserts trial $s$ into the stacked vector.
Define the moment matching Jacobians as
\[
J_{\mu_s}^{\eta_s}=H_s^{-1},\qquad
H_s=\mathbb E_Q[w_s(X;\eta_{s,0})h_s^+(X)h_s^+(X)^\top],
\]
and
\[
J_{\mu_s}^m=J_{\eta_s}^mJ_{\mu_s}^{\eta_s}.
\]

\begin{Proposition}[Asymptotic linearity of the estimators]
\label{prop:asymptotic-linearity}
Assuming A1-A6 and C1-C3, the estimators $\hat \psi$ and $\hat \theta$ admit the following asymptotic linear representations. First, 
\[
\hat\theta-\theta_0=J_m^\theta 
\left\{\mathbb Q_{n_q}[\xi_q] +\sum_{s=1}^M \mathbb P_{n_s}[\xi_s]\right\}
+ o_p(r_n),
\]
where
\[
\xi_q(X)=\phi_q^m(X)-\sum_{s=1}^M J_{\mu_s}^m \tilde \phi_s^{\mu+}(X), \qquad
\xi_s(O)=J_{\tau_s}^m \phi_s^\tau(O) + J_{\mu_s}^m\phi_s^{\mu+}(O),
\]
and
\[
(\phi_q^m(X))_{s,j}=\alpha_{s,j}(X;\eta_{s,0})\{g(X;\theta_0)-\tau_{s,j,0}\}, \qquad
\tilde \phi_s^{\mu+}(X) = w_s(X;\eta_{s,0})h_s^+(X)-\mu_{s,0}^+.
\]
Second,
\[
\hat\psi-\psi_0 = \mathbb P_{n_0}[g(\theta_0) - \psi_0] + J_\theta^\psi(\hat\theta-\theta_0)
+ o_p(n_0^{-1/2}+r_n),
\]
where $J_\theta^\psi=\mathbb E_{P_0}[\partial_\theta g(X;\theta_0)]$ and $r_n=n_q^{-1/2}+\sum_{s=1}^M n_s^{-1/2}$.
\end{Proposition}

The asymptotic linear representations make explicit the first-order propagation of uncertainty from the multiple data sources. For the CATE estimator $\hat\theta$, the first-order uncertainty arises from three components: (1) estimation of the stacked moments using the distribution $Q$; (2) estimation of the aggregate data from trials including treatment effects and covariate summaries; and (3) estimation of the exponential-tilting parameters using the distribution $Q$. When $Q$ is specified artificially and evaluated by Monte Carlo sampling, the contribution of the $Q$-sample is controlled by the Monte Carlo size $n_q$; the dominant sources of uncertainty are typically the trial-based estimates. When $Q$ is fitted parametrically, e.g., a copula model, its parameter uncertainty must also be propagated; increasing Monte Carlo size does not remove that
contribution.

The estimator is mainly intended for meta-analyses that combine multiple randomized trials. Nonetheless, it also applies when only a single trial, or a small number of trials, is available. Its validity does not rely on the number of trials per se; rather, it relies on large-sample behavior within trials, which is the same requirement underlying the original trial-level estimators. In other words, our approach applies directly to single-trial generalizability and transportability analyses when only summary data are available from the trial \citep{dahabreh2019extending, dahabreh2020extending}. Appendix \ref{app:pooling} presents an alternative approach that estimates CATE functions separately for each trial, marginalizes them to the target population to obtain trial-specific transported ATE estimates, and combines these estimates using minimum-variance linear pooling.

When combining multiple trials with different lengths of follow-up, the estimator can be modified to incorporate follow-up time as part of the parameterization of the CATE function. The resulting specification then allows estimation of the corresponding ATEs at the observed follow-up time points in the target population (see Appendix \ref{app:followup}).

\begin{Corollary}[Asymptotic normality]
\label{cor:asymptotic-normality}
Under Proposition \ref{prop:asymptotic-linearity}, suppose $n/n_s\to\pi_s\in(0,\infty)$ for $s=0,\ldots,M$ and $n/n_q\to\pi_q\in[0,\infty)$, with all sample sizes tending to infinity. Then
\[
\sqrt n(\hat\theta-\theta_0)\xrightarrow{d}N(0,V_\theta), \qquad
\sqrt n(\hat\psi-\psi_0)\xrightarrow{d}N(0,V_\psi),
\]
where
\[
V_\theta = J_m^\theta\Omega(J_m^\theta)^\top, \qquad V_\psi=\pi_0\mathrm{Var}_{P_0}(g(\theta_0))+J_\theta^\psi V_\theta(J_\theta^\psi)^\top,
\]
and
\[
\Omega = \pi_q\mathbb E_Q[\xi_q\xi_q^\top]+\sum_{s=1}^M\pi_s\mathbb E_{P_s}[\xi_s\xi_s^\top].
\]
\end{Corollary}

By Corollary \ref{cor:asymptotic-normality}, the first-order sampling variances are $n^{-1}V_\psi$ and $n^{-1}V_\theta$. Under C1s and the plug-in regularity in Appendix \ref{app:variance}, we estimate them by $n^{-1}\hat V_\psi$ and $n^{-1}\hat V_\theta$. For nondegenerate limits, marginal Wald $95\%$ confidence intervals are $\hat\psi\pm1.96\sqrt{n^{-1}\hat V_\psi}$ and, for coordinate $k$, $\hat\theta_k\pm1.96\sqrt{n^{-1}(\hat V_\theta)_{kk}}$. Appendix~\ref{app:covariance-reconstruction} describes working covariance reconstruction when the information required by C1s is incomplete. For a linear CATE approximation, Appendix~\ref{app:misspecification} gives the pseudo-true expansion, and Appendix~\ref{app:variance-misspecification} gives the corresponding variance estimator. These intervals concern the pseudo-true target rather than the true ATE.

\subsection{Indirect treatment comparisons in a target population}

Indirect treatment comparison, also referred to as indirect head-to-head comparison, is an important approach for health technology assessments that allows comparison of two interventions in the absence of direct head-to-head trials. Conventionally, such indirect comparisons are constructed to reflect the trial population of one of the interventions. Here, we propose an indirect comparison in the target population and demonstrate that our proposed method immediately applies to this situation.

Let $A \in \{0,1,2\}$ index three interventions, where $a=0$ denotes the common reference or placebo. Suppose that one or more trials compare either $a=1$ versus $a=0$ or $a=2$ versus $a=0$, while the target comparison of interest is the head-to-head contrast $a=1$ versus $a=2$ in a target population indexed by $S=0$. We define the indirect treatment comparison estimand as
\[
\psi^{12} \equiv \mathbb{E}[Y(2)-Y(1)\mid S=0].
\]

Let $g_s(x, a) \equiv \mathbb{E}[Y(a) - Y(0)| X = x, S = s]$ be the CATE functions for both $a = 1, 2$ and $s = 0, \ldots, M$. Let the index set $\mathcal{S}_{10}$ be the trials with comparison $a = 1$ vs $a = 0$ and the index set $\mathcal{S}_{20}$ be the trials with comparison $a = 2$ vs $a = 0$. We consider the slightly modified assumptions:

\textit{A1e. Transportability of CATE functions across the trials and the target population}: For each $a \in \{1,2\}$,
\[
g_s(x,a)=g(x,a)\equiv \mathbb{E}[Y(a)-Y(0)\mid X=x]
\]
for every $s \in \{0\}\cup \mathcal{S}_{a0}$ and every $x \in \{x: p(x, S=s)>0\}$.

\textit{A2e. Population overlap}: For each $a \in \{1,2\}$,
\[
\{x: p(x,S \in \mathcal{S}_{a0})>0\}\supseteq \{x: p(x,S=0)>0\}.
\]

Under assumptions \textit{A1e} and \textit{A2e}, we express the estimand as:
\[
\psi^{12} = \mathbb{E}[g(X, 2) - g(X, 1)|S = 0].
\]

We identify and estimate $g(x,2)$ using trials that compare $a=2$ versus $a=0$, and $g(x,1)$ using trials that compare $a=1$ versus $a=0$. Depending on data availability, either individual or aggregate data can be used for this step. Here, we focus on estimation of parametric CATE functions using aggregate data.

Let $g(x,a;\hat\theta^a)$ denote the estimated parametric CATE function for treatment level $a\in\{1,2\}$ obtained from the corresponding set of trials. We then estimate the indirect comparison estimand using the plug-in estimator
\[
\hat \psi^{12} = \mathbb P_{n_0}\bigl[g(2;\hat \theta^2) - g(1;\hat \theta^1)\bigr].
\]
Under the corresponding regularity conditions, an argument analogous to Proposition \ref{prop:asymptotic-linearity} gives the asymptotic linear representation
\begin{align*}
    \hat \psi^{12} - \psi_0^{12} &= \mathbb P_{n_0}[g(2; \theta_0^2) - g(1; \theta_0^1) - \psi_0^{12}] + J^\psi_{\theta^1} (\hat \theta^1 - \theta_0^1) + J^\psi_{\theta^2} (\hat \theta^2 - \theta_0^2) \\
    & \quad + o_p\!\left(n_0^{-1/2} + \|\hat \theta^1 - \theta_0^1\| + \|\hat \theta^2 - \theta_0^2\|\right),
\end{align*}
where $J^\psi_{\theta^1} = - \mathbb{E}_{P_0}[\partial_\theta g(X,1; \theta_0^1)] \in \mathbb{R}^{1 \times d_1}$ and $J^\psi_{\theta^2} = \mathbb{E}_{P_0}[\partial_\theta g(X,2; \theta_0^2)] \in \mathbb{R}^{1 \times d_2}$. Consequently,
\[
\sqrt{n}(\hat \psi^{12} - \psi_0^{12}) \xrightarrow{d} N(0, V_\psi^{12}).
\]
A plug-in variance follows Appendix \ref{app:variance}. Contributions from shared target observations, reference observations, or trials must be combined before their covariance is computed.

\section{SIMULATION STUDY}
\label{simulation}

We conducted a simulation study to examine the finite-sample behavior of the proposed method, CIMAgD, under a stylized set of data-generating mechanisms motivated by \citet{dahabreh2023efficient}. We designed the simulation to illustrate whether the method can recover the target-population average treatment effect in the presence of effect modification and differences in covariate distributions between the trials and the target population. We report the simulation study using the aims, data-generating mechanisms, estimands, methods, and performance measures (ADEMP) framework \citep{morris2019using}.

\subsection{Methods}

\subsubsection{Aims}

We aimed to assess whether CIMAgD can recover the average treatment effect in a target population using aggregate data from randomized trials when treatment effects vary with baseline covariates. We considered both a multi-trial setting, to mimic a meta-analysis, and a single-trial setting, to show that the method does not require multiple trials and can be used for single-trial generalizability or transportability analyses. In the multi-trial setting, we compared CIMAgD with standard meta-analysis, meta-regression, an individual participant data (IPD) g-formula benchmark, and an IPD pseudo-outcome regression. In the single-trial setting, we compared CIMAgD with the IPD g-formula and the IPD pseudo-outcome regression.

\subsubsection{Data-generating mechanisms}

We considered settings with a binary treatment and a binary outcome. Appendix \ref{app:sim-add} additionally considers a continuous outcome. We first generated individual-level covariate data from an overall population, and then selected the overall trial population and allocated individuals to each trial. We treated the remaining participants as the target-population sample. Within each trial, we randomly assigned treatment with a fixed probability of 0.5, and generated the potential outcomes for each participant based on the covariates and defined the observed outcome under the treatment actually assigned. Finally, we summarized the individual-level data into aggregate-level data in terms of the treatment effects and covariates.

\begin{enumerate}
    \item Generate the covariates for the overall population from a joint distribution: $(X_1, X_2, X_3) \sim p(x; \eta)$ with $X = (1, X_1, X_2, X_3)^\top$.
    \item Select participants into the overall trial population using a logistic model: $\Pr(S \in \mathcal{S} \mid X) = \frac{\exp(\beta^\top X)}{1 + \exp(\beta^\top X)}$.
    \item Allocate individuals in the trial population to each trial using a multinomial logistic model: $\Pr(S = s \mid S \in \mathcal{S}, X) = \frac{\exp(\gamma_s^\top X)}{\sum_{j = 1}^M \exp(\gamma_j^\top X)}$.
    \item Randomly assign treatment to participants in each trial: $\Pr(A = 1 \mid S = s) = 0.5$.
    \item Generate the potential outcomes using logistic models and define the observed outcome by causal consistency: $\mathrm{logit}\{\mathbb{E}[Y(a)\mid X]\} = \theta(a)^\top X$ for both $a = 0$ and $a = 1$, and $Y = A Y(1) + (1 - A) Y(0)$.
    \item Summarize the individual-level data into aggregate-level data: for the outcomes, $\hat \tau_{s,0,0} = \hat{\mathbb{E}}[Y \mid A = 1, S = s] - \hat{\mathbb{E}}[Y \mid A = 0, S = s]$ and $\hat \tau_{s,k,l} = \hat{\mathbb{E}}[Y \mid A = 1, X_k = x_{k,l}, S = s] - \hat{\mathbb{E}}[Y \mid A = 0, X_k = x_{k,l}, S = s]$ for $s = 1,\ldots,M$, $k = 1,\ldots,K$, and $l = 1,\ldots,L$; for binary covariates, $\hat \mu_{s,k} = \hat{\mathbb{E}}[X_k \mid S=s]$, and for continuous covariates, $\hat \mu_{s,k} = (\hat{\mathbb{E}}[X_k \mid S=s], \hat{\mathrm{SD}}(X_k \mid S=s))$, for $s = 1,\ldots,M$ and $k = 1,\ldots,K$.
\end{enumerate}

Appendix \ref{app:sim-parameters} specifies the parameter values and the exact scenario tables for the data-generating processes, and provides the details of the simulation study implementation. The binary-outcome generation implies the true CATE function as $g_0(x) = \mathrm{expit}\{\theta(1)^\top x\} - \mathrm{expit}\{\theta(0)^\top x\}$. The five-trial setting has total overall sample size 5000, including trial and target participants; expected individual-trial sizes range approximately from 273 to 737 under the specified allocation models. Actual sizes vary randomly. The single-trial setting has overall sample sizes 1000 and 2000. Each main setting has 16 scenarios, with 1000 replications planned per scenario. Appendix \ref{app:sim-add} specifies the additional experiments.

\subsubsection{Estimands}

We considered the average treatment effect in the target population as the target estimand. The estimand is defined as $\psi = \mathbb{E}[Y(1)-Y(0)|S=0]$.

\subsubsection{Methods}

We used five methods in the meta-analysis setting with five trials and three methods in the transportability analysis of a single trial. For the 5-trial setting, we compared: (1) the proposed method CIMAgD, as implemented in the \texttt{R} package \texttt{CIMAgD}; (2) a standard meta-analysis method (\texttt{R} package \texttt{meta}); (3) a standard meta-regression method (\texttt{R} package \texttt{metafor}); (4) an IPD pseudo-outcome regression (\texttt{R} function \texttt{lm}); and (5) an IPD g-formula (\texttt{R} function \texttt{glm}). For the single-trial setting, we considered only the proposed method, the IPD pseudo-outcome regression and the IPD g-formula, because meta-analysis and meta-regression are not applicable.

For the CIMAgD, we used the linear working model for the CATE function $g(X; \theta) = \theta^\top X$, which serves as a linear approximation. We used the target population as the base distribution to match the summary statistics of each trial, including only mean or proportion and/or standard deviation to mimic the information typically available in practice. Only point estimates of treatment effects with marginal standard errors are available without any covariance information. Event numbers from both arms are available for the covariance reconstruction. Appendix \ref{app:sim-parameters} provides the implementation details. 

For the conventional meta-analysis and meta-regression methods, we used random-effects models, and regressed the estimated trial-level treatment effects on study-level mean covariates for meta-regression. The 5-trial setting is intended to not favor meta-regression due to the limited number of trials. We used it as a comparator for the CIMAgD under the same limited number of trials. For meta-analysis, we intended to show the gap between estimators for the parameter of standard meta-analysis models (e.g., the average of the random-effects distribution) and estimators for the target-population causal parameter.

For the IPD pseudo-outcome regression, we used the known propensity score for the pseudo-outcome transformation, and regressed the pseudo-outcome on covariates under the same misspecified linear approximation as the CIMAgD. This is a pooled pseudo-outcome regression over five trials, and serves as a misspecified linear CATE approximation benchmark (also see Appendix \ref{app:misspecification}).

For the IPD g-formula, we used the correctly specified logistic outcome regression, and marginalized it over the target samples. The true CATE specification is implied by the logistic regression. We estimated the standard error using its influence function. Under the correct specification, this yields an efficient plug-in MLE that attains the parametric efficiency bound and serves as a gold-standard benchmark.

\subsubsection{Performance metrics}

We evaluated statistical performance using four metrics: (1) mean squared error $\mathbb{E}[(\hat{\psi} - \psi)^2]$; (2) bias $\mathbb{E}(\hat{\psi} - \psi)$; (3) variance $\mathrm{Var}(\hat{\psi})$; (4) coverage probability of the $95\%$ confidence interval $\Pr\!\left(\psi \in [\hat{\psi} - z_{0.975}\,\mathrm{se}(\hat{\psi}), \hat{\psi} + z_{0.975}\,\mathrm{se}(\hat{\psi})]\right)$, where $\mathrm{se}(\hat{\psi})$ denotes the estimated standard error of $\hat{\psi}$ and $z_{0.975}$ denotes the 97.5th percentile of the standard normal distribution. The true population ATE is evaluated by deterministic integration under each specified covariate and selection model. The coverage is defined for the true target ATE and therefore reflects both approximation bias and miscalibration of standard error. 

\subsection{Results}

Figures \ref{figure1} and \ref{figure2} show the simulation results for the 5-trial and single-trial settings for binary outcome, with 16 scenarios in each setting. Appendix \ref{app:sim-add} shows the additional simulations for continuous outcome and small sample size. Appendix \ref{app:sim-results} provides the detailed numerical results. In the 5-trial setting, both CIMAgD and IPD-PO (pseudo-outcome regression) were approximately unbiased with small variance and MSE across all scenarios, and performed similarly to IPD-GF (g-formula), although they were less efficient. The 95\%CI coverage of CIMAgD ranged from 93.5\% to 95.9\%, indicating similar finite-sample performance compared to IPD-PO with 93.7\% to 96.4\% and IPD-GF with 93.9\% to 96.0\%. The meta-analysis had large bias due to mismatched target by design, leading to poor coverage, which indicates the marginal targets need not agree under covariate distribution shift. The meta-regression was dominated by large variance due to limited number of trials, leading to poor MSE. In contrast, by incorporating subgroup effects, CIMAgD had much lower MSE under the same setup, indicating that subgroup effects supply important information. In the single-trial setting, CIMAgD performed similarly to IPD-PO and IPD-GF in terms of bias, coverage, variance, and MSE across scenarios. For the continuous outcome, CIMAgD performed similarly to the binary-outcome setting. For small trials (approximately 130–170 participants each), CIMAgD, IPD-PO, and IPD-GF showed similar bias, variance, and MSE, but lower coverage than in the main setting: 91.7\% to 95.2\% for CIMAgD, 92.9\% to 95.9\% for IPD-PO, and 93.7\% to 96.1\% for IPD-GF. This indicates that, compared with the IPD-based methods, CIMAgD faces a greater challenge in reconstructing the covariance with small samples.

\begin{figure}
    \centering
    \includegraphics[width=\textwidth]{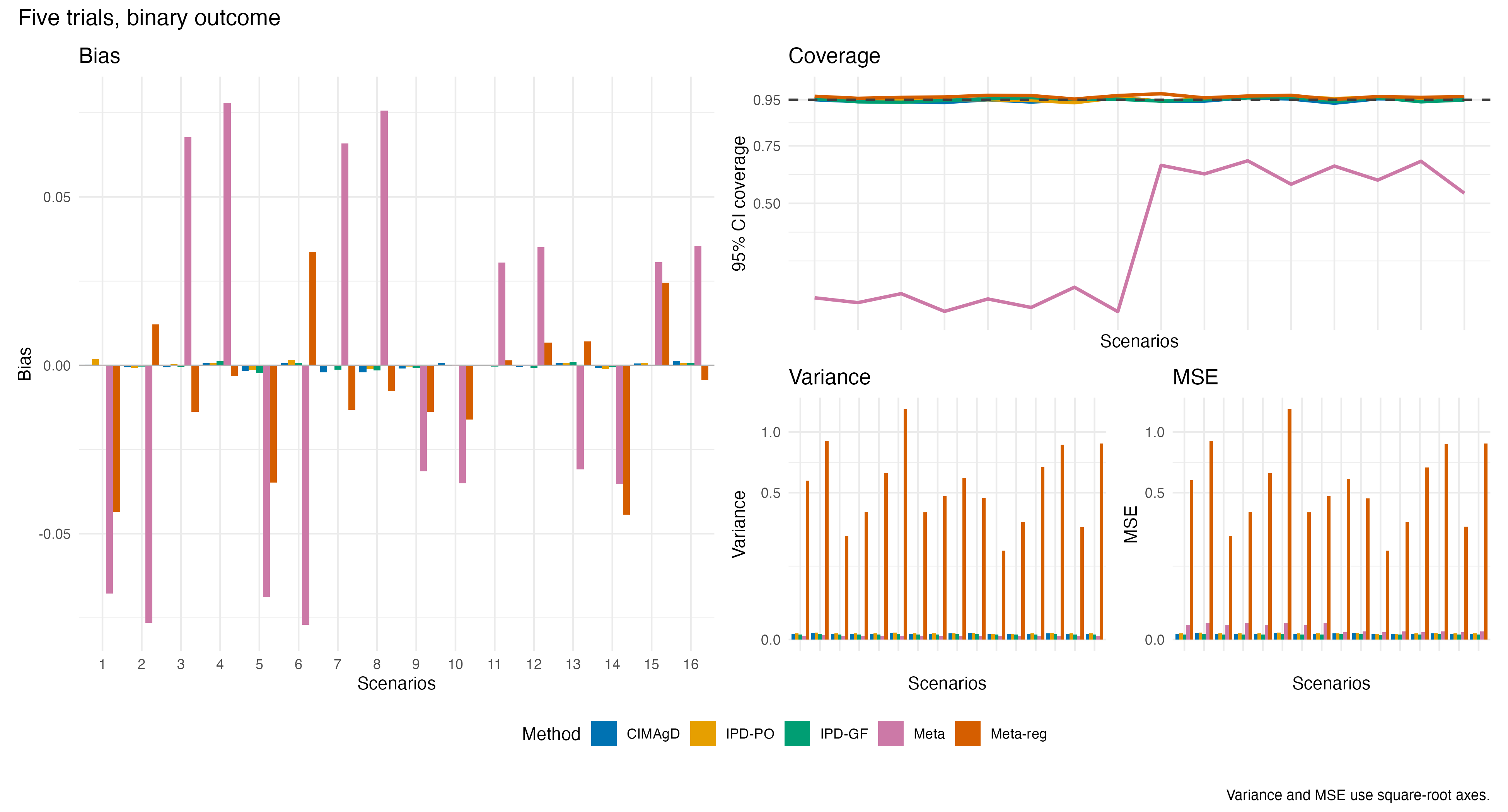}
    \caption{Simulation results in 5-trial setting}
    \label{figure1}
\end{figure}

\begin{figure}
    \centering
    \includegraphics[width=\textwidth]{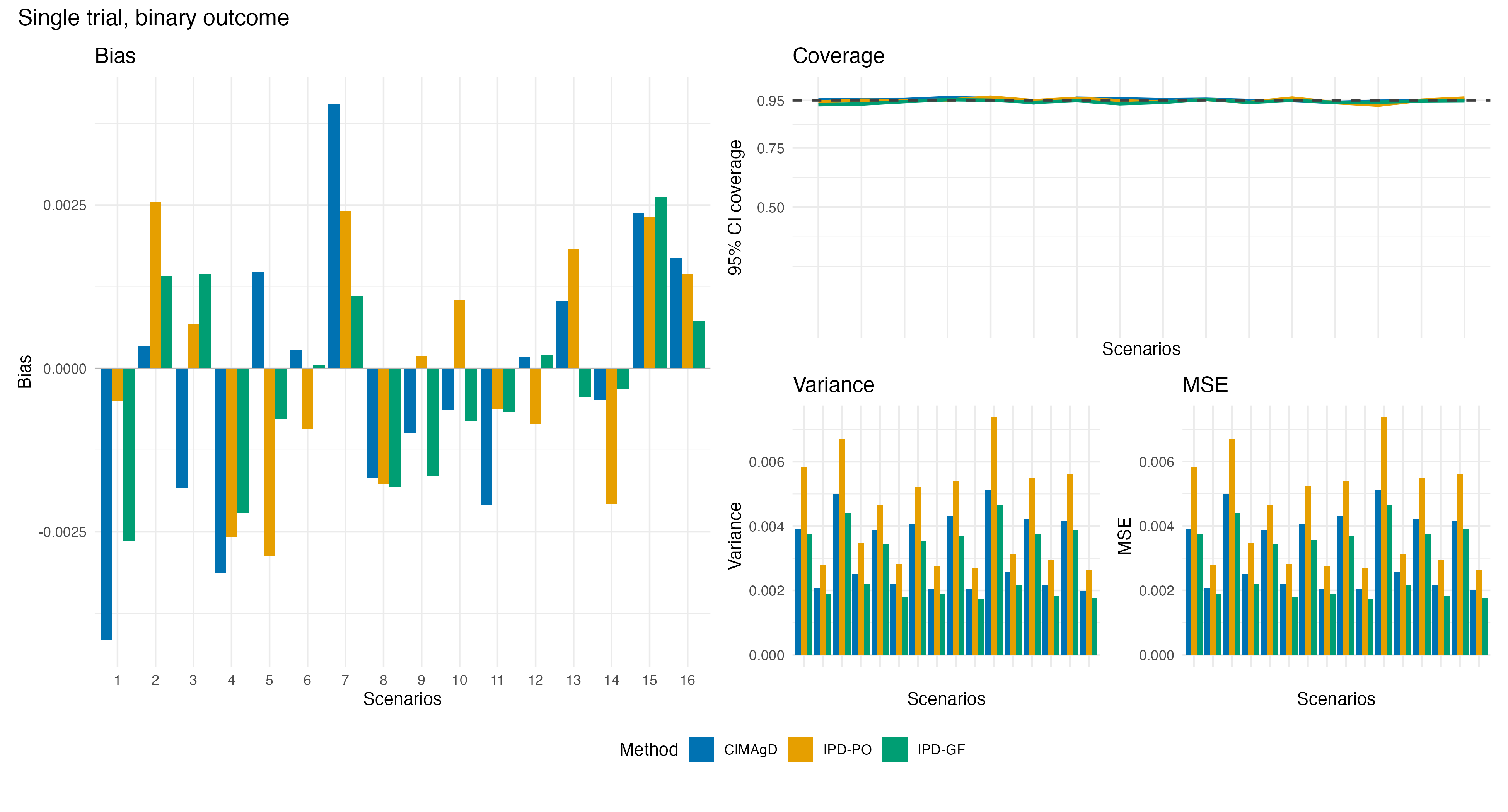}
    \caption{Simulation results in single-trial setting}
    \label{figure2}
\end{figure}

\section{APPLICATION}

\subsection{Background}

Patients with heart failure face increased risks of disability, morbidity, and mortality. Several large randomized trials have shown that SGLT2 inhibitors reduce the risk of heart failure hospitalization and cardiovascular death. However, heterogeneity in patient characteristics, including left ventricular ejection fraction (LVEF), history of hospitalization for heart failure (HHF), and diabetes status, may lead to heterogeneous treatment effects. A recent meta-analysis \citep{vaduganathan2022sglt2} examined the treatment effects of SGLT2 inhibitors across several subgroups, one at a time. That study found that the treatment effects were generally consistent across 14 subgroups, including LVEF, HHF, and diabetes. However, those results relied on separate statistical tests for subgroup variables considered one at a time. Such analyses do not necessarily imply homogeneous treatment effects with respect to those subgroup variables jointly, and they depend on dichotomous decisions based on $p$-value thresholds. In contrast, our proposed method does not test subgroups one at a time. Instead, it directly targets a new decision population and assumes a transportable underlying CATE function, which it estimates using the subgroup information reported by the trials. This approach avoids test-based decisions and directly targets a well-defined estimand for decision-making. In this application, we estimated the average treatment effect in a new target population using aggregate data reported from the trials included in that meta-analysis, and we explored joint subgroup treatment effects in the target population.

\subsection{Data}

We used aggregate data extracted from published reports of four large randomized cardiovascular outcome trials (CVOTs) of SGLT2 inhibitors in patients with heart failure: EMPEROR-Reduced, EMPEROR-Preserved, DAPA-HF, and DELIVER. The outcome of interest was the primary composite outcome defined in all four trials: cardiovascular death or hospitalization for heart failure. We extracted the event numbers, marginal treatment effects and subgroup treatment effects from the four trials; Table \ref{table1} summarizes the extracted treatment-effect data. We also extracted summary statistics for each covariate from the trial reports; Table \ref{table2} summarizes the extracted covariate information.

\begin{table}
\caption{Treatment effects in trials}
\renewcommand{\arraystretch}{1.5}
\centering
\begin{tabular}{lllllllll}
\toprule
\multirow{2}{*}{Trial} & \multirow{2}{*}{Subgroup} & \multirow{2}{*}{Level} & \multicolumn{2}{l}{Treatment group} & \multicolumn{2}{l}{Control group} & \multirow{2}{*}{Risk difference} & \multirow{2}{*}{SE} \\
  &       &        & Event   & N   & Event     & N     &           &       \\
\midrule
EMPEROR-Preserved & overall & —  & 415  & 2997 & 511 & 2991 & -0.032 & 0.009   \\
EMPEROR-Preserved & LVEF  & 40-49\%  & 145   & 995   & 193     & 988   & -0.050  & 0.017   \\
EMPEROR-Preserved & LVEF    & 50-59\%  & 138   & 1028   & 173  & 1030   & -0.034  & 0.016   \\
EMPEROR-Preserved & LVEF    & $\ge$60\%   & 132  & 974   & 145   & 973   & -0.014     & 0.016    \\
EMPEROR-Preserved & preHHF& yes& 157& 699& 192& 670& -0.062& 0.023\\
EMPEROR-Preserved & preHHF& no& 258& 2298& 319& 2321& -0.025& 0.009\\
EMPEROR-Preserved & diabetes   & yes   & 239    & 1466  & 291  & 1472   & -0.035   & 0.014    \\
EMPEROR-Preserved & diabetes    & no    & 176   & 1531   & 220   & 1519  & -0.030    & 0.012  \\
DELIVER           & overall     & —   & 475    & 3131   & 577    & 3132   & -0.033    & 0.009    \\
DELIVER           & LVEF      & 40-49\%   & 193   & 1067  & 220    & 1049  & -0.029   & 0.017  \\
DELIVER           & LVEF        & 50-59\%   & 161   & 1133  & 196  & 1123   & -0.032  & 0.015   \\
DELIVER           & LVEF        & $\ge$60\%   & 121 & 931   & 161  & 960  & -0.038  & 0.016     \\
DELIVER           & preHHF& yes& 184& 829& 230& 805& -0.064& 0.021\\
DELIVER           & preHHF& no& 291& 2302& 347& 2327& -0.023& 0.010    \\
DELIVER           & diabetes   & yes  & 248   & 1401 & 298  & 1405  & -0.035  & 0.015    \\
DELIVER           & diabetes     & no  & 227   & 1730     & 279  & 1727   & -0.030  & 0.012  \\
DAPA-HF           & overall     & —   & 382   & 2373     & 495  & 2371   & -0.048  & 0.011    \\
DAPA-HF           & preHHF& yes& 117& 638& 181& 663& -0.089& 0.023\\
DAPA-HF           & preHHF& no& 265& 1735& 314& 1708& -0.031& 0.013   \\
DAPA-HF           & diabetes     & yes   & 213      & 1075     & 268   & 1064     & -0.054   & 0.018   \\
DAPA-HF           & diabetes     & no    & 169      & 1298     & 227    & 1307     & -0.043   & 0.014    \\
EMPEROR-Reduced   & overall      & —    & 361      & 1863     & 462   & 1867     & -0.054  & 0.014   \\
EMPEROR-Reduced   & preHHF& yes& 153& 577& 177& 574& -0.043& 0.027\\
EMPEROR-Reduced   & preHHF& no& 208& 1286& 285& 1293& -0.058& 0.015    \\
EMPEROR-Reduced   & diabetes    & yes    & 200   & 927     & 265   & 929   & -0.070  & 0.020    \\
EMPEROR-Reduced   & diabetes    & no    & 161    & 936     & 197   & 938      & -0.038   & 0.018    \\
\bottomrule
\end{tabular}
\par\smallskip
\begin{minipage}{\textwidth}\footnotesize SE: standard error of the risk difference, computed from the reported event counts.\end{minipage}
\label{table1}
\end{table}

\begin{table}
\caption{Summary statistics of covariates}
\renewcommand{\arraystretch}{1.5}
\centering
\begin{tabular}{llllll}
\toprule
Trial             & Covariate    & Type     & N    & Mean  & SD  \\
\midrule
EMPEROR-Preserved & LVEF         & Continuous & 5988 & 54.3  & 8.8 \\
EMPEROR-Preserved & preHHF-yes  & Binary     & 5988 & 0.229 & —   \\
EMPEROR-Preserved & diabetes-yes & Binary     & 5988 & 0.491 & —   \\
DELIVER           & LVEF         & Continuous & 6263 & 54.0  & 8.9 \\
DELIVER           & preHHF-yes  & Binary     & 6263 & 0.261 & —   \\
DELIVER           & diabetes-yes & Binary     & 6263 & 0.448 & —   \\
DAPA-HF           & LVEF         & Continuous & 4744 & 31.1  & 6.8 \\
DAPA-HF           & preHHF-yes  & Binary     & 4744 & 0.274 & —   \\
DAPA-HF           & diabetes-yes & Binary     & 4744 & 0.451 & —   \\
EMPEROR-Reduced   & LVEF         & Continuous & 3730 & 27.2  & 6.1 \\
EMPEROR-Reduced   & preHHF-yes  & Binary     & 3730 & 0.309 & —   \\
EMPEROR-Reduced   & diabetes-yes & Binary     & 3730 & 0.498 & —   \\
\bottomrule
\end{tabular}
\label{table2}
\end{table}

\subsection{Analysis}

We applied the proposed method to transport the treatment effects reported in the published trials to a new target population, namely the PULSE study \citep{bellanca2023incidence}, a linked primary and secondary care database for heart failure in the UK. We considered three effect modifiers: LVEF, history of hospitalization for heart failure (HHF), and diabetes status. The PULSE study included 383,896 patients with heart failure, with a mean age of 75 years, 53\% male, a mean LVEF of 45.4, 9.1\% with a history of HHF, and 26.5\% with diabetes (see \citet{bellanca2023incidence} for additional details). We used the reported covariate summary statistics from this target population to generate target-population covariate data using a parametric copula model. We then used the resulting generated covariate data to marginalize the estimated CATE function over the target population. We estimated the marginal average treatment effect in the target population and eight subgroup treatment effects defined jointly by the three covariates. We implemented the proposed method using the \texttt{R} package \texttt{CIMAgD}.

\subsection{Results}

Table \ref{table3} summarizes the results. Overall, SGLT2 inhibitors were associated with 34 fewer composite events of cardiovascular death or hospitalization for heart failure per 1000 patients in the target population. Across the joint subgroups, the estimated treatment effect ranged from a reduction of 73 events per 1000 patients with LVEF $\le 40\%$, a history of HHF, and no diabetes to a reduction of 26 events per 1000 patients with LVEF $> 40\%$, no history of HHF, and diabetes. These results may be taken to suggest treatment-effect heterogeneity across clinically relevant subgroups in the target population; the largest difference between the joint subgroups (47 events per 1000 patients) is potentially clinically relevant. Statistical tests comparing subgroup-specific estimates can be obtained from their joint asymptotic linear expansions (Appendix A.1); here we focus on magnitude rather than testing for illustration purposes.

\begin{table}
\caption{Estimated treatment effects in the target population}
\renewcommand{\arraystretch}{1.5}
\centering
\begin{tabular}{llllll}
\toprule
& \multicolumn{3}{l}{Subgroup variables} &        &                  \\
& LVEF    & preHHF    & diabetes & Risk difference & 95\% CI           \\
\midrule
Overall   & --              & --     & --       & -0.034  & (-0.045, -0.024) \\
          & $\le 40$        & no     & no       & -0.041  & (-0.060, -0.023) \\
          & $> 40$          & no     & no       & -0.026  & (-0.041, -0.012) \\
          & $\le 40$        & yes    & no       & -0.073  & (-0.107, -0.039) \\
          & $> 40$          & yes    & no       & -0.058  & (-0.092, -0.025) \\
          & $\le 40$        & no     & yes      & -0.041  & (-0.065, -0.016) \\
          & $> 40$          & no     & yes      & -0.026  & (-0.047, -0.005) \\
          & $\le 40$        & yes    & yes      & -0.071  & (-0.097, -0.045) \\
          & $> 40$          & yes    & yes      & -0.057  & (-0.081, -0.032) \\
\bottomrule
\end{tabular}
\label{table3}
\end{table}

\section{DISCUSSION}

Conventional meta-analyses based on aggregate data often do not target a clearly defined decision population. When covariate distributions differ between trial populations and the population in which decisions will be made, trial-level marginal effects need not equal the average treatment effect in that target population. In this paper, we addressed this problem by extending recent work on causally interpretable meta-analysis to the setting in which only aggregate trial data are available \citep{dahabreh2020toward,dahabreh2023efficient}. Our approach uses a transportable CATE function to link trial results to a target-population estimand and thereby supports decision-making in a clearly defined population.

The proposed method uses only the types of aggregate information that are often available from randomized trials, namely treatment-effect estimates, subgroup treatment-effect estimates, and descriptive covariate summaries. It combines these summaries under an assumption that the CATE is transportable, and uses a base distribution $Q$ to connect information across trials and the target population. The resulting estimator remains grounded in a clearly defined target population and supports standard large-sample inference under the stated assumptions. Importantly, the method does not require a large number of trials; rather, its validity depends on large-sample behavior within trials and on the identification conditions for the chosen CATE model. For this reason, the same approach can also be used in settings with only a few trials or even a single trial when the identifying information is sufficient. When each trial reports enough information to identify its own CATE, our method also provides an alternative approach that estimates and marginalizes CATEs separately for each trial to the target population, and combines these estimates using minimum-variance linear pooling. In addition, although we focus on aggregate data from randomized trials, the same approach could also use aggregate estimates from observational studies, provided they are adequately adjusted for confounding and target comparable causal estimands.

Our simulation study showed that the proposed method performed well in the settings we considered, including both multiple-trial and single-trial settings. The method produced estimates that were close to the target-population estimand and performed similarly to both IPD pseudo-outcome regression and IPD g-formula benchmark. In contrast, conventional meta-analysis and meta-regression performed less well for this task because they do not directly target the same estimand or use the same subgroup information.

Recent work has developed several methods that combine individual participant data and aggregate data in evidence synthesis, including matching-adjusted indirect comparison (MAIC) \citep{signorovitch2012matching}, simulated treatment comparison (STC) \citep{remiro2022parametric}, multilevel network meta-regression (ML-NMR) \citep{phillippo2020multilevel}, and more recent methods that integrate IPD and aggregate data within causal or network-based frameworks \citep{rott2025causally,vo2025integration}. These methods address important settings, particularly when individual-level data are available from at least one study. In contrast, our approach requires only aggregate trial data, targets a user-specified population rather than a trial-defined population, and identifies the target-population estimand through an explicit transportable CATE model. It also provides a large-sample theory that makes the role of aggregate-data uncertainty transparent. In this sense, our method extends causally interpretable evidence synthesis to settings where individual-level trial data are unavailable while retaining a more explicit causal identification and inferential framework. In addition, using aggregate data only, \citet{hasegawa2017myth} clarified the target of a standard meta-analysis by defining it as a mixture of the trial populations with explicit weights and extrapolated these weights toward an external population; in contrast, our approach targets an explicit target-population estimand through a transportable CATE function.

A limitation of our method is that it relies on subgroup treatment effects reported by the trials. These subgroup summaries provide important information for identifying the CATE function. In practice, trials may omit subgroup results. Nevertheless, our method allows using subgroup information in a flexible way when trials omit some of subgroup results, report them inconsistently across studies, or define subgroups in ways that are not well aligned across trials. However, when the available subgroup information is sparse or poorly aligned with the chosen CATE model, identification may fail or the resulting estimates may become unstable. For a linear CATE approximation, identification can be diagnosed directly by the rank of the moment Jacobian matrix (condition C2), which provides a practical empirical diagnostic. The method also assumes that the reported subgroup effects are not selectively reported. Under selective reporting, for example, based on statistical significance, the estimator remains consistent as long as the reported effects still identify the CATE parameters, but valid inference generally requires assumptions about the selection mechanism; sensitivity analyses, such as excluding trials suspected of selective reporting, can help assess robustness.

In conclusion, we developed a method for estimating average treatment effects in a user-specified target population using aggregate data from randomized trials. By linking reported trial summaries to a transportable CATE function, the method enables causally interpretable evidence synthesis when individual-level trial data are unavailable. This framework may be particularly useful for clinical decision-making, guideline development, and health technology assessment, which often rely on published aggregate data and require treatment-effect estimates for clearly defined target populations rather than ill-defined mixtures of trial populations.

\section{ACKNOWLEDGMENTS}

We thank Veerle Coupé, Paul Gustafson, Spyros Balafas, and Steef Konings for their helpful comments and feedback on an early stage of this work. The authors obtained all results presented in the manuscript and drafted and revised the manuscript. Key mathematical statements were verified using ChatGPT (OpenAI, https://chatgpt.com/). The authors are responsible for the final version of all results and the entire manuscript.

\section{DATA AVAILABILITY}

The code for the simulation study is publicly accessible at: \url{https://github.com/qingyshi/CIMAgD-Sim}. The data used in the application were extracted from the published studies and are presented in the tables.

\bibliographystyle{unsrtnat}
\bibliography{references}

\clearpage
\begin{appendices}
\section{METHODS}
\label{app:methods}

\Needspace{6\baselineskip}
\subsection{Subgroup estimands}
\label{app:subgroups}

In addition to the marginal treatment effect in the target population, we may also be interested in the subgroup-specific treatment effects in terms of the subgroup variable $X_k$ in the target population, defined as:
\[
\psi(x_k) \equiv \mathbb{E}[Y(1) - Y(0) | X_k = x_k, S = 0],
\]
or the joint subgroup, for example, for both $X_1$ and $X_2$ as:
\[
\psi(x_1,x_2) \equiv \mathbb{E}[Y(1) - Y(0) | X_1 = x_1, X_2 = x_2, S = 0].
\]

With assumptions \textit{A1} and \textit{A2}, we express the subgroup estimand in the same way as the marginal estimand:
\[
\psi(x_k) \equiv \mathbb{E}[Y(1) - Y(0) | X_k = x_k, S = 0] = \mathbb{E}[g(X) | X_k = x_k, S = 0].
\]

Whenever we obtain the estimated CATE parameters $\hat \theta$ by the proposed estimator, the plug-in estimator follows immediately:
\[
\hat \psi(x_k) = \frac{1}{|\mathcal{G}_k|} \sum_{i:i \in \mathcal{G}_k} g(X_i; \hat \theta),
\]
where $\mathcal{G}_k = \{i: S_i = 0, X_{i,k} = x_k\}$ and $|\mathcal{G}_k| = \sum_i I(i \in \mathcal{G}_k)$.

Following Proposition \ref{prop:asymptotic-linearity}, for a fixed target subgroup $H$ with $p_{0,H}=P_0(H)>0$, write $\psi_H=\mathbb E_{P_0}[g(\theta_0)(X)\mid H]$. The estimator has asymptotic linear expansion
\[
\hat\psi_H-\psi_H
=\mathbb P_{n_0}\!\left[\frac{I_H}{p_{0,H}}(g(\theta_0)-\psi_H)\right]
+\mathbb E_{P_0}[\partial_\theta g(\theta_0)(X)\mid H]
 (\hat\theta-\theta_0)+o_p(n_0^{-1/2}+r_n).
\]
Variance estimation and inference follow by combining these contributions as in Appendix \ref{app:variance-influence}.

\Needspace{6\baselineskip}
\subsection{Relative treatment effects}
\label{app:relative-effects}

A compatible marginal or subgroup risk ratio can contribute to identifying an additive CATE when its baseline risk is available or estimable. For a row $j$ belonging to trial $s$, suppress $s$ and write the population quantities as
\[
\tau_j^R=\frac{\mathbb E[Y(1)\mid I_j=1,S=s]}
                  {\mathbb E[Y(0)\mid I_j=1,S=s]},\qquad
\nu_j^0=\mathbb E[Y(0)\mid I_j=1,S=s]>0.
\]
With $\gamma_s^0(x)=\mathbb E[Y(0)\mid X=x,S=s]$, we have
\[
(\tau_j^R-1)\nu_j^0=(Tg_0)_j,\qquad
\nu_j^0=(T\gamma_s^0)_j.
\]
The ratio must compare risks for the same outcome horizon, population, and subgroup. Hazard ratios and conditional odds ratios do not satisfy this identity. When event counts or both arm risks are reported, a compatible risk difference can be used directly. A ratio and a reported control-risk estimate also give $\hat\tau_j=(\hat\tau_j^R-1)\hat\nu_j^0$.

\subsubsection{Baseline risks from aggregate arm means}
Use Appendix \ref{app:arm-means} to fit a control conditional mean from reported control-arm means and covariate summaries. A trial-specific fit $\hat\gamma_s^0$ avoids baseline transport assumptions when that trial provides sufficient information. A common fit $\hat\gamma^0$ can borrow information from other trials under control-mean transportability to the receiving trial. Denote the appropriate fitted function by $\hat\gamma_s^0$ in either case, and define
\[
\hat\nu_j^0=\mathbb Q_{n_q}[\hat\alpha_j\hat\gamma_s^0],\qquad
\hat M_j^R(\theta)
=\mathbb Q_{n_q}[\hat\alpha_jg(\theta)]
 -(\hat\tau_j^R-1)\hat\nu_j^0.
\]

For the modeled baseline risks, we retain A1 and the main CATE identification assumptions for the selected additive and relative rows. Require the following additional baseline assumptions:
\begin{enumerate}
\item[RR1.] The conditional control group model is identified and consistently estimable from the chosen arm-mean summaries under the corresponding conditions of Appendix \ref{app:arm-means}. If a shared function is used, its transportability must hold for each receiving trial; it need not extend to the target unless a target risk ratio is required.
\item[RR2.] $(\tilde T\gamma_s^0)_j=(T\gamma_s^0)_j=\nu_j^0$ for the relative-effect rows with $\tau_j^R\ne1$.
\end{enumerate}
Together with A6 at the truth, these imply
\[
M_j^R(\theta_0)
=\tilde T_j\{g(\theta_0)-(\tau_j^R-1)\gamma_s^0\}=0.
\]
Identification of $\theta$ uses the rank and global separation of the combined moment map after accounting for baseline parameters. We note that an additive and a relative effect computed from the same risks are redundant population information.

A two-stage estimation first fits the control arm mean model and then fits the CATE GMM. Alternatively, stack the control-mean model equations and the relative-effect equations and estimate the control mean parameters and CATE parameters jointly by GMM. See GMM for control arm mean model in Appendix \ref{app:arm-means}. The asymptotic linearity follows the proposition \ref{prop:asymptotic-linearity} and Appendix \ref{app:arm-means} with relevant Jacobians. 

\subsubsection{Baseline risks from external control outcomes}
An external population with observed $(X,Y)$ can instead inform the baseline model. If the external control outcomes come from the target population, we assume uniform use of control (i.e., everyone receiving control; see \citet{dahabreh2024learning}), consistency, and transportability of control group conditional mean function $\gamma^0(x)=\mathbb E[Y(0)\mid X=x,S=s] = \mathbb E[Y(0)\mid X=x,S=0] = \mathbb E[Y\mid X=x,S=0]$. Require baseline moment preservation as in RR2. Suppose $Q\ll P_0$ and $r=dQ/dP_0$, we have
\[
\nu_j^{0,\mathrm{ort}}(\tilde\gamma^0,\tilde r)
=\mathbb E_Q[\alpha_j\tilde\gamma^0]
 +\mathbb E_{P_0}[\tilde r\alpha_j(Y-\tilde\gamma^0)],
\]
which admits
\[
\nu_j^{0,\mathrm{ort}}(\tilde\gamma^0,\tilde r) - \tilde T_j\gamma^0 = \mathbb E_{P_0}[(\tilde r-r)\alpha_j(\gamma^0-\tilde\gamma^0)].
\]
The augmented baseline risk functional is double robust with respect to its nuisances $(\gamma^0, r)$. Use its empirical version in place of $\hat\nu_j^0$ above. For root-$n$ inference, we require appropriate source-size ratios, finite influence variances, empirical-process control or suitable cross-fitting, and
\[
\|\hat r-r\|_{L^2(P_0)}
\|\alpha_j(\hat\gamma^0-\gamma^0)\|_{L^2(P_0)}=o_p(n^{-1/2}),
\]
with the analogous controls for estimated representers. At the correct nuisances, the base-source influence is
$\alpha_j\{\gamma^0-\tilde T_j\gamma^0\}$ and the external-outcome influence is $r\alpha_j(Y-\gamma^0)$; the others follow the Proposition \ref{prop:asymptotic-linearity}.

If $Q=P_0$ and the same observations are used for both empirical averages, then $\hat r=1$. Then
\[
\hat\nu_j^{0,\mathrm{ort}}=\mathbb P_{n_0}[\hat\alpha_jY],\qquad
\hat M_j^R(\theta)=\mathbb P_{n_0}
 [\hat\alpha_j\{g(\theta)-(\hat\tau_j^R-1)Y\}],
\]
that is, the fitted baseline regression cancels exactly. If control outcome is observed only for a subpopulation, (partial) exchangeability for $Y(0)$ and positive control probability are additionally needed; replace the residual term by $I(A=0)(Y-\gamma^0)/\Pr(A=0\mid X,S=0)$ and propagate propensity estimation uncertainty. Related external-population conditions are discussed by \citet{dahabreh2024learning}.

\subsubsection{Target risk ratios}
A target baseline risk $\nu_0^0>0$ is needed in addition to the transported additive treatment effect:
\[
\psi^R=1+\psi/\nu_0^0,\qquad
\hat\psi^R-\psi^R
=(\nu_0^0)^{-1}(\hat\psi-\psi)
-\psi(\nu_0^0)^{-2}(\hat\nu_0^0-\nu_0^0)+o_p(n^{-1/2}),
\]
where $\hat \nu_0^0$ is estimated target baseline risk, obtained by either marginalizing an estimated control group model or averaging the external control outcomes.

\Needspace{6\baselineskip}
\subsection{Time-indexed outcomes and follow-up}
\label{app:followup}

For each clinically meaningful horizon $t\in\mathcal T$, let $Y_t(a)$ denote the potential outcome by, or measured at, time $t$ under baseline treatment assignment $a$. For a binary event outcome one may take $Y_t(a)=I\{T(a)\le t\}$ with a common definition of the event and handling of competing events. Define
\[
g_s(x,t)=\mathbb E[Y_t(1)-Y_t(0)\mid X=x,S=s],
\qquad
\psi(t)=\mathbb E[Y_t(1)-Y_t(0)\mid S=0].
\]
Here $X$ consists of baseline covariates as effect modifiers; no post-treatment covariates are involved. 

Let $\mathcal T_s\subseteq\mathcal T$ be the finite set of horizons reported by trial $s$, and let $\mathcal S_t=\{s:t\in\mathcal T_s\}$. A time-indexed version of A1 requires $g_s(x,t)=g(x,t)$ for $s\in\{0\}\cup\mathcal S_t$ on their respective supports. For identification without extrapolation at each $t$, require the target support to be covered by the union of trials in $\mathcal S_t$. Then $\psi(t)=\mathbb E_{P_0}[g(X,t)]$. These assumptions are at the same time horizon; they do not require treatment effects at different times to coincide.

For every reported trial, subgroup, and horizon, the population and sample equations match those in the main text:
\[
\tau_{s,j,t}=\mathbb E_{P_s}[g(X,t)\mid I_{s,j}=1],
\qquad
\hat M_{s,j,t}(\theta)=
 \mathbb Q_{n_q}[\hat\alpha_{s,j}(\hat\eta_s)g(X,t;\theta)]-\hat\tau_{s,j,t}.
\]
A6 becomes time-indexed moment invariance. The same covariate moment matching is reused across horizons within a trial. C1 becomes joint asymptotic linearity of effect estimates at different horizons, with their cross-time covariances.

For distinct reported horizons $t_1,\ldots,t_L$, a fully stratified CATE model is
\[
g(x,t;\theta)=\sum_{\ell=1}^L I(t=t_\ell)\{\gamma_\ell+f(x;\beta_\ell)\}.
\]
For example, $f(x;\beta_\ell)=\beta_\ell^\top b(x)$, with a separate slope vector for every horizon. A partially pooled alternative is
\[
g(x,t;\theta)=\sum_{\ell=1}^L I(t=t_\ell)\gamma_\ell+f(x;\beta).
\]
The latter assumes time-invariant effect modification on the additive scale while allowing the reference effect to vary. Smooth models in $t$ impose other explicit restrictions and result in an interpolation between reported time horizons. The fully stratified model does not identify effects at unreported times.

With full stratification, the aggregation Jacobian has separate parameter blocks for the distinct horizons. Each block must have adequate rank for CATE identification, as required by C2. Partial pooling can enhance identification through its additional assumptions.

The target estimator is $\hat\psi(t)=\mathbb P_{n_0}[g(X,t;\hat\theta)]$. Its asymptotic linear expansion is obtained by replacing $g(\theta_0)$ and $J_\theta^\psi$ in Proposition \ref{prop:asymptotic-linearity} by their values at $t$. Stacking these expansions supplies the covariance of the estimated effects across the different time horizons and their pointwise intervals.

Reported follow-up duration, such as a median observation time, does not by itself define risk at a fixed horizon. When censoring or missing outcomes occur, the trial summaries must consistently estimate the chosen time-specific quantities under appropriate censoring or missing-data assumptions and satisfy C1 jointly. Restricted mean survival-time differences can be handled analogously by taking $Y_t(a)=\min\{T(a),t\}$.

\Needspace{6\baselineskip}
\subsection{Moment invariance}
\label{app:local-invariance}

\subsubsection{A6 under exponential tilting and sufficient diagnosis}
Let $d\tilde P_s=w_s(\eta_{s,0})dQ$ be the tilted covariate law, with $\mathbb E_{\tilde P_s}h_s^+=\mathbb E_{P_s}h_s^+$. It agrees with $P_s$ on every function in $\mathcal H_s=\operatorname{span}(h_s^+)$. Exponential tilting selects the law of least KL divergence from $Q$ satisfying these equations \citep{cover2006maximum}. Since its density ratio depends only on $h_s(X)$,
\[
\tilde P_s\{X\in B\mid h_s(X)\}=Q\{X\in B\mid h_s(X)\}
\]
for measurable set $B$, almost surely where defined. Thus the tilting inherits conditional structure from $Q$ while changing the distribution of $h_s(X)$. A6 does not require equality of the two full covariate laws: it asserts that their differences do not change the averages needed by the candidate CATE functions.

For a linear CATE $g(x;\theta) = b(x)^\top\theta$, a subgroup basis mean is a ratio,
\[
\mathbb E_{P_s}[b_k\mid I_{s,j}=1]
=\frac{\mathbb E_{P_s}(I_{s,j}b_k)}{\mathbb E_{P_s}I_{s,j}}.
\]
The sufficient condition is to preserve both the numerator and the denominator of this ratio. Define the subgroup basis vector as
\[
b_s^{\mathrm{agg}}(x)
=\bigl(I_{s,j}(x),\ I_{s,j}(x)b_k(x):j=1,\ldots,J_s,\ k=1,\ldots,d\bigr)^\top,
\]
and remove the duplicates if desired. Use a testing measure $Q$ dominating the relevant trial laws, with all these functions and $h_s^+$ square-integrable. Define
\begin{align*}
G_s^h&=\mathbb E_Q[h_s^+h_s^{+\top}],&
G_s^{hb}&=\mathbb E_Q[h_s^+(b_s^{\mathrm{agg}})^\top],&
G_s^b&=\mathbb E_Q[b_s^{\mathrm{agg}}(b_s^{\mathrm{agg}})^\top],\\
\Pi_s&=(G_s^h)^+ G_s^{hb},&
r_s&=b_s^{\mathrm{agg}}-\Pi_s^\top h_s^+,&
R_s&=G_s^b-(G_s^{hb})^\top (G_s^h)^+ G_s^{hb}.
\end{align*}
Here $+$ denotes the Moore-Penrose inverse, allowing redundant moment functions; it reduces to an ordinary inverse after removing linear dependencies. The normal equations imply
\[
R_s=\mathbb E_Q[r_sr_s^\top]\succeq0,\qquad
(R_s)_{kk}=\inf_c\|b_{s,k}^{\mathrm{agg}}-c^\top h_s^+\|_{L^2(Q)}^2.
\]
Consequently,
\[
R_s=0
\quad\Longleftrightarrow\quad
b_{s,k}^{\mathrm{agg}}\in\mathcal H_s\ \text{in }L^2(Q)\ \text{for every }k.
\]
Since both $P_s$ and $\tilde P_s$ are dominated by $Q$, exact zero at the population level (with feasible exact moment matching and positive subgroup probabilities) implies preservation of every subgroup basis mean, hence A6 for the linear model. This is a sufficient check for A6 through a stronger span condition; it is not a test of the unobserved equality in A6 itself. A6 can hold when this span condition fails, because $Q$ can correctly supply the remaining subgroup basis means. In practice, the base law used for tilting is a natural choice of $Q$ for testing. 

Let $\Delta_s = dP_s/dQ-w_s(\eta_{s,0})$ be the tilting misspecification error with $\Delta_s \in L^2(Q)$. Then, we have
\[
(\mathbb E_{P_s}-\mathbb E_{\tilde P_s})b_{s,k}^{\mathrm{agg}}
=\mathbb E_Q[\Delta_s r_{s,k}],\qquad
\left|(\mathbb E_{P_s}-\mathbb E_{\tilde P_s})b_{s,k}^{\mathrm{agg}}\right|
\le C_s^\Delta (R_s)_{kk}^{1/2},
\]
where $C_s^\Delta = \|\Delta_s\|_{L^2(Q)}$ is fixed. The error in the sufficient condition of A6 is thus bounded in proportion to $(R_s)_{kk}^{1/2}$; however, a small absolute residual does not by itself imply a small error, because the constant $C_s^\Delta$ may be large. Nevertheless, this supplies a practical diagnosis that helps to choose $Q$. 

\textbf{Choosing the information supplied by $Q$.} The residual columns show exactly what remains to be supplied: their trial expectations must agree with their tilted expectations, or their discrepancies must cancel in the required ratios. Choose $Q$ to supply the large-residual columns such that they directly approximate the trials’ information. This requires domain knowledge about those remaining subgroup basis means. Without individual data, this is hard to verify in general. Nevertheless, it can be transparent for future verification. 

\subsubsection{Local invariance and identification in the working span}
Let $T$ and $\tilde T$ be the true and working aggregation operators defined after A6, and let
\[
\tilde{\mathcal A}
=\operatorname{span}\{\alpha_1(\cdot;\eta_0),\ldots,
                       \alpha_J(\cdot;\eta_0)\}\subseteq L^2(Q).
\]
Consider the following replacement for A5--A6.

\textit{A5w. Global identification in the working span}: for every $\theta\in\Theta\setminus\{\theta_0\}$,
\[
g(\theta)-g(\theta_0)\notin\tilde{\mathcal A}^{\perp},
\quad\text{equivalently}\quad
\tilde T(g(\theta)-g(\theta_0))\ne0.
\]

\textit{A6loc. Local moment invariance}: there exists an open neighborhood $\mathcal N$ of $\theta_0$, contained in the model's parameter domain, such that
\[
(T-\tilde T)g(\theta)=0
\qquad\text{for every }\theta\in\mathcal N\cap\Theta.
\]
The population tilting parameter $\eta_0$, and hence $\tilde T$, is fixed in this condition.

\begin{Proposition}[Identification under local invariance]
\label{prop:local-invariance}
Under A1, A3, A4, A5w, A6loc, and existence of the normalized working operator, $\theta_0$ is the unique solution in $\Theta$ of $\tilde Tg(\theta)=\tau_0$.
\end{Proposition}
\begin{proof}
A1 and A4 imply $Tg(\theta_0)=\tau_0$, and A6loc implies $\tilde Tg(\theta_0)=Tg(\theta_0)$. If another parameter solves the working equations, subtraction gives $\tilde T(g(\theta)-g(\theta_0))=0$, contradicting A5w.
\end{proof}

\textit{Relationship to the main assumptions.}
The main A5 states that the true representers distinguish the model from its true parameter, whereas A5w requires the working representers selected by the tilting to distinguish it. Under global A6 these identification statements coincide, because $T$ and $\tilde T$ agree on every candidate function and hence on their differences. Therefore A5 plus global A6 implies A5w plus A6loc. In the alternative formulation, preservation is required only near the truth, but global separation is asserted directly for the working span. Without global A6, true-span identification by itself need not establish global working-span identification. Thus the alternative changes the meaning of the identification assumption as well as the domain of the invariance assumption.

For the population identification, even A6loc is stronger than necessary: validity at the single point, $(T-\tilde T)g(\theta_0)=0$, suffices when paired with A5w. Local invariance additionally gives agreement of the derivatives at $\theta_0$. Proposition \ref{prop:asymptotic-linearity} otherwise needs only valid moments at the truth and a full-rank \emph{working} derivative; it does not require agreement of derivatives with the true operator. Although technically weaker, this statement does not necessarily carry practical meaning in applications, so we keep it in the Appendix only.

\Needspace{6\baselineskip}
\subsection{Linear CATE misspecification}
\label{app:misspecification}

This appendix studies a linear working CATE under misspecification. We retain A1--A3 but not A4, and assume the linear CATE by a general basis expression as
\[
g(x;\theta)=b(x)^\top\theta,\qquad \theta\in\mathbb R^d,
\]
with the fixed basis $b(x) = (b_1(x), ..., b_d(x))^\top$. The basis $b$ is prespecified and its dimension is fixed. We impose the positive-definiteness and full-rank conditions below so that the pseudo-true limits are uniquely defined. Below, we use an individual-data-based linear CATE estimator as the linear-approximation benchmark and give conditions under which the pseudo-true limits of the IPD benchmark and the aggregate-data estimator in the main text agree. Then, we show the asymptotic linearity and inference around the pseudo-true limits. 

\subsubsection{Pseudo-true limits}
\paragraph{An IPD-based linear CATE benchmark.}
Let $N=\sum_{s=1}^M n_s$, and suppose $n_s/N\to\lambda_s>0$, with $\sum_s\lambda_s=1$. For the pooled estimator below, assume $n_s/N-\lambda_s=o(N^{-1/2})$. Define the pooled trial covariate law
\[
Q_I=\sum_{s=1}^M\lambda_sP_s(X).
\]
The pooled trial law $Q_I$ need not coincide with the base law $Q$ used in aggregate data estimation. The hypothetical IPD estimator regresses the known-propensity pseudo-outcome
\[
U_s(O)=\frac{AY}{e_s(X)}-\frac{(1-A)Y}{1-e_s(X)},
\qquad \mathbb E_{P_s}[U_s\mid X]=g_0(X),
\]
on $b(X)$, using the pooled estimating equation
\[
\frac1N\sum_{s=1}^M\sum_{i=1}^{n_s}
b(X_{i,s})\{U_s(O_{i,s})-b(X_{i,s})^\top\hat\theta_I\}=0.
\]
Here $e_s(X)$ is the known randomized treatment probability. Assume positivity, $(g_0,b_k)\in L^2(Q_I)$ for all $k$, the relevant LLN, and positive-definite $G=\mathbb E_{Q_I}[bb^\top]$. The pseudo-true limit is
\[
\theta_I^*=G^{-1}\mathbb E_{Q_I}[bg_0]
=\arg\min_\theta\mathbb E_{Q_I}[(g_0-b^\top\theta)^2].
\]
Alternatively, fixed positive trial weights can be used in the estimating equation $\sum_s\lambda_s\mathbb P_{n_s}[b(U_s-b^\top\theta)]=0$; this removes the rate requirement on $n_s/N-\lambda_s$.

\paragraph{Aggregate-based linear CATE.}
Let $B$ and $\tilde B$ be the $J\times d$ matrices with respective rows
\[
B_{s,j}=\mathbb E_{P_s}[b(X)^\top\mid I_{s,j}=1],\qquad
\tilde B_{s,j}=\mathbb E_{\tilde P_s}[b(X)^\top\mid I_{s,j}=1],
\]
where $\tilde P_s$ is obtained by tilting the reference law $Q$ used in the analysis. Suppose $\tilde B$ has full column rank and $\hat W\to_pW$ with a positive definite $W$. Define
\[
\tilde K=(\tilde B^\top W\tilde B)^{-1}\tilde B^\top W,
\qquad
\theta_A^*=\tilde K\tau_0
=\arg\min_\theta\|\tilde B\theta-Tg_0\|_W^2,
\]
where $\|v\|_W^2=v^\top Wv$. Unlike under correct specification, the pseudo-true limit generally depends on $W$.

Let $c=\mathbb E_{P_0}b$ be finite; the target ATEs are then
\[
\psi_I^*=c^\top\theta_I^*,\qquad
\psi_A^*=c^\top\theta_A^*,\qquad
\psi_0=\mathbb E_{P_0}g_0.
\]
Under LLN for relevant empirical quantities, we have
\[
\hat\theta_I\to_p\theta_I^*,\quad \hat\psi_I\to_p\psi_I^*,
\qquad
\hat\theta\to_p\theta_A^*,\quad \hat\psi\to_p\psi_A^*.
\]
Aggregate GMM minimizes a finite collection of weighted average residuals, whereas IPD regression minimizes an individual residual's $L^2(Q_I)$ norm.

\paragraph{Conditions for agreement under misspecification.}
For this comparison, we retain A6 on the linear working model, so that $\tilde B=B$, and define
\[
K=(B^\top WB)^{-1}B^\top W,\qquad
r_I=g_0-b^\top\theta_I^*,\qquad \mathbb E_{Q_I}[br_I]=0.
\]
For $p_{s,j}=\mathbb E_{P_s}I_{s,j}>0$, the true aggregation representers relative to the benchmark law are
\[
\alpha_{I,s,j,0}=\frac{dP_s(X\mid I_{s,j}=1)}{dQ_I},\qquad
\alpha_{I,0}=(\alpha_{I,s,j,0})_{s,j}.
\]
These representers are bounded because $Q_I$ contains each trial law with positive weight. In this Hilbert space, $Tf=\mathbb E_{Q_I}[\alpha_{I,0}f]$ and $B=\mathbb E_{Q_I}[\alpha_{I,0}b^\top]$. Define
\[
\alpha_\perp=\alpha_{I,0}-BG^{-1}b,\qquad
q_\perp=c^\top K\alpha_\perp,\qquad
S_\perp=\mathbb E_{Q_I}[\alpha_\perp\alpha_\perp^\top].
\]
Here, $\alpha_\perp$ is the residual from projection of the true aggregation representers onto the basis span in $L^2(Q_I)$.

\begin{Proposition}[IPD-AgD agreement under linear CATE misspecification]
\label{prop:misspecification-bound}
Under the preceding conditions, for a given $g_0$, the limits $\theta_I^*$ and $\theta_A^*$ agree if and only if $KTr_I=0$. They agree for every $g_0\in L^2(Q_I)$ if and only if
\begin{equation}
K\alpha_{I,0}(X)=G^{-1}b(X)\quad Q_I\text{-almost surely}.
\label{eq:misspecification-equivalence}
\end{equation}
For a given $g_0$, the target limits agree if and only if $\mathbb E_{Q_I}[q_\perp r_I]=0$; they agree for every $g_0\in L^2(Q_I)$ if and only if $q_\perp=0$ $Q_I$-almost surely.
More generally, the discrepancies satisfy
\begin{align}
\theta_A^*-\theta_I^*&=KTr_I=K\mathbb E_{Q_I}[\alpha_\perp r_I],
\label{eq:misspecification-parameter-gap}\\
\|b^\top(\theta_A^*-\theta_I^*)\|_{L^2(Q_I)}
&\le\|G^{1/2}KS_\perp^{1/2}\|_{\mathrm{op}}\|r_I\|_{L^2(Q_I)},
\label{eq:misspecification-function-bound}\\
\psi_A^*-\psi_I^*&=\mathbb E_{Q_I}[q_\perp r_I],
\label{eq:misspecification-target-gap}\\
|\psi_A^*-\psi_I^*|&\le\|q_\perp\|_{L^2(Q_I)}\|r_I\|_{L^2(Q_I)}.
\label{eq:misspecification-target-bound}
\end{align}
\end{Proposition}
\begin{proof}
Substituting $g_0=b^\top\theta_I^*+r_I$ into $\theta_A^*=KTg_0$, and using $KB=I_d$ and $\mathbb E_{Q_I}[br_I]=0$ gives the parameter discrepancy; multiplication by $c^\top$ yields the target ATE discrepancy. Equality of the parameter maps for every $g_0\in L^2(Q_I)$ is equivalent to equality of their representers. For the target ATE, the difference representer is $q_\perp$. Cauchy--Schwarz gives the scalar bound. For the function bound, the operator mapping $f$ to $G^{1/2}K\mathbb E_{Q_I}[\alpha_\perp f]$ has squared norm equal to the largest eigenvalue of $G^{1/2}KS_\perp K^\top G^{1/2}$, giving the bound with operator norm. 
\end{proof}

A stronger sufficient condition for agreement is $\alpha_{I,0}=BG^{-1}b$: every true aggregation representer lies in the basis span. More generally, its residual component can cancel after multiplication by $K$. When exact agreement is implausible, we have the error bound
\[
|\psi_A^*-\psi_0|
\le |\psi_I^*-\psi_0|
 +\|q_\perp\|_{L^2(Q_I)}\|r_I\|_{L^2(Q_I)}.
\]
This bounds the additional discrepancy relative to the IPD benchmark.

We assume A6 above. If A6 on the basis fails, then
\[
\theta_A^*-\theta_I^*
=\tilde K\{Tr_I+(B-\tilde B)\theta_I^*\}.
\]
The second term represents error in the aggregation of the basis functions. The following inference concerns the actual working limit $\theta_A^*$ and does not require A6 with $\tilde B=B$.

\subsubsection{Inference around the pseudo-true limits}
\paragraph{IPD estimator.}
For the IPD estimator, define $v_s(O)=b(X)\{U_s(O)-b(X)^\top\theta_I^*\}$. Under finite second moments of $v_s$, LLN for the Gram matrix, and the sampling conditions above, we have
\[
\hat\theta_I-\theta_I^*
=G^{-1}\sum_{s=1}^M\lambda_s(\mathbb P_{n_s}-P_s)v_s
+o_p\left(\sum_s n_s^{-1/2}\right).
\]
A consistent variance estimator is
\[
\hat G^{-1}\left\{\sum_{s=1}^M\frac{\lambda_s^2}{n_s}
\hat{\operatorname{Var}}_{P_s}[b\{U_s-b^\top\hat\theta_I\}]
\right\}\hat G^{-1},\qquad
\hat G=\sum_s\lambda_s\mathbb P_{n_s}[bb^\top].
\]
This variance estimator is robust to misspecification of the linear regression.

\paragraph{Aggregate estimator.}
For aggregate linear GMM, let $\hat B$ have rows $\mathbb Q_{n_q}[\hat\alpha_{s,j}(\hat\eta_s)b^\top]$. Let $\hat W$ be symmetric positive definite with probability tending to one, and suppose $\hat W\to_pW$ for a symmetric positive-definite $W$. The full-rank normal equations give
\[
\hat\theta=(\hat B^\top\hat W\hat B)^{-1}\hat B^\top\hat W\hat\tau,
\qquad \hat B^\top\hat W(\hat B\hat\theta-\hat\tau)=0.
\]
Define
\[
r_A^*=\tilde B\theta_A^*-\tau_0,\qquad
H_A=\tilde B^\top W\tilde B,\qquad \tilde B^\top Wr_A^*=0,
\]
and $\Delta B=\hat B-\tilde B$, $\Delta\tau=\hat\tau-\tau_0$, and $\Delta W=\hat W-W$. Under C1, the relevant conditions in C3, finite second moments of the relevant basis influences, and an interior pseudo-true parameter, we have $\Delta B=O_p(r_n)$ and $\Delta\tau=O_p(r_n)$. This argument does not require A4 or A6; full column rank of $\tilde B$ identifies the pseudo-true limit.

Suppose $\Delta W=O_p(a_n)$ for a deterministic $a_n\to0$. Expanding the normal equations gives
\begin{align*}
\hat\theta-\theta_A^*
=-H_A^{-1}\bigl\{&\tilde B^\top W(\Delta B\theta_A^*-\Delta\tau)
 +(\Delta B)^\top Wr_A^* \\
&+\tilde B^\top\Delta W r_A^*\bigr\}
 +o_p(r_n+a_n).
\end{align*}
Consistency of $\hat W$ alone does not guarantee root-$n$ inference when $r_A^*\ne0$, because $\tilde B^\top\Delta W r_A^*$ can have a slower rate. If $r_A^*=0$, both residual-dependent terms vanish and the normal equations give the main first-order expansion with only consistency of $\hat W$. Correct CATE specification A4 together with moment invariance A6 implies this zero residual. A fixed weighting matrix removes the $\Delta W$ term, but generally leaves the $(\Delta B)^\top Wr_A^*$ term under misspecification.

\textit{C4. Joint asymptotic linearity of the estimated weight.}
We assume the weights are estimated from trials and admit the asymptotic linear expansion as
\[
\Delta W=\sum_{s=1}^M\mathbb P_{n_s}[\Phi_s^W]+o_p\left(\sum_{s=1}^M n_s^{-1/2}\right),
\]
where each $\Phi_s^W$ is a symmetric $J\times J$ matrix with mean zero and finite second moments under $P_s$. Its joint behavior with the C1 influences is retained, including any preliminary estimates used to construct the weight. With the sample-size ratios of Corollary~\ref{cor:asymptotic-normality}, this implies $\Delta W=O_p(n^{-1/2})$ and $r_n=O(n^{-1/2})$. A common example of estimated weights is the inverse marginal variances from the trials. For a general two-step GMM weight, we can add the relevant $\mathbb Q_{n_q}[\Phi_q^W]$ and $\mathbb P_{n_0}[\Phi_0^W]$ terms with the relevant remainders. 

Following Proposition \ref{prop:asymptotic-linearity}, we write $\Delta B$ with its asymptotic linear expansion as
\[
\Delta B=\mathbb Q_{n_q}[\Phi_q^B]
+\sum_{s=1}^M\mathbb P_{n_s}[\Phi_s^B]+o_p(r_n),
\]
where $\Phi_q^B$ and $\Phi_s^B$ are $J\times d$ matrices with mean zero and finite second moments under their respective sources. Let $\mathcal S_A=\{q,1,\ldots,M\}$, $P_q=Q$, $\mathbb P_{n_q}=\mathbb Q_{n_q}$, $\Phi_s^\tau=E_s\phi_s^\tau$, and $\Phi_q^\tau=0$. Under C4, we have the asymptotic linearity as
\begin{align*}
\ell_a=-H_A^{-1}\bigl\{&\tilde B^\top W
(\Phi_a^B\theta_A^*-\Phi_a^\tau)
+(\Phi_a^B)^\top Wr_A^*\notag\\
&+\tilde B^\top\Phi_a^W r_A^*\bigr\},
\qquad a\in\mathcal S_A,\\
\hat\theta-\theta_A^*
&=\mathbb Q_{n_q}[\ell_q]
+\sum_{s=1}^M\mathbb P_{n_s}[\ell_s]+o_p(r_n).
\end{align*}
Each influence $\ell_a$ has mean zero under the corresponding source law. 

\paragraph{Target ATE.}
For $L\in\{A,I\}$, both target plug-in estimators satisfy
\[
\hat\psi_L-\psi_L^*
 =(\mathbb P_{n_0}-P_0)[b^\top\theta_L^*]
   +c^\top(\hat\theta_L-\theta_L^*)
   +o_p(n_0^{-1/2}+r_n).
\]
A joint central limit theorem with the consistent variance estimator gives Wald intervals with asymptotically correct coverage of $\theta_L^*$ and $\psi_L^*$ for both $L\in\{A,I\}$. Under Proposition \ref{prop:misspecification-bound}'s agreement condition, the two estimators can give asymptotically valid intervals for the same approximation. Their interval widths need not agree because their influence functions and available information differ.

\Needspace{6\baselineskip}
\subsection{Variance estimation}
\label{app:variance}

\subsubsection{Variance estimator with complete information}
\label{app:variance-influence}

This subsection gives the variance estimator under A1--A6, C1--C3, C1s, and the plug-in covariance regularity from main text. It follows Proposition~\ref{prop:asymptotic-linearity} and Corollary~\ref{cor:asymptotic-normality}. Appendix~\ref{app:variance-misspecification} gives the variance for the pseudo-true limit under a linear working model, which does not require A4--A6.

Let $\Sigma_s^+$ be the C1 joint covariance of $(\phi_s^{\tau\top},\phi_s^{\mu+\top})^\top$, and define
\[
D_s=(J_{\tau_s}^m,J_{\mu_s}^m),\qquad
\Gamma_s=\mathbb E_{P_s}[\xi_s\xi_s^\top]=D_s\Sigma_s^+D_s^\top.
\]
This includes the cross-covariance $J_{\tau_s}^m\Sigma_s^{\tau\mu+}(J_{\mu_s}^m)^\top$ and its transpose.

Write $\hat m_{s,j}=\hat m_{s,j}(\hat\theta,\hat\eta_s)$ and $\hat\alpha_{s,j}=\hat\alpha_{s,j}(\hat\eta_s)$. Estimate the derivatives by
\begin{align*}
(\hat J_\theta^m)_{s,j,\cdot}
 &=\mathbb Q_{n_q}[\hat\alpha_{s,j}\partial_\theta g(\hat\theta)],\\
(\hat J_{\eta_s}^m)_{s,j,\cdot}
 &=\mathbb Q_{n_q}[\hat\alpha_{s,j}\{g(\hat\theta)-\hat m_{s,j}\}h_s^{+\top}],\\
\hat H_s&=\mathbb Q_{n_q}[w_s(\hat\eta_s)h_s^+h_s^{+\top}],
\qquad \hat J_{\mu_s}^m=\hat J_{\eta_s}^m\hat H_s^{-1}.
\end{align*}
Define
\begin{align*}
(\hat\phi_{q,i}^m)_{s,j}
 &=\hat\alpha_{s,j}(X_i)\{g(X_i;\hat\theta)-\hat m_{s,j}\},\\
\hat\xi_{q,i}&=\hat\phi_{q,i}^m-\sum_s\hat J_{\mu_s}^m
 \{w_s(X_i;\hat\eta_s)h_s^+(X_i)-\hat\mu_s^+\},\\
\hat\Gamma_q&=\mathbb Q_{n_q}[(\hat\xi_q-\bar{\hat\xi}_q)
                      (\hat\xi_q-\bar{\hat\xi}_q)^\top],
\qquad \hat\Gamma_s=\hat D_s\hat\Sigma_s^+\hat D_s^\top.
\end{align*}
The variance estimators are
\begin{align*}
\hat\Omega&=(n/n_q)\hat\Gamma_q+\sum_s(n/n_s)\hat\Gamma_s,\\
\hat J_m^\theta&=-\{(\hat J_\theta^m)^\top\hat W\hat J_\theta^m\}^{-1}
                       (\hat J_\theta^m)^\top\hat W,\\
\hat{\mathrm{Var}}(\hat\theta)&=n^{-1}\hat J_m^\theta\hat\Omega(\hat J_m^\theta)^\top,\\
\hat{\mathrm{Var}}(\hat\psi)&=
 n_0^{-1}\mathbb P_{n_0}[\{g(\hat\theta)-\hat\psi\}^2]
 +\hat J_\theta^\psi\hat{\mathrm{Var}}(\hat\theta)(\hat J_\theta^\psi)^\top,
\end{align*}
where $\hat J_\theta^\psi=\mathbb P_{n_0}[\partial_\theta g(\hat\theta)]$.

\paragraph{Target/Q reuse and estimated Q.}
If $Q=P_0(X)$ and the same target observations form $\mathbb Q_{n_q}=\mathbb P_{n_0}$, combine the two target contributions:
\[
\zeta_0(X)=g(X;\theta_0)-\psi_0+J_\theta^\psi J_m^\theta\xi_q(X),
\qquad
\zeta_s(O)=J_\theta^\psi J_m^\theta\xi_s(O).
\]
Then the asymptotic variance is given by
\[
n_0^{-1}\mathbb E_{P_0}[\zeta_0^2]+\sum_s n_s^{-1}\mathbb E_{P_s}[\zeta_s^2].
\]
If $Q=Q_\gamma$ is estimated, derive the relevant Jacobians with respect to $\gamma$ and add them to the influence function propagation. An equivalent parametric bootstrap refits $\gamma$, followed by the re-computations in each replicate; a large Monte Carlo size from the estimated model makes the numerical error negligible.

\subsubsection{Misspecification-robust variance estimator with complete information}
\label{app:variance-misspecification}

\textit{C4s. Consistent joint influence covariance estimation.}
Similar to C1s, we require a consistent estimator of the joint influence covariance in C4. For the inverse marginal variance weight, we have $\hat W=\operatorname{diag}_{s,j}(\hat v_{s,j}^{-1})$ and $\hat v_{s,j}=n_s\widehat{\operatorname{se}}(\hat\tau_{s,j})^2$, with $v_{s,j}>0$. Suppose the variance estimates admit the joint expansion
\[
\hat v_s-v_s=\mathbb P_{n_s}[\phi_s^v]+o_p(n_s^{-1/2}),
\]
then the GMM weight influence is given by
\[
\Phi_s^W
=-E_s\operatorname{diag}_{j}
\left(\frac{\phi_{s,j}^v}{v_{s,j}^2}\right)E_s^\top.
\]

For $a\in\mathcal S_A=\{q,1,\ldots,M\}$, let
\begin{align*}
\hat\ell_a=-\hat H_A^{-1}\bigl\{&\hat B^\top\hat W
(\hat\Phi_a^B\hat\theta-\hat\Phi_a^\tau)
+(\hat\Phi_a^B)^\top\hat W\hat r\\
&+\hat B^\top\hat\Phi_a^W\hat r\bigr\}
\end{align*}
be the empirical influence vector, where $\hat r = \hat B \hat \theta - \hat \tau$ and $\hat H_A = \hat B^\top \hat W \hat B$. Let $\hat\Gamma_a$ consistently estimate $\operatorname{Var}_{P_a}(\ell_a)$ as
\[
\hat\Gamma_a
=\mathbb P_{n_a}[(\hat\ell_a-\bar{\hat\ell}_a)
(\hat\ell_a-\bar{\hat\ell}_a)^\top],\qquad
\bar{\hat\ell}_a=\mathbb P_{n_a}\hat\ell_a.
\]
The variance estimator is given by
\[
\hat{\operatorname{Var}}_{\mathrm{MR}}(\hat\theta)
=n_q^{-1} \hat\Gamma_q
+\sum_{s=1}^M n_s^{-1} \hat\Gamma_s.
\]

\subsubsection{Covariance reconstruction when reports are incomplete}
\label{app:covariance-reconstruction}
Reported standard errors supply marginal sampling variances. We reconstruct missing covariance by evaluating the joint influence functions under the working trial law, then retain its correlations while inserting the reported variances. This provides a simple positive-semidefinite approximation using the conditional arm-mean framework in Appendix \ref{app:arm-means}.

\paragraph{Basic approximation structure.}
Suppose a canonical working trial structure with i.i.d. observations, complete outcomes, fixed randomized allocation probability $e_s\in(0,1)$, covariate moments, and unadjusted differences of arm means. Write $e_1=e_s$ and $e_0=1-e_s$ for convenience, and let
\[
\gamma_s^a(x)=\mathbb E[Y(a)\mid X=x,S=s],\quad
\nu_j^a=\mathbb E_{P_s}[\gamma_s^a(X)\mid I_j=1],\quad
p_j=\mathbb E_{P_s}I_j.
\]
The arm-mean and treatment effect influence functions are
\[
\phi_j^{\nu^a}(O)=\frac{I_j(X)}{p_j}\frac{I(A=a)}{e_a}(Y-\nu_j^a),\qquad
\phi_j^\tau=\phi_j^{\nu^1}-\phi_j^{\nu^0},\qquad
\phi^{\mu+}=h^+-\mu^+.
\]
Consequently,
\begin{align}
\Sigma_{jk}^{\nu^a\nu^a}
&=\mathbb E_{P_s}\!\left[\frac{I_jI_k}{p_jp_ke_a}
\{v_a+(\gamma_s^a-\nu_j^a)(\gamma_s^a-\nu_k^a)\}\right],
\label{eq:cov-reconstruction-exact}\\
\Sigma^{\nu^1\nu^0}&=0,\qquad
\Sigma_{j,\cdot}^{\nu^a\mu+}
=\operatorname{Cov}_{P_s}(\gamma_s^a,h^+\mid I_j=1),\qquad
\Sigma^{\mu+}=\operatorname{Cov}_{P_s}(h^+),
\label{eq:cov-reconstruction-cross}
\end{align}
where $v_a(x)=\operatorname{Var}(Y(a)\mid X=x,S=s)$. Thus, we have
\[
\Sigma^{\tau\tau}=\Sigma^{\nu^1\nu^1}+\Sigma^{\nu^0\nu^0},
\]
and
\[
\Sigma_{j,\cdot}^{\tau\mu+}=\operatorname{Cov}_{P_s}(g_0,h^+\mid I_j=1).
\]

\paragraph{Main case: binary outcomes, arm risks, and marginal standard errors.}
Suppose the available reports contain marginal/subgroup risk differences and their standard errors, both arms' risks and their standard errors, and covariate summaries, but no covariances, outcome-covariate cross-moments, or subgroup-overlap summaries. This captures the common reporting case for binary outcomes; the
same method also works for continuous outcomes with reported arm means and their standard errors. The following three steps use only the summaries needed for the fitted estimator.

\textit{1. Fit a working control-risk regression.}
Use the tilted law $\tilde P_s$ from the main analysis, with reported subgroup proportions included among its moments when available. Let $\tilde{\mathbb E}_s$ denote integration under its empirical weighted version and set $\hat g(x)=g(x;\hat\theta)$. Fit a low-dimensional trial-specific working mean $\gamma_s^0(x;\beta_s^0)=b(x)^\top\beta_s^0$ to the reported control risks by
\begin{equation}
\hat\beta_s^0=\arg\min_{\beta}
\sum_{j\in\mathcal J_s^0}
\frac{\{\tilde{\mathbb E}_s[b(X)^\top\beta\mid I_j=1]-\hat\nu_j^0\}^2}
     {\operatorname{se}(\hat\nu_j^0)^2},
\quad
\max(0,-\hat g(x))\le b(x)^\top\beta\le\min(1,1-\hat g(x)).
\label{eq:cov-control-fit}
\end{equation}
The bounds make the working control mean compatible with the estimated CATE in main text. The working control mean need not be correctly specified; choose smaller basis or use minimum-norm solution for working identification. This is the control arm outcome mean estimation with the identity link function in Appendix \ref{app:arm-means}. The minimization uses inverse marginal variances; fully efficient estimation is not needed here. Define
\[
\tilde\gamma_s^0(x)=b(x)^\top\hat\beta_s^0,\qquad
\tilde\gamma_s^1(x)=\tilde\gamma_s^0(x)+\hat g(x),\qquad
\tilde v_a(x)=\tilde\gamma_s^a(x)\{1-\tilde\gamma_s^a(x)\}.
\]
No conditional variance model is needed for binary outcomes.

\textit{2. Evaluate one joint working covariance.}
Calculate all products in \eqref{eq:cov-reconstruction-exact}-\eqref{eq:cov-reconstruction-cross} under $\tilde P_s$, with the fitted
\[
\tilde p_j=\tilde{\mathbb E}_s I_j,\qquad
\tilde\nu_j^a=\tilde{\mathbb E}_s[\tilde\gamma_s^a\mid I_j=1].
\]
This gives a positive-semidefinite matrix $\tilde\Sigma_s$. The same covariate law supplies the missing overlap probabilities $\tilde{\mathbb E}_s(I_jI_k)$ and $\operatorname{Cov}_{\tilde P_s}(h^+)$.

\textit{3. Anchor the reported marginal variances.}
Let $\hat{\mathrm{se}}_i$ be the reported marginal standard error for each coordinate $i$, and let
\[
v_i = n_s\hat{\mathrm{se}}_i^2, \qquad R_{ii}=\sqrt{v_i/\tilde\Sigma_{s,ii}},
\qquad \hat\Sigma_s=R\tilde\Sigma_sR.
\]
For a missing variance $v_i$, take $\tilde\Sigma_{s,ii}$; for zero variance, use any convention that avoids a zero denominator.

This construction requires no additional transportability of the control mean function across trials when each $\gamma_s^0$ is fitted separately. A common control mean model can instead pool information under the stronger assumptions in Appendix \ref{app:arm-means}, when the trial's sample size is small, or arm risks/means are insufficiently reported in some trial.  

\paragraph{Other reporting patterns.}
Available overlap summaries or covariate cross-moments can be added to moment matching. With binary effects and standard errors but no arm means, baseline risks require a further risk model; with continuous outcomes, conditional residual variances require additional modeling. For adjusted effects, incomplete outcomes, or covariate-dependent randomization design, replace the raw influence functions by those of the reported estimators, retaining the same working-law and variance-anchoring principle.

\Needspace{6\baselineskip}
\subsection{Pooling trial-specific transported effects}
\label{app:pooling}

\paragraph{Identification and estimation.}
Write $\psi_0=\mathbb E_{P_0}[g_0(X)]$ for the target ATE, where $g_0$ denotes the CATE function in the target population. Replace A1 and A2 by
\[
\begin{aligned}
\textit{A1p.}\quad &\mathbb E_{P_0}[g_s(X)]=\psi_0,
&&s=1,\ldots,M,\\
\textit{A2p.}\quad &\mathcal X_0\subseteq\mathcal X_s,
&&s=1,\ldots,M.
\end{aligned}
\]
A1p requires the trial-specific CATE functions to yield the same ATE after marginalization to the target covariate distribution, while allowing them to differ across trials. Marginalizing all trials to the same covariate distribution does not by itself imply A1p. Retain A3 and impose A4--A6 separately for each trial, replacing $g$, $\theta_0$, and $\mathcal A$ by $g_s$, $\theta_{s,0}$, and $\mathcal A_s=\mathrm{span}\{\alpha_{s,j,0}:j=1,\ldots,J_s\}$, respectively. Thus $g_s(x)=g(x;\theta_{s,0})$, Proposition \ref{prop:cate-identification} applied within trial $s$ identifies $\theta_{s,0}$, and A1p and A2p give $\psi_0=\mathbb E_{P_0}[g(X;\theta_{s,0})]$ for every $s$. 

Apply the moment-matching and GMM procedure of Section \ref{estimation} to each trial separately:
\[
\hat\theta_s=\arg\min_{\theta_s}
\hat M_s(\theta_s)^\top\hat W_s\hat M_s(\theta_s),
\qquad
\hat\psi_s=\mathbb P_{n_0}[g(\hat\theta_s)],
\]
where $\hat M_s(\theta_s)=\hat m_s(\theta_s,\hat\eta_s)-\hat\tau_s$ retains only the moments of trial $s$, and $\hat W_s\to_p W_s$ with $W_s$ positive definite. All trials use the same target sample and may share the reference sample from $Q$.

\paragraph{Minimum-variance linear pooling.}
Assume the trial-specific versions of C1--C3 and the sample-size ratios of Corollary \ref{cor:asymptotic-normality}. Let $K_s$, $c_s$, $\xi_q^{(s)}$, and $\xi_s^{(s)}$ denote $J_m^\theta$, $J_\theta^\psi$, $\xi_q$, and $\xi_s$ from Section \ref{estimation}, computed with only trial $s$ retained, with weight $W_s$, and evaluated at $(\theta_{s,0},\eta_{s,0})$; in particular, $J_{\tau_s}^m=-I_{J_s}$. Define the mean-zero scalar functions for notational convenience as
\[
u_s(X)=g_s(X)-\psi_0,\qquad
v_s(X)=c_sK_s\xi_q^{(s)}(X),\qquad
r_s(O)=c_sK_s\xi_s^{(s)}(O).
\]
Proposition \ref{prop:asymptotic-linearity} gives, jointly over trials,
\[
\hat\psi_s-\psi_0
=\mathbb P_{n_0}[u_s]+\mathbb Q_{n_q}[v_s]
 +\mathbb P_{n_s}[r_s]+o_p(n^{-1/2}).
\]
Let $V$ be the asymptotic covariance of $\sqrt n(\hat{\boldsymbol\psi}-\psi_0\mathbf1)$, where $\hat{\boldsymbol\psi}=(\hat\psi_1,\ldots,\hat\psi_M)^\top$. Its entries are
\[
V_{sj}=\pi_0\mathbb E_{P_0}[u_su_j]
       +\pi_q\mathbb E_Q[v_sv_j]
       +\mathbf1\{s=j\}\pi_s\mathbb E_{P_s}[r_s^2].
\]
Thus, because the trials share the target sample and possibly the base sample, the marginalized estimates are generally dependent. If the target sample is reused as the reference sample, i.e., $\mathbb Q_{n_q}=\mathbb P_{n_0}$, replace the first two terms by $\pi_0\mathbb E_{P_0}[(u_s+v_s)(u_j+v_j)]$. For an estimated $Q$, propagate its additional uncertainty as in Appendix \ref{app:variance}.

Under C1s and the plug-in regularity of Appendix \ref{app:variance}, we have the variance estimator $\hat V$. Suppose $V$ is positive definite and $\hat V\to_p V$. We obtain the minimum-variance pooling estimator as
\[
w^* = \frac{V^{-1}\mathbf1}{\mathbf1^\top V^{-1}\mathbf1},
\qquad
\hat w = \frac{\hat V^{-1}\mathbf1}{\mathbf1^\top\hat V^{-1}\mathbf1},
\qquad
\hat\psi_{\mathrm{pool}}=\hat w^\top\hat{\boldsymbol\psi}.
\]
The pooling estimator has the asymptotic linear expansion
\begin{align*}
\hat\psi_{\mathrm{pool}}-\psi_0
={}&\mathbb P_{n_0}\!\left[\sum_s w_{s}^* u_s\right]
   +\mathbb Q_{n_q}\!\left[\sum_s w_{s}^*v_s\right]\\
 &+\sum_s w_{s}^*\mathbb P_{n_s}[r_s]+o_p(n^{-1/2}).
\end{align*}
Because both weight vectors sum to one, estimating the weights contributes only $(\hat w-w^*)^\top(\hat{\boldsymbol\psi}-\psi_0\mathbf1)=o_p(n^{-1/2})$. Among possible linear pooling estimators with limiting weights summing to one, $w^*$ minimizes the asymptotic variance, yielding
\[
\sqrt n(\hat\psi_{\mathrm{pool}}-\psi_0)
\xrightarrow{d}N\!\left(0,\frac{1}{\mathbf1^\top V^{-1}\mathbf1}\right),
\qquad
\hat{\mathrm{Var}}(\hat\psi_{\mathrm{pool}})
=\frac{1}{n\mathbf1^\top\hat V^{-1}\mathbf1}.
\]

\paragraph{Comparison with the common-CATE approach.}
A1p is technically weaker than A1, but A2p requires the target support to be covered by every trial rather than by their union, and identification assumption is required for each trial, i.e., they must report enough information to identify its own CATE parameters. For a linear CATE $g_s(x)=b(x)^\top\theta_{s,0}$ with $c=\mathbb E_{P_0}[b(X)]$, A1p restricts only $c^\top\theta_{s,0}=\psi_0$, so the coefficient vectors may differ across trials. Even for a common linear CATE, pooling the scalar standardized estimates generally gives a numerically different estimate from the joint GMM fit, which combines information on the full coefficient vector. When both approaches are valid, they consistently estimate the same ATE; the common-CATE fit can gain precision from the stronger A1, but it need not remain consistent under A1p alone. Nevertheless, the common-CATE approach can also be misspecified while remaining consistent, as long as the misspecification error lies in the kernel of the target-ATE functional. In contrast, A1p is equivalent to requiring that between-trial CATE heterogeneity lies in the kernel space of target-ATE functional. Neither is strictly weaker than the other.

\Needspace{6\baselineskip}
\subsection{Conditional arm means from aggregate group means}
\label{app:arm-means}

The method in the main text also applies to the identification and estimation of the conditional outcome mean functions. Define the conditional arm means as
\[
\gamma^a(x)=\mathbb E[Y(a)\mid X=x],\qquad
\gamma_s^a(x)=\mathbb E[Y(a)\mid X=x,S=s],\qquad a\in\{0,1\},
\]
and the trials' marginal and subgroup arm means and target arm means by
\[
\nu_{s,j,0}^a=\mathbb E[Y(a)\mid I_{s,j}=1,S=s],\qquad
\nu_0^a=\mathbb E[Y(a)\mid S=0].
\]
The subscript follows the same convention as in the main text, and the superscript denotes treatment arm. Unadjusted arm means estimate these quantities under the standard causal assumptions, complete outcomes, and fixed randomized allocation probability. Standardized adjusted means are also usable if they estimate the same quantities.

\subsubsection{Identification}
For joint transport of both arm means, replace A1 by

\textit{A1-AM. Transportability of conditional outcome means across trials and target population:} $\gamma_s^a(x)=\gamma^a(x),\quad a=0,1$, on the contributing trials' and target's supports. 

Retain A2-A3 for each arm's contributing trials and the required integrability. Let $T_a$ keep the rows for which arm-$a$ means are available. Then
\[
\nu_{s,j,0}^a=(T_a\gamma^a)_{s,j}
=\int\alpha_{s,j,0}(x)\gamma^a(x)\,dQ(x).
\]
For known representers, these arm means identify the projection of $\gamma^a$ onto their span.

Assume a generalized linear model with the chosen general basis $b(x)$, allowing a selected subset for each fitted model, and let
\[
\gamma^a(x;\beta^a)=\ell_a^{-1}\{b(x)^\top\beta^a\}.
\]
The identity and logit links give linear means and binary risks, respectively. Replace A4-A6 by:
\begin{enumerate}
\item[A4-AM.] Each true conditional mean belongs to its model at $\beta_0^a$.
\item[A5-AM.] $T_a\{\gamma^a(\beta^a)-\gamma^a(\beta_0^a)\}\ne0$ for every $\beta^a \neq \beta^a_0$.
\item[A6-AM.] $(T_a-\tilde T_a)\gamma^a(\beta^a)=0$ for every candidate parameter.
\end{enumerate}

These assumptions make $\beta_0^a$ the unique solution of 
\[
\tilde T_a\gamma^a(\beta^a)=\nu_{a,0},
\]
where $\nu_{a,0}$ stacks that arm's population summaries. The argument is the same as Proposition \ref{prop:cate-identification}: A4-AM and A6-AM establish valid moments, and A5-AM establishes uniqueness. Local A6 and working-span alternatives are given in Appendix \ref{app:local-invariance}.

\subsubsection{Estimation and inference}
For $\beta=((\beta^0)^\top,(\beta^1)^\top)^\top$, write
\[
\hat M_{s,a,j}^{\gamma}(\beta)
=\mathbb Q_{n_q}\!\left[\hat\alpha_{s,j}
\ell_a^{-1}\{b^\top\beta^a\}\right]-\hat\nu_{s,j}^a,
\qquad
\hat\beta=\arg\min_\beta(\hat M^\gamma)^\top\hat W_\gamma\hat M^\gamma.
\]
A diagonal positive GMM weight, including inverse marginal variances, is valid for point estimation under identification. Estimate target arm means and their contrasts by
\[
\hat\nu_0^a=\mathbb P_{n_0}[\gamma^a(\hat\beta^a)],\qquad
\hat\psi=\hat\nu_0^1-\hat\nu_0^0,
\qquad \hat\psi^R=\hat\nu_0^1/\hat\nu_0^0
\]
when the ratio denominator is positive.

The moment Jacobian follows
\[
(J_\beta^\gamma)_{s,a,j,\beta^a}
=\mathbb E_Q[\alpha_{s,j}(\ell_a^{-1})'(b^\top\beta_0^a)b^\top]
\]
in its own block; zero otherwise, where $(\ell_a^{-1})'$ is the derivative of the inverse link function, e.g., $(\ell_a^{-1}(c))' = \mathrm{expit}(c)\{1-\mathrm{expit}(c)\}$ for the logit link. Assume full column rank as in C2, the analogues of C3, and joint asymptotic linearity of the reported arm means and covariate summaries. As in the main text, with $E_s^\gamma$ the trial block insertion matrix, define
\begin{align*}
J_{\nu_s}^\gamma&=-E_s^\gamma,&
(J_{\eta_s}^\gamma)_{s,a,j,\cdot}
&=\mathbb E_Q[\alpha_{s,j}\{\gamma^a(\beta_0^a)-\nu_{s,j,0}^a\}h_s^{+\top}],\\
J_{\mu_s}^\gamma&=J_{\eta_s}^\gamma H_s^{-1},&
(\phi_q^\gamma)_{s,a,j}
&=\alpha_{s,j}\{\gamma^a(\beta_0^a)-\nu_{s,j,0}^a\}.
\end{align*}
Define
\[
\xi_q^\gamma=\phi_q^\gamma-\sum_sJ_{\mu_s}^\gamma\tilde\phi_s^{\mu+},\qquad
\xi_s^\gamma=J_{\nu_s}^\gamma\phi_s^\nu+J_{\mu_s}^\gamma\phi_s^{\mu+}.
\]
The analogue of Proposition \ref{prop:asymptotic-linearity} is
\[
\hat\beta-\beta_0
=-\{(J_\beta^\gamma)^\top W_\gamma J_\beta^\gamma\}^{-1}
 (J_\beta^\gamma)^\top W_\gamma
 \left\{\mathbb Q_{n_q}\xi_q^\gamma+\sum_s\mathbb P_{n_s}\xi_s^\gamma\right\}
 +o_p(r_n).
\]
Also,
\[
\hat\nu_0^a - \nu_0^a = \mathbb P_{n_0}[\gamma^a(\beta_0^a)-\nu_0^a] + \mathbb E_{P_0}[\partial_{\beta^a}\gamma^a(\beta_0^a)](\hat\beta^a-\beta_0^a) + o_p(n_0^{-1/2} + r_n).
\]
Appendix \ref{app:covariance-reconstruction} gives the raw arm-mean influence functions and the binary working covariance.

\subsubsection{Comparison with the CATE construction}
A1-AM implies A1, but not conversely, since A1 allows prognostic differences across populations that cancel in the treatment contrast. Transporting arm means thus requires stronger causal assumptions. A4-AM specifies two conditional means. With identity links and a common basis it implies a correct linear CATE; however, a correct CATE need not make either arm mean linear. Nonlinear inverse links might improve arm mean modeling but generate a larger span of candidate outcome-scale functions. Their moment Jacobian requires its own rank checks; small link derivatives can produce weak information. Nonlinear arm models and a linear CATE model are not generally nested, so their modeling restrictions are not universally ordered. For these reasons, we recommend the CATE construction over the conditional arm-mean construction whenever possible.

\Needspace{6\baselineskip}
\subsection{Proof of CATE identification}
\label{app:identification-proof}

\begin{proof}[Proof of Proposition \ref{prop:cate-identification}]
A1, A3, and A4 give $Tg(\theta_0)=Tg_0=\tau_0$, and A6 gives $\tilde Tg(\theta_0)=Tg(\theta_0)$. Thus the working moments vanish at $\theta_0$. If another $\theta$ solves the working equations, A6 also gives $Tg(\theta)=\tilde Tg(\theta)=\tau_0$. Subtracting yields
\[
0=T\{g(\theta)-g(\theta_0)\}
=\bigl(\langle\alpha_{s,j,0},g(\theta)-g(\theta_0)\rangle_{L^2(Q)}\bigr)_{s,j}.
\]
The difference therefore belongs to $\mathcal A^\perp$, which contradicts A5 unless $\theta=\theta_0$.
\end{proof}

\Needspace{6\baselineskip}
\subsection{Proof of asymptotic linearity}
\label{app:asymptotic-proof}

\textit{C3. Sufficient estimation regularity.} The numbers of trials, summaries, and parameters are fixed. Each moment matching equation has a unique finite interior solution with feasible tilting; sample roots are localized in a compact neighborhood containing no other population root. The matrices $H_s$ are non-singular. All subgroup normalizers stay positive locally. Sample moment functions and their relevant derivatives/Jacobians obey uniform laws of large numbers with integrable local envelopes. Every displayed influence contribution has finite second moments. The population GMM criterion has a separated global minimizer on a compact parameter space, and the estimator selects the global minimum. The estimated weight satisfies $\hat W\to_p W$ with $W$ positive definite. Target averages obey a uniform law of large numbers, and the target derivative has an integrable local envelope.

Write $m_{s,j}(\theta,\eta_s)=\mathbb E_Q[\alpha_{s,j}(\eta_s)g(\theta)]$ and stack these population averages into $m(\theta,\eta)$. We use the normalized population and empirical representers defined in the main text. Expectations and derivatives in this appendix are evaluated at $(\theta_0,\eta_0)$ unless otherwise stated. Standard uniform convergence and extremum-estimation arguments are as in \citet{newey1994large}.

\begin{Lemma}[Consistency]
\label{lem:consistency}
Under A1-A6, C1 and C3, the global GMM estimator and target plug-in estimator satisfy
\[
\hat\eta\to_p\eta_0,\qquad \hat\theta\to_p\theta_0,
\qquad \hat\psi\to_p\psi_0.
\]
\end{Lemma}
\begin{proof}
C1 implies consistency of $(\hat\tau,\hat\mu)$. Uniform convergence of the moment matching equations, their uniquely separated population roots, and localization of their sample roots give $\hat\eta\to_p\eta_0$. Positive denominators and uniform convergence of the numerator and denominator averages imply uniform convergence of $\hat m(\theta,\hat\eta)$ to $m(\theta,\eta_0)$. Because $\hat W\to_p W$, the GMM criterion converges uniformly to $M(\theta)^\top W M(\theta)$. Proposition \ref{prop:cate-identification} and positive definiteness of $W$ make $\theta_0$ its unique zero and uniquely separated minimizer. The argmin consistency theorem gives $\hat\theta\to_p\theta_0$. The target uniform law of large numbers and continuity give $\hat\psi\to_p \mathbb E_{P_0}g(\theta_0)=\psi_0$ by A1, A2, and A4.
\end{proof}

\begin{proof}[Proof of Proposition \ref{prop:asymptotic-linearity}]
\textbf{Sample moment.}
For a row $(s,j)$, let $c_{s,j}=\mathbb E_Q[w_sI_{s,j}]>0$. Since $\mathbb E_Q[w_sI_{s,j}(g(\theta_0)-\tau_{s,j,0})]=0$, the exact ratio identity gives
\begin{align*}
\hat m_{s,j}(\theta_0,\eta_{s,0})-\tau_{s,j,0}
&=\frac{(\mathbb Q_{n_q}-Q)[w_sI_{s,j}(g(\theta_0)-\tau_{s,j,0})]}
        {\mathbb Q_{n_q}[w_sI_{s,j}]}\\
&=\mathbb Q_{n_q}[\alpha_{s,j}(g(\theta_0)-\tau_{s,j,0})]
  +o_p(n_q^{-1/2}).
\end{align*}
Thus $\hat M(\theta_0,\tau_0,\eta_0)=\mathbb Q_{n_q}[\phi_q^m]+o_p(n_q^{-1/2})$. This gives the first-order sample moment expansion.

\textbf{Moment matching.}
Differentiating the normalized population ratio gives
\[
\partial_{\eta_s}m_{s,j}
=\mathbb E_Q[\alpha_{s,j}(g(\theta_0)-\tau_{s,j,0})h_s^{+\top}],
\qquad
\partial_\theta m_{s,j}=\mathbb E_Q[\alpha_{s,j}\partial_\theta g(\theta_0)],
\]
which are the rows of $J_{\eta_s}^m$ and $J_\theta^m$. Differentiating the moment matching equations gives $H_s$. Taylor expansion around their population solution yields
\begin{align*}
\hat\eta_s-\eta_{s,0}
&=-H_s^{-1}\{\mathbb Q_{n_q}[\tilde\phi_s^{\mu+}]-(\hat\mu_s^+-\mu_{s,0}^+)\}
  +o_p(n_q^{-1/2}+n_s^{-1/2})\\
&=-H_s^{-1}\mathbb Q_{n_q}[\tilde\phi_s^{\mu+}]+H_s^{-1}\mathbb P_{n_s}[\phi_s^{\mu+}]
  +o_p(n_q^{-1/2}+n_s^{-1/2}).
\end{align*}
Here the mean of $\tilde\phi_s^{\mu+}$ under $Q$ is zero. Hence $\hat\eta-\eta_0=O_p(r_n)$.

\textbf{GMM.}
The empirical normalized ratios have sample derivatives converging uniformly to the derivatives above. Combining the sample moment expansion, C1, and the moment matching expansion gives
\[
\hat M(\hat\theta,\hat\tau,\hat\eta)
=J_\theta^m(\hat\theta-\theta_0)+\Delta_n
  +o_p(r_n+\|\hat\theta-\theta_0\|),
\]
where
\begin{align*}
\Delta_n
&=\mathbb Q_{n_q}[\phi_q^m]-\sum_s E_s\mathbb P_{n_s}[\phi_s^\tau]
 +\sum_sJ_{\mu_s}^m
       \{\mathbb P_{n_s}[\phi_s^{\mu+}]-\mathbb Q_{n_q}[\tilde\phi_s^{\mu+}]\}\\
&=\mathbb Q_{n_q}[\xi_q]+\sum_s\mathbb P_{n_s}[\xi_s]=O_p(r_n).
\end{align*}
For an interior minimizer, the first-order condition is
\[
\partial_\theta\hat m(\hat\theta,\hat\eta)^\top
  \hat W\hat M(\hat\theta,\hat\tau,\hat\eta)=0.
\]
Consistency, C2, and $\hat W\to_p W$ imply
\[
\hat\theta-\theta_0
=-\{(J_\theta^m)^\top WJ_\theta^m\}^{-1}(J_\theta^m)^\top W\Delta_n
 +o_p(r_n+\|\hat\theta-\theta_0\|).
\]
Absorbing the last parameter-error term first establishes $\hat\theta-\theta_0=O_p(r_n)$, and substituting this rate establishes the stated $o_p(r_n)$ remainder. Random GMM weighting contributes no additional first-order term here because the population moment is zero. For a linear CATE approximation under misspecification, the population moment is not necessarily zero. Appendix \ref{app:misspecification} gives the additional first-order terms under the misspecification.

\textbf{Target averaging.}
An integral Taylor expansion and the target derivative uniform law of large numbers give
\begin{align*}
\hat\psi-\psi_0
&=\mathbb P_{n_0}[g(\theta_0)-\psi_0]
 +\left\{\int_0^1\mathbb P_{n_0}[\partial_\theta
    g(\theta_0+t(\hat\theta-\theta_0))]\,dt\right\}
   (\hat\theta-\theta_0)\\
&=\mathbb P_{n_0}[g(\theta_0)-\psi_0]
  +J_\theta^\psi(\hat\theta-\theta_0)
  +o_p(n_0^{-1/2}+r_n).
\end{align*}
The expansion itself does not require independence of the target and reference observations.
\end{proof}

\begin{proof}[Proof of Corollary \ref{cor:asymptotic-normality}]
Every influence contribution has mean zero and finite covariance. Apply the multivariate central limit theorem separately to the independent source samples and combine them with the limiting sample-size ratios. Since $M$ is fixed and the ratios are finite, $\sqrt n\,r_n=O(1)$ and the remainders vanish after multiplication by $\sqrt n$. Independence of the target sample supplies the additive target-variance term. The resulting covariance matrices are the displayed $V_\theta$ and $V_\psi$.
\end{proof}

\clearpage
\section{SIMULATION STUDY}

\subsection{Parameters in the simulation study}
\label{app:sim-parameters}

This appendix specifies the data-generating mechanisms, aggregate reports, implementation choices, and scenario indices for the main simulations in Section \ref{simulation}. The parameters below define the simulations; the scenario indices are also used in the simulation scripts and their output files at: \url{https://github.com/qingyshi/CIMAgD-Sim}. 

\paragraph{Covariates and population selection.}
Write $\bar X=(1,X_1,X_2,X_3)^\top$ so that every coefficient vector below includes an intercept. To specify the mixed joint distribution, generate $Z\sim N_3(0,R_\rho)$, where $R_\rho$ has unit diagonal and off-diagonal entries $\rho$, and set
\[
 X_1=\mathbb I\{Z_1>\Phi^{-1}(1-p_1)\},\qquad
 X_2=\mathbb I\{Z_2>\Phi^{-1}(1-p_2)\},\qquad
 X_3=\mu+Z_3.
\]
Thus $X_k\sim\mathrm{Bernoulli}(p_k)$ for $k=1,2$ and $X_3\sim N(\mu,1)$, and the parameter $\rho$ is the correlation of the latent normal variables. The two covariate specifications are
\[
 \eta^{(1)}=(p_1,p_2,\mu,\rho)=(0.3,0.3,0,0.3),\qquad
 \eta^{(2)}=(0.5,0.5,0,0.5).
\]
Let $R=\mathbb I(S\ne0)$ indicate selection into any trial. Independently across participants, draw $R\mid X\sim\mathrm{Bernoulli}\{\mathrm{expit}(\beta^\top\bar X)\}$; those with $R=0$ form the target sample. For the five-trial setting, use
\begin{align*}
 \beta_5^{(1)}&=(\log 2,\log 0.5,\log 0.5,\log 0.5)^\top, &
 \beta_5^{(2)}&=(\log 0.8,\log 2,\log 2,\log 2)^\top.
\end{align*}
For the single-trial setting, use
\begin{align*}
 \beta_1^{(1)}&=(\log 1.2,\log 0.5,\log 0.5,\log 0.5)^\top, &
 \beta_1^{(2)}&=(\log 0.5,\log 2,\log 2,\log 2)^\top.
\end{align*}
Intercepts only determine the overall trial fraction; they do not change covariate-dependent selection.

\paragraph{Allocation to five trials.}
Conditional on $R=1$, assign $S\in\{1,\ldots,5\}$ with probabilities proportional to $\exp(\gamma_s^\top\bar X)$. Trial 1 is the reference, $\gamma_1=(0,0,0,0)^\top$. Let $a=\log(1.25)$ and $c=\log 2$. Rows of the following matrices correspond to trials 1 through 5 and columns to the intercept, $X_1$, $X_2$, and $X_3$:
\[
 \Gamma^{(1)}=
 \begin{pmatrix}
 0&0&0&0\\
 c&-a&a&-a\\
 c&-a/2&a/2&-a/2\\
 c&-a&a/2&a/2\\
 c&-a/2&a&a/2
 \end{pmatrix},\qquad
 \Gamma^{(2)}=\Gamma^{(1)}\operatorname{diag}(1,-1,-1,-1).
\]
For $n=5000$, expected individual-trial sizes range approximately from 270 to 740 across these scenarios. Realized trial and target sample sizes are random; no samples are discarded to enforce size bounds. In the single-trial setting, $S=1$ whenever $R=1$, and no allocation parameter is needed.

\paragraph{Treatment and binary outcomes.}
Within each trial, assign $A\sim\mathrm{Bernoulli}(0.5)$ independently of covariates. Conditional arm risks are $m_a(x)=\mathrm{expit}\{\theta(a)^\top\bar x\}$, with two specifications:
\begin{align*}
 \theta^{(1)}(1)&=(\log 0.5,\log 2,\log 0.5,\log 1.25)^\top,\\
 \theta^{(1)}(0)&=(\log 0.5,\log 0.5,\log 2,\log 0.8)^\top,\\
 \theta^{(2)}(1)&=(\log 0.5,\log 1.25,\log 0.8,\log 1.1)^\top,\\
 \theta^{(2)}(0)&=(\log 0.5,\log 0.8,\log 1.25,\log 0.9)^\top.
\end{align*}
The first specification gives stronger effect modification. Generate the observed trial outcome as $Y\mid A,X,S\sim\mathrm{Bernoulli}\{m_A(X)\}$, which is equivalent to generate potential outcomes and then impose consistency. Target outcomes and unobserved potential outcomes are unnecessary for estimation or for calculating the population truth. The true CATE is $g_0(x)=m_1(x)-m_0(x)$ and is generally nonlinear, whereas CIMAgD and the IPD pseudo-outcome regression use a linear working CATE.

\paragraph{Aggregate reports and CIMAgD implementation.}
For every trial, report the overall treatment effect and effects for each of the six one-at-a-time subgroups $X_1=0$, $X_1=1$, $X_2=0$, $X_2=1$, $X_3\leq0$, and $X_3>0$. Each report contains treatment- and control-arm sample sizes and event numbers. The risk difference and its standard error are
\[
 \hat\tau=\hat p_1-\hat p_0,\qquad
 \hat{\mathrm{se}}(\hat\tau)=
 \left\{\frac{\hat p_1(1-\hat p_1)}{n_1}
       +\frac{\hat p_0(1-\hat p_0)}{n_0}\right\}^{1/2},
\]
where $n_0$ here denotes the control-arm size within the reported group. There is no continuity correction. A subgroup with an empty arm is unreportable and is omitted; if an overall trial arm is empty, the aggregate-data methods fail for that replication. Reports with zero estimated standard errors are retained for CIMAgD, which then uses its identity-weight. Meta-analysis and meta-regression require strictly positive overall effect variances; otherwise they fail for that replication. 

The covariate reports contain the proportions of $X_1$ and $X_2$, and the mean and sample standard deviation of $X_3$. For trial size $n_s$, we convert the latter two summaries to the raw second moment $\overline{X_3^2}=\bar X_3^2+(n_s-1)s_{X_3}^2/n_s$. Fit the working CATE with basis $(1,X_1,X_2,X_3)$ and tilt with basis $(1,X_1,X_2,X_3,X_3^2)$. The marginal standard errors of the two proportions and of the continuous mean are supplied; the standard error of the raw second moment is left unspecified because it cannot be recovered from the reported mean and standard deviation alone. No joint covariance matrices or cross-covariate moments are supplied. CIMAgD reconstructs missing covariances from the arm reports and the tilted working law as described in Appendix \ref{app:covariance-reconstruction}.

The simulation used the entire target covariate sample both as the empirical reference sample $Q$ and for target averaging, with the shared-sample covariance contribution enabled. There is no sample splitting. Use one-step GMM with the package's default diagonal weights $1/(n_s\widehat{\mathrm{se}}(\hat\tau_{s,j})^2)$ when all standard errors are positive. The default control-mean working basis is the CATE basis. Variance estimator follows Appendix \ref{app:variance-influence} under correct specification, and therefore the coverage reflects both approximation error and miscalibration error. Singular systems are judged as failures with numerical rank checks. 

\paragraph{Comparators and evaluation.}
The IPD g-formula fits $Y\sim A*(X_1+X_2+X_3)$ with the correct logistic link and standardizes over the same full target sample. Its variance combines the model-based coefficient covariance with the empirical target-average contribution, using analytic derivatives. The IPD pseudo-outcome regression fits $U=AY/0.5-(1-A)Y/0.5$ on $(1,X_1,X_2,X_3)$ by ordinary least squares; its variance uses the HC1 sandwich variance and the target-average contribution. The five-trial simulations additionally involve inverse-variance random-effects meta-analysis and meta-regression with REML heterogeneity estimator and normal Wald intervals. Meta-regression includes all three trial-level covariate means, predicts at the target means, and adds the sampling variance of those target means. Five trials provide only one residual degree of freedom for this regression; its instability is part of the intended comparison.

For each scenario, the fixed true population estimand is
\[
 \psi=\frac{\mathbb E_\eta[(1-r_\beta(X))\{m_1(X)-m_0(X)\}]}
              {\mathbb E_\eta[1-r_\beta(X)]},\qquad
 r_\beta(x)=\mathrm{expit}(\beta^\top\bar x).
\]
It is evaluated once by deterministic Gaussian quadrature integration, summing over both binary covariates and integrating over the common latent normal factor and the residual normal component of $X_3$ (60 nodes per dimension). Quadrature integration accuracy is checked against 80 nodes. The realized target-sample effect and an average of simulated potential outcomes are not used as the truth. Each full simulation uses 1000 independent replications. Bias, empirical variance, MSE, and coverage of normal $95\%$ intervals are calculated against this fixed $\psi$. Performance uses successful fits, and failed replications are not replaced. We require at least 95\% successful fits for each method to quantify the performance.

\paragraph{Scenario indices.}
There are 16 scenarios in each main setting. In the five-trial setting, cross two choices each for $\eta$, $\beta_5$, $\Gamma$, and $\theta$, holding $n=5000$. In the single-trial setting, cross $n\in\{1000,2000\}$ with two choices each for $\eta$, $\beta_1$, and $\theta$. Tables \ref{tab:sim-scenarios-five} and \ref{tab:sim-scenarios-single} give the exact indices. In the five-trial table, $\eta$ changes fastest, followed by $\beta_5$, $\Gamma$, and $\theta$; in the single-trial table, $n$ changes fastest, followed by $\eta$, $\beta_1$, and $\theta$.

\begin{table}[htbp]
\centering
\small
\renewcommand{\arraystretch}{1.2}
\caption{Scenario indices for the main five-trial binary-outcome simulation.}
\label{tab:sim-scenarios-five}
\begin{tabular}{rlllll}
\toprule
Scenario & $n$ & Covariates & Selection & Allocation & Outcome coefficients \\
\midrule
1 & 5000 & $\eta^{(1)}$ & $\beta_5^{(1)}$ & $\Gamma^{(1)}$ & $\theta^{(1)}$ \\
2 & 5000 & $\eta^{(2)}$ & $\beta_5^{(1)}$ & $\Gamma^{(1)}$ & $\theta^{(1)}$ \\
3 & 5000 & $\eta^{(1)}$ & $\beta_5^{(2)}$ & $\Gamma^{(1)}$ & $\theta^{(1)}$ \\
4 & 5000 & $\eta^{(2)}$ & $\beta_5^{(2)}$ & $\Gamma^{(1)}$ & $\theta^{(1)}$ \\
5 & 5000 & $\eta^{(1)}$ & $\beta_5^{(1)}$ & $\Gamma^{(2)}$ & $\theta^{(1)}$ \\
6 & 5000 & $\eta^{(2)}$ & $\beta_5^{(1)}$ & $\Gamma^{(2)}$ & $\theta^{(1)}$ \\
7 & 5000 & $\eta^{(1)}$ & $\beta_5^{(2)}$ & $\Gamma^{(2)}$ & $\theta^{(1)}$ \\
8 & 5000 & $\eta^{(2)}$ & $\beta_5^{(2)}$ & $\Gamma^{(2)}$ & $\theta^{(1)}$ \\
9 & 5000 & $\eta^{(1)}$ & $\beta_5^{(1)}$ & $\Gamma^{(1)}$ & $\theta^{(2)}$ \\
10 & 5000 & $\eta^{(2)}$ & $\beta_5^{(1)}$ & $\Gamma^{(1)}$ & $\theta^{(2)}$ \\
11 & 5000 & $\eta^{(1)}$ & $\beta_5^{(2)}$ & $\Gamma^{(1)}$ & $\theta^{(2)}$ \\
12 & 5000 & $\eta^{(2)}$ & $\beta_5^{(2)}$ & $\Gamma^{(1)}$ & $\theta^{(2)}$ \\
13 & 5000 & $\eta^{(1)}$ & $\beta_5^{(1)}$ & $\Gamma^{(2)}$ & $\theta^{(2)}$ \\
14 & 5000 & $\eta^{(2)}$ & $\beta_5^{(1)}$ & $\Gamma^{(2)}$ & $\theta^{(2)}$ \\
15 & 5000 & $\eta^{(1)}$ & $\beta_5^{(2)}$ & $\Gamma^{(2)}$ & $\theta^{(2)}$ \\
16 & 5000 & $\eta^{(2)}$ & $\beta_5^{(2)}$ & $\Gamma^{(2)}$ & $\theta^{(2)}$ \\
\bottomrule
\end{tabular}
\end{table}

\begin{table}[htbp]
\centering
\small
\renewcommand{\arraystretch}{1.2}
\caption{Scenario indices for the main single-trial binary-outcome simulation.}
\label{tab:sim-scenarios-single}
\begin{tabular}{rllll}
\toprule
Scenario & $n$ & Covariates & Selection & Outcome coefficients \\
\midrule
1 & 1000 & $\eta^{(1)}$ & $\beta_1^{(1)}$ & $\theta^{(1)}$ \\
2 & 2000 & $\eta^{(1)}$ & $\beta_1^{(1)}$ & $\theta^{(1)}$ \\
3 & 1000 & $\eta^{(2)}$ & $\beta_1^{(1)}$ & $\theta^{(1)}$ \\
4 & 2000 & $\eta^{(2)}$ & $\beta_1^{(1)}$ & $\theta^{(1)}$ \\
5 & 1000 & $\eta^{(1)}$ & $\beta_1^{(2)}$ & $\theta^{(1)}$ \\
6 & 2000 & $\eta^{(1)}$ & $\beta_1^{(2)}$ & $\theta^{(1)}$ \\
7 & 1000 & $\eta^{(2)}$ & $\beta_1^{(2)}$ & $\theta^{(1)}$ \\
8 & 2000 & $\eta^{(2)}$ & $\beta_1^{(2)}$ & $\theta^{(1)}$ \\
9 & 1000 & $\eta^{(1)}$ & $\beta_1^{(1)}$ & $\theta^{(2)}$ \\
10 & 2000 & $\eta^{(1)}$ & $\beta_1^{(1)}$ & $\theta^{(2)}$ \\
11 & 1000 & $\eta^{(2)}$ & $\beta_1^{(1)}$ & $\theta^{(2)}$ \\
12 & 2000 & $\eta^{(2)}$ & $\beta_1^{(1)}$ & $\theta^{(2)}$ \\
13 & 1000 & $\eta^{(1)}$ & $\beta_1^{(2)}$ & $\theta^{(2)}$ \\
14 & 2000 & $\eta^{(1)}$ & $\beta_1^{(2)}$ & $\theta^{(2)}$ \\
15 & 1000 & $\eta^{(2)}$ & $\beta_1^{(2)}$ & $\theta^{(2)}$ \\
16 & 2000 & $\eta^{(2)}$ & $\beta_1^{(2)}$ & $\theta^{(2)}$ \\
\bottomrule
\end{tabular}
\end{table}

\clearpage
\subsection{Additional simulation study}
\label{app:sim-add}

Two additional experiments examine continuous outcomes and smaller individual trials. They use the same covariate generator, aggregate-data construction, shared target/reference sample, fitting options, failure handling, performance measures, and 1000 replications per scenario as in Appendix \ref{app:sim-parameters}, with the following changes. Tables \ref{tab:sim-scenarios-continuous} and \ref{tab:sim-scenarios-small} explicitly index their 16 scenarios.

\paragraph{Single trial with a continuous outcome.}
Use $n\in\{1000,2000\}$, $\eta^{(1)}$ or $\eta^{(2)}$, and $\beta_1^{(1)}$ or $\beta_1^{(2)}$ exactly as in the main single-trial experiment. Retain the numerical outcome coefficients $\theta^{(1)}(a)$ and $\theta^{(2)}(a)$, but replace the logistic outcome mechanism by
\[
 Y(a)=\theta(a)^\top\bar X+\varepsilon_a,\qquad
 \varepsilon_a\sim N(0,\sigma^2),\qquad \sigma=1,
\]
with errors independent of covariates, selection, and treatment. Thus $g_0(x)=\{\theta(1)-\theta(0)\}^\top\bar x$ is linear. This experiment examines a correctly specified CATE with continuous outcomes.

For each reported group, provide each arm's sample size, mean, and standard error of the mean $s_a/\sqrt{n_a}$. The effect is the mean difference and its standard error is $\sqrt{s_1^2/n_1+s_0^2/n_0}$. A subgroup requires at least two observations in each arm to report both variances; otherwise omit that subgroup. An overall arm with fewer than two observations produces an aggregate-method failure. Use CIMAgD's continuous-outcome covariance reconstruction with arm-specific constant conditional residual variances. Compare CIMAgD with the same IPD pseudo-outcome regression and the correctly specified identity-link Gaussian g-formula $Y\sim A*(X_1+X_2+X_3)$, fitted by least squares. The latter uses the homoskedastic coefficient covariance and the target-average contribution. Compute $\psi$ by the formula in Appendix \ref{app:sim-parameters} with $m_a(x)=\theta(a)^\top\bar x$.

\begin{table}[htbp]
\centering
\small
\renewcommand{\arraystretch}{1.2}
\caption{Scenario indices for the additional single-trial continuous-outcome simulation. The coefficients $\theta$ enter identity-link arm means and $\sigma=1$ throughout.}
\label{tab:sim-scenarios-continuous}
\begin{tabular}{rllll}
\toprule
Scenario & $n$ & Covariates & Selection & Outcome coefficients \\
\midrule
1 & 1000 & $\eta^{(1)}$ & $\beta_1^{(1)}$ & $\theta^{(1)}$ \\
2 & 2000 & $\eta^{(1)}$ & $\beta_1^{(1)}$ & $\theta^{(1)}$ \\
3 & 1000 & $\eta^{(2)}$ & $\beta_1^{(1)}$ & $\theta^{(1)}$ \\
4 & 2000 & $\eta^{(2)}$ & $\beta_1^{(1)}$ & $\theta^{(1)}$ \\
5 & 1000 & $\eta^{(1)}$ & $\beta_1^{(2)}$ & $\theta^{(1)}$ \\
6 & 2000 & $\eta^{(1)}$ & $\beta_1^{(2)}$ & $\theta^{(1)}$ \\
7 & 1000 & $\eta^{(2)}$ & $\beta_1^{(2)}$ & $\theta^{(1)}$ \\
8 & 2000 & $\eta^{(2)}$ & $\beta_1^{(2)}$ & $\theta^{(1)}$ \\
9 & 1000 & $\eta^{(1)}$ & $\beta_1^{(1)}$ & $\theta^{(2)}$ \\
10 & 2000 & $\eta^{(1)}$ & $\beta_1^{(1)}$ & $\theta^{(2)}$ \\
11 & 1000 & $\eta^{(2)}$ & $\beta_1^{(1)}$ & $\theta^{(2)}$ \\
12 & 2000 & $\eta^{(2)}$ & $\beta_1^{(1)}$ & $\theta^{(2)}$ \\
13 & 1000 & $\eta^{(1)}$ & $\beta_1^{(2)}$ & $\theta^{(2)}$ \\
14 & 2000 & $\eta^{(1)}$ & $\beta_1^{(2)}$ & $\theta^{(2)}$ \\
15 & 1000 & $\eta^{(2)}$ & $\beta_1^{(2)}$ & $\theta^{(2)}$ \\
16 & 2000 & $\eta^{(2)}$ & $\beta_1^{(2)}$ & $\theta^{(2)}$ \\
\bottomrule
\end{tabular}
\end{table}

\paragraph{Five trials with smaller binary-outcome samples.}
Set $n=2000$ and retain the two covariate distributions, binary outcome models, and selection slope vectors from the main five-trial experiment. To obtain small trials while keeping the two selection directions, we calibrate the selection intercept separately for each covariate distribution and slope specification. Write
\[
 \beta_{\mathrm{sm},e}^{(k)}=
 \bigl(b_{e,k},(\beta_5^{(k)})_1,(\beta_5^{(k)})_2,(\beta_5^{(k)})_3\bigr)^\top,
 \quad e,k\in\{1,2\},
\]
where the subscripts $1,2,3$ select the non-intercept coefficients, and choose $b_{e,k}$ to satisfy
\[
 \mathbb E_{\eta^{(e)}}\!\left[
 \mathrm{expit}\left\{b_{e,k}+\sum_{j=1}^3(\beta_5^{(k)})_jX_j\right\}\right]=0.375.
\]
The intercepts are
\[
 \begin{array}{c|rr}
 &k=1&k=2\\\hline
 e=1&-0.20579233&-1.01533394\\
 e=2&\phantom{-}0.05142331&-1.33487105
 \end{array}
\]
and the expected total trial size is therefore 750 and the expected target size is 1250.

Define $\Gamma_{\mathrm{sm}}^{(1)}$ and $\Gamma_{\mathrm{sm}}^{(2)}$ by replacing every allocation intercept in $\Gamma^{(1)}$ and $\Gamma^{(2)}$ by zero, retaining their slopes and the trial-1 reference row. Expected trial sizes across these scenarios are approximately 130--170. These are design targets: multinomial allocation and selection yield random realized sizes. Samples are never truncated, rejected, or regenerated to force a range. All five methods from the main five-trial experiment are applied. The true ATE is recalculated under the calibrated selection model using the same method as main simulation. Matching row indices pair the same choices of covariate distribution, selection direction, allocation direction, and outcome coefficients across the main and small-trial experiments.

\begin{table}[htbp]
\centering
\small
\renewcommand{\arraystretch}{1.2}
\caption{Scenario indices for the additional five-trial small-sample binary-outcome simulation.}
\label{tab:sim-scenarios-small}
\begin{tabular}{rlllll}
\toprule
Scenario & $n$ & Covariates & Selection & Allocation & Outcome coefficients \\
\midrule
1 & 2000 & $\eta^{(1)}$ & $\beta_{\mathrm{sm},1}^{(1)}$ & $\Gamma_{\mathrm{sm}}^{(1)}$ & $\theta^{(1)}$ \\
2 & 2000 & $\eta^{(2)}$ & $\beta_{\mathrm{sm},2}^{(1)}$ & $\Gamma_{\mathrm{sm}}^{(1)}$ & $\theta^{(1)}$ \\
3 & 2000 & $\eta^{(1)}$ & $\beta_{\mathrm{sm},1}^{(2)}$ & $\Gamma_{\mathrm{sm}}^{(1)}$ & $\theta^{(1)}$ \\
4 & 2000 & $\eta^{(2)}$ & $\beta_{\mathrm{sm},2}^{(2)}$ & $\Gamma_{\mathrm{sm}}^{(1)}$ & $\theta^{(1)}$ \\
5 & 2000 & $\eta^{(1)}$ & $\beta_{\mathrm{sm},1}^{(1)}$ & $\Gamma_{\mathrm{sm}}^{(2)}$ & $\theta^{(1)}$ \\
6 & 2000 & $\eta^{(2)}$ & $\beta_{\mathrm{sm},2}^{(1)}$ & $\Gamma_{\mathrm{sm}}^{(2)}$ & $\theta^{(1)}$ \\
7 & 2000 & $\eta^{(1)}$ & $\beta_{\mathrm{sm},1}^{(2)}$ & $\Gamma_{\mathrm{sm}}^{(2)}$ & $\theta^{(1)}$ \\
8 & 2000 & $\eta^{(2)}$ & $\beta_{\mathrm{sm},2}^{(2)}$ & $\Gamma_{\mathrm{sm}}^{(2)}$ & $\theta^{(1)}$ \\
9 & 2000 & $\eta^{(1)}$ & $\beta_{\mathrm{sm},1}^{(1)}$ & $\Gamma_{\mathrm{sm}}^{(1)}$ & $\theta^{(2)}$ \\
10 & 2000 & $\eta^{(2)}$ & $\beta_{\mathrm{sm},2}^{(1)}$ & $\Gamma_{\mathrm{sm}}^{(1)}$ & $\theta^{(2)}$ \\
11 & 2000 & $\eta^{(1)}$ & $\beta_{\mathrm{sm},1}^{(2)}$ & $\Gamma_{\mathrm{sm}}^{(1)}$ & $\theta^{(2)}$ \\
12 & 2000 & $\eta^{(2)}$ & $\beta_{\mathrm{sm},2}^{(2)}$ & $\Gamma_{\mathrm{sm}}^{(1)}$ & $\theta^{(2)}$ \\
13 & 2000 & $\eta^{(1)}$ & $\beta_{\mathrm{sm},1}^{(1)}$ & $\Gamma_{\mathrm{sm}}^{(2)}$ & $\theta^{(2)}$ \\
14 & 2000 & $\eta^{(2)}$ & $\beta_{\mathrm{sm},2}^{(1)}$ & $\Gamma_{\mathrm{sm}}^{(2)}$ & $\theta^{(2)}$ \\
15 & 2000 & $\eta^{(1)}$ & $\beta_{\mathrm{sm},1}^{(2)}$ & $\Gamma_{\mathrm{sm}}^{(2)}$ & $\theta^{(2)}$ \\
16 & 2000 & $\eta^{(2)}$ & $\beta_{\mathrm{sm},2}^{(2)}$ & $\Gamma_{\mathrm{sm}}^{(2)}$ & $\theta^{(2)}$ \\
\bottomrule
\end{tabular}
\end{table}

\clearpage
\subsection{Simulation results}
\label{app:sim-results}

Figures \ref{figure3} and \ref{figure4} show the performance in the additional simulations with small sample sizes and a continuous outcome, respectively. The following tables show the bias, variance, coverage, MSE, and MAE results for all settings: Tables \ref{tab:5trial-binary-main-bias_variance}, \ref{tab:5trial-binary-main-coverage}, \ref{tab:5trial-binary-main-mse_mae}, \ref{tab:1trial-binary-main-bias_variance}, \ref{tab:1trial-binary-main-coverage}, \ref{tab:1trial-binary-main-mse_mae}, \ref{tab:5trial-binary-small-add-bias_variance}, \ref{tab:5trial-binary-small-add-coverage}, \ref{tab:5trial-binary-small-add-mse_mae}, \ref{tab:1trial-continuous-add-bias_variance}, \ref{tab:1trial-continuous-add-coverage}, and \ref{tab:1trial-continuous-add-mse_mae}.

\begin{figure}
    \centering
    \includegraphics[width=\textwidth]{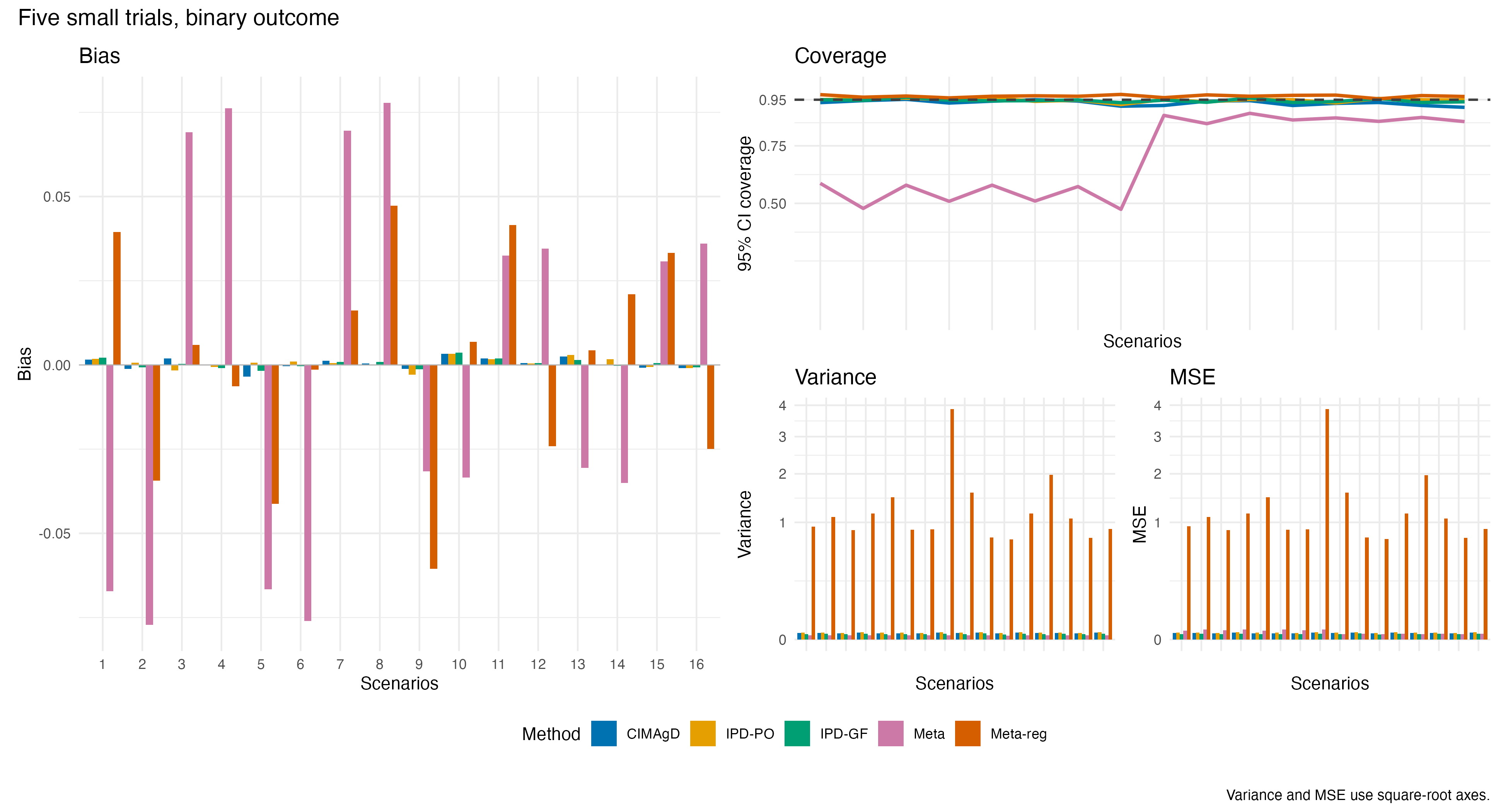}
    \caption{Simulation results in 5-trial setting with small sample sizes}
    \label{figure3}
\end{figure}

\begin{figure}
    \centering
    \includegraphics[width=\textwidth]{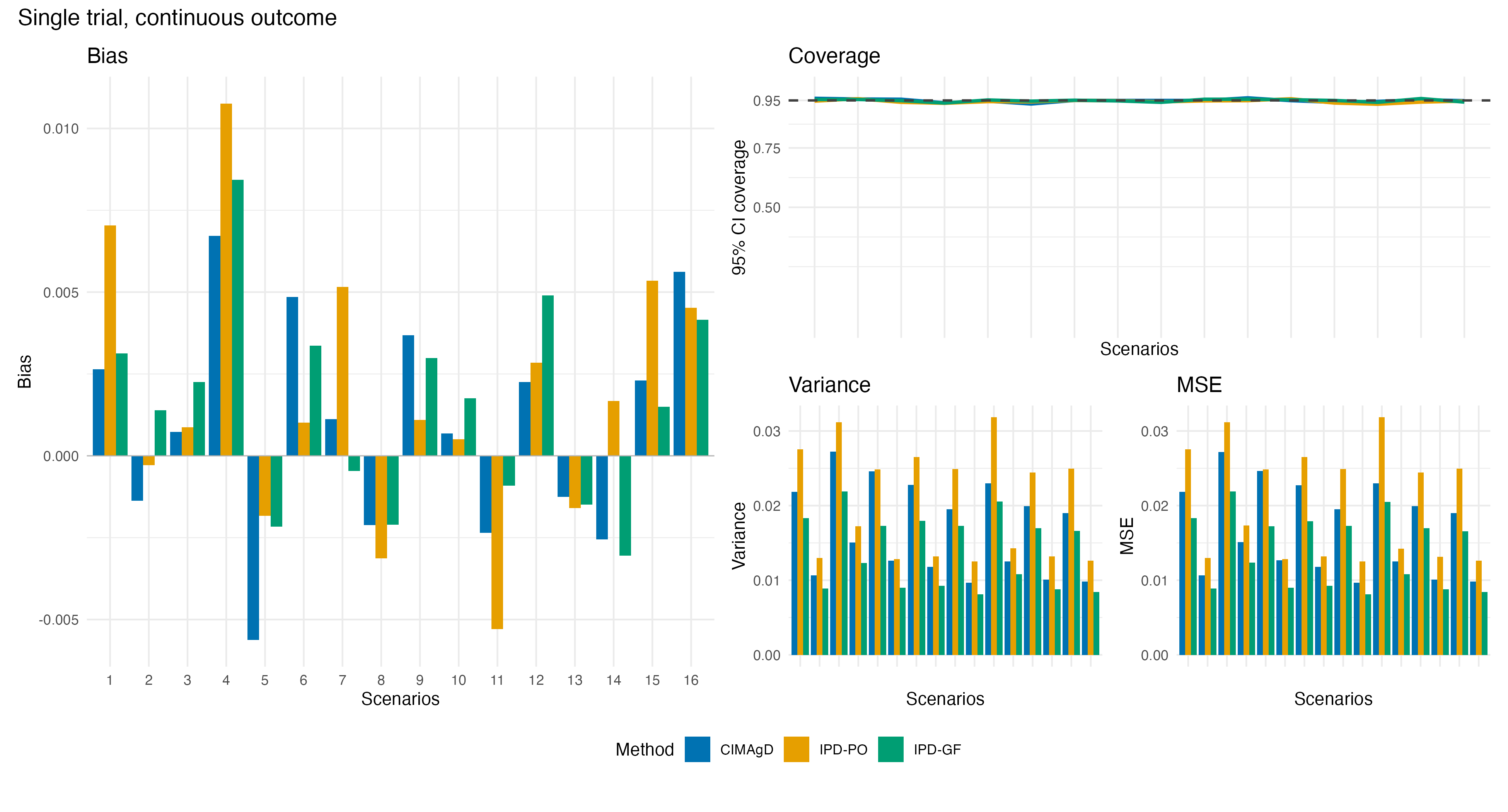}
    \caption{Simulation results in single-trial setting with continuous outcome}
    \label{figure4}
\end{figure}

\clearpage
\begin{table}[htbp]
\centering
\small
\caption{Five trials, binary outcome: bias and variance.}
\label{tab:5trial-binary-main-bias_variance}
\resizebox{\textwidth}{!}{%
\begin{tabular}{rrrrrrrrrrr}
\toprule
Scenario & \multicolumn{2}{c}{CIMAgD} & \multicolumn{2}{c}{IPD-PO} & \multicolumn{2}{c}{IPD-GF} & \multicolumn{2}{c}{Meta} & \multicolumn{2}{c}{Meta-reg} \\
 & Bias & Variance & Bias & Variance & Bias & Variance & Bias & Variance & Bias & Variance \\
\midrule
1 & 0.0002 & 0.0008 & 0.0018 & 0.0009 & -0.0003 & 0.0006 & -0.0678 & 0.0003 & -0.0435 & 0.5862 \\
2 & -0.0006 & 0.0010 & -0.0007 & 0.0011 & -0.0004 & 0.0007 & -0.0765 & 0.0004 & 0.0122 & 0.9173 \\
3 & -0.0006 & 0.0008 & 0.0003 & 0.0009 & -0.0005 & 0.0006 & 0.0677 & 0.0003 & -0.0138 & 0.2476 \\
4 & 0.0006 & 0.0008 & 0.0006 & 0.0009 & 0.0012 & 0.0006 & 0.0779 & 0.0003 & -0.0032 & 0.3779 \\
5 & -0.0016 & 0.0008 & -0.0014 & 0.0009 & -0.0024 & 0.0006 & -0.0688 & 0.0004 & -0.0348 & 0.6395 \\
6 & 0.0006 & 0.0010 & 0.0015 & 0.0011 & 0.0008 & 0.0007 & -0.0771 & 0.0003 & 0.0337 & 1.2304 \\
7 & -0.0021 & 0.0008 & -0.0001 & 0.0009 & -0.0013 & 0.0005 & 0.0659 & 0.0003 & -0.0133 & 0.3755 \\
8 & -0.0021 & 0.0007 & -0.0012 & 0.0009 & -0.0015 & 0.0006 & 0.0757 & 0.0003 & -0.0077 & 0.4779 \\
9 & -0.0010 & 0.0009 & -0.0004 & 0.0009 & -0.0008 & 0.0006 & -0.0315 & 0.0003 & -0.0138 & 0.6016 \\
10 & 0.0007 & 0.0010 & -0.0002 & 0.0011 & -0.0002 & 0.0007 & -0.0351 & 0.0003 & -0.0162 & 0.4631 \\
11 & -0.0001 & 0.0007 & $4.26\times10^{-5}$ & 0.0008 & -0.0003 & 0.0005 & 0.0305 & 0.0003 & 0.0015 & 0.1827 \\
12 & -0.0005 & 0.0008 & -0.0003 & 0.0008 & -0.0007 & 0.0006 & 0.0352 & 0.0003 & 0.0067 & 0.3205 \\
13 & 0.0006 & 0.0008 & 0.0008 & 0.0009 & 0.0010 & 0.0006 & -0.0309 & 0.0003 & 0.0070 & 0.6882 \\
14 & -0.0009 & 0.0009 & -0.0012 & 0.0010 & -0.0006 & 0.0007 & -0.0353 & 0.0003 & -0.0444 & 0.8813 \\
15 & 0.0006 & 0.0008 & 0.0007 & 0.0009 & 0.0001 & 0.0006 & 0.0306 & 0.0004 & 0.0246 & 0.2944 \\
16 & 0.0014 & 0.0008 & 0.0007 & 0.0009 & 0.0006 & 0.0006 & 0.0354 & 0.0003 & -0.0044 & 0.8918 \\
\bottomrule
\end{tabular}%
}
\par\smallskip
\begin{minipage}{\textwidth}\footnotesize IPD-PO: pseudo-outcome regression; IPD-GF: g-formula.\end{minipage}
\end{table}

\clearpage
\begin{table}[htbp]
\centering
\small
\caption{Five trials, binary outcome: coverage.}
\label{tab:5trial-binary-main-coverage}
\resizebox{\textwidth}{!}{%
\begin{tabular}{rrrrrr}
\toprule
Scenario & \multicolumn{1}{c}{CIMAgD} & \multicolumn{1}{c}{IPD-PO} & \multicolumn{1}{c}{IPD-GF} & \multicolumn{1}{c}{Meta} & \multicolumn{1}{c}{Meta-reg} \\
 & Coverage (\%) & Coverage (\%) & Coverage (\%) & Coverage (\%) & Coverage (\%) \\
\midrule
1 & 95.0 & 95.4 & 95.5 & 9.0 & 96.5 \\
2 & 94.2 & 95.2 & 94.1 & 6.9 & 95.6 \\
3 & 94.1 & 94.7 & 93.9 & 10.8 & 96.0 \\
4 & 93.8 & 95.0 & 94.7 & 3.1 & 96.2 \\
5 & 95.0 & 95.0 & 95.4 & 8.5 & 96.9 \\
6 & 94.0 & 94.7 & 95.9 & 4.8 & 96.8 \\
7 & 94.9 & 93.7 & 95.1 & 13.6 & 95.3 \\
8 & 95.9 & 95.9 & 95.2 & 3.0 & 96.8 \\
9 & 94.4 & 94.4 & 94.4 & 66.5 & 97.6 \\
10 & 94.4 & 95.3 & 95.0 & 62.8 & 95.8 \\
11 & 95.9 & 96.0 & 95.8 & 68.5 & 96.6 \\
12 & 95.3 & 96.4 & 96.0 & 58.3 & 96.9 \\
13 & 93.5 & 95.5 & 94.3 & 66.2 & 95.1 \\
14 & 95.4 & 96.1 & 95.9 & 60.1 & 96.4 \\
15 & 95.2 & 94.4 & 94.1 & 68.3 & 96.0 \\
16 & 95.2 & 94.8 & 94.9 & 54.4 & 96.4 \\
\bottomrule
\end{tabular}%
}
\par\smallskip
\begin{minipage}{\textwidth}\footnotesize IPD-PO: pseudo-outcome regression; IPD-GF: g-formula. Coverage ranges across scenarios (minimum--maximum): CIMAgD: 93.5--95.9\%; IPD-PO: 93.7--96.4\%; IPD-GF: 93.9--96.0\%; Meta: 3.0--68.5\%; Meta-reg: 95.1--97.6\%.\end{minipage}
\end{table}

\clearpage
\begin{table}[htbp]
\centering
\small
\caption{Five trials, binary outcome: MSE and MAE.}
\label{tab:5trial-binary-main-mse_mae}
\resizebox{\textwidth}{!}{%
\begin{tabular}{rrrrrrrrrrr}
\toprule
Scenario & \multicolumn{2}{c}{CIMAgD} & \multicolumn{2}{c}{IPD-PO} & \multicolumn{2}{c}{IPD-GF} & \multicolumn{2}{c}{Meta} & \multicolumn{2}{c}{Meta-reg} \\
 & MSE & MAE & MSE & MAE & MSE & MAE & MSE & MAE & MSE & MAE \\
\midrule
1 & 0.0008 & 0.0223 & 0.0009 & 0.0232 & 0.0006 & 0.0191 & 0.0049 & 0.0678 & 0.5875 & 0.4239 \\
2 & 0.0010 & 0.0253 & 0.0011 & 0.0269 & 0.0007 & 0.0218 & 0.0062 & 0.0765 & 0.9165 & 0.4341 \\
3 & 0.0008 & 0.0223 & 0.0009 & 0.0234 & 0.0006 & 0.0193 & 0.0049 & 0.0677 & 0.2475 & 0.2769 \\
4 & 0.0008 & 0.0223 & 0.0009 & 0.0241 & 0.0006 & 0.0197 & 0.0064 & 0.0779 & 0.3776 & 0.3286 \\
5 & 0.0008 & 0.0226 & 0.0009 & 0.0240 & 0.0006 & 0.0198 & 0.0051 & 0.0688 & 0.6400 & 0.4263 \\
6 & 0.0010 & 0.0251 & 0.0011 & 0.0263 & 0.0007 & 0.0218 & 0.0063 & 0.0771 & 1.2303 & 0.4121 \\
7 & 0.0008 & 0.0219 & 0.0009 & 0.0238 & 0.0005 & 0.0186 & 0.0047 & 0.0659 & 0.3753 & 0.3184 \\
8 & 0.0007 & 0.0213 & 0.0009 & 0.0236 & 0.0006 & 0.0190 & 0.0060 & 0.0757 & 0.4775 & 0.3450 \\
9 & 0.0009 & 0.0235 & 0.0009 & 0.0243 & 0.0006 & 0.0203 & 0.0013 & 0.0320 & 0.6012 & 0.4191 \\
10 & 0.0010 & 0.0252 & 0.0010 & 0.0261 & 0.0007 & 0.0212 & 0.0016 & 0.0355 & 0.4629 & 0.3819 \\
11 & 0.0007 & 0.0212 & 0.0008 & 0.0225 & 0.0005 & 0.0186 & 0.0012 & 0.0311 & 0.1826 & 0.2599 \\
12 & 0.0008 & 0.0224 & 0.0008 & 0.0231 & 0.0006 & 0.0190 & 0.0016 & 0.0355 & 0.3202 & 0.3303 \\
13 & 0.0008 & 0.0224 & 0.0009 & 0.0239 & 0.0006 & 0.0199 & 0.0013 & 0.0316 & 0.6875 & 0.4426 \\
14 & 0.0009 & 0.0245 & 0.0010 & 0.0256 & 0.0007 & 0.0210 & 0.0016 & 0.0356 & 0.8824 & 0.4201 \\
15 & 0.0008 & 0.0230 & 0.0009 & 0.0242 & 0.0006 & 0.0195 & 0.0013 & 0.0315 & 0.2947 & 0.3036 \\
16 & 0.0008 & 0.0217 & 0.0009 & 0.0241 & 0.0006 & 0.0197 & 0.0016 & 0.0357 & 0.8910 & 0.3964 \\
\bottomrule
\end{tabular}%
}
\par\smallskip
\begin{minipage}{\textwidth}\footnotesize IPD-PO: pseudo-outcome regression; IPD-GF: g-formula.\end{minipage}
\end{table}

\clearpage
\begin{table}[htbp]
\centering
\small
\caption{Single trial, binary outcome: bias and variance.}
\label{tab:1trial-binary-main-bias_variance}
\resizebox{\textwidth}{!}{%
\begin{tabular}{rrrrrrr}
\toprule
Scenario & \multicolumn{2}{c}{CIMAgD} & \multicolumn{2}{c}{IPD-PO} & \multicolumn{2}{c}{IPD-GF} \\
 & Bias & Variance & Bias & Variance & Bias & Variance \\
\midrule
1 & -0.0042 & 0.0039 & -0.0005 & 0.0058 & -0.0026 & 0.0037 \\
2 & 0.0003 & 0.0021 & 0.0025 & 0.0028 & 0.0014 & 0.0019 \\
3 & -0.0018 & 0.0050 & 0.0007 & 0.0067 & 0.0014 & 0.0044 \\
4 & -0.0031 & 0.0025 & -0.0026 & 0.0035 & -0.0022 & 0.0022 \\
5 & 0.0015 & 0.0039 & -0.0029 & 0.0047 & -0.0008 & 0.0034 \\
6 & 0.0003 & 0.0022 & -0.0009 & 0.0028 & $4.42\times10^{-5}$ & 0.0018 \\
7 & 0.0041 & 0.0041 & 0.0024 & 0.0052 & 0.0011 & 0.0036 \\
8 & -0.0017 & 0.0021 & -0.0018 & 0.0028 & -0.0018 & 0.0019 \\
9 & -0.0010 & 0.0043 & 0.0002 & 0.0054 & -0.0017 & 0.0037 \\
10 & -0.0006 & 0.0020 & 0.0010 & 0.0027 & -0.0008 & 0.0017 \\
11 & -0.0021 & 0.0051 & -0.0006 & 0.0074 & -0.0007 & 0.0047 \\
12 & 0.0002 & 0.0026 & -0.0009 & 0.0031 & 0.0002 & 0.0022 \\
13 & 0.0010 & 0.0042 & 0.0018 & 0.0055 & -0.0004 & 0.0038 \\
14 & -0.0005 & 0.0022 & -0.0021 & 0.0030 & -0.0003 & 0.0018 \\
15 & 0.0024 & 0.0042 & 0.0023 & 0.0056 & 0.0026 & 0.0039 \\
16 & 0.0017 & 0.0020 & 0.0014 & 0.0026 & 0.0007 & 0.0018 \\
\bottomrule
\end{tabular}%
}
\par\smallskip
\begin{minipage}{\textwidth}\footnotesize IPD-PO: pseudo-outcome regression; IPD-GF: g-formula.\end{minipage}
\end{table}

\clearpage
\begin{table}[htbp]
\centering
\small
\caption{Single trial, binary outcome: coverage.}
\label{tab:1trial-binary-main-coverage}
\resizebox{\textwidth}{!}{%
\begin{tabular}{rrrr}
\toprule
Scenario & \multicolumn{1}{c}{CIMAgD} & \multicolumn{1}{c}{IPD-PO} & \multicolumn{1}{c}{IPD-GF} \\
 & Coverage (\%) & Coverage (\%) & Coverage (\%) \\
\midrule
1 & 95.2 & 94.4 & 93.2 \\
2 & 95.3 & 95.1 & 93.5 \\
3 & 95.4 & 95.0 & 94.5 \\
4 & 96.2 & 95.1 & 95.3 \\
5 & 95.8 & 96.4 & 95.1 \\
6 & 94.0 & 94.9 & 94.2 \\
7 & 95.9 & 95.9 & 94.9 \\
8 & 95.7 & 94.9 & 93.6 \\
9 & 95.3 & 94.3 & 94.2 \\
10 & 95.5 & 95.4 & 95.4 \\
11 & 95.0 & 94.2 & 94.2 \\
12 & 94.9 & 96.0 & 95.0 \\
13 & 94.2 & 94.1 & 94.2 \\
14 & 94.7 & 93.0 & 94.3 \\
15 & 94.8 & 95.1 & 94.7 \\
16 & 95.3 & 96.0 & 94.8 \\
\bottomrule
\end{tabular}%
}
\par\smallskip
\begin{minipage}{\textwidth}\footnotesize IPD-PO: pseudo-outcome regression; IPD-GF: g-formula. Coverage ranges across scenarios (minimum--maximum): CIMAgD: 94.0--96.2\%; IPD-PO: 93.0--96.4\%; IPD-GF: 93.2--95.4\%.\end{minipage}
\end{table}

\clearpage
\begin{table}[htbp]
\centering
\small
\caption{Single trial, binary outcome: MSE and MAE.}
\label{tab:1trial-binary-main-mse_mae}
\resizebox{\textwidth}{!}{%
\begin{tabular}{rrrrrrr}
\toprule
Scenario & \multicolumn{2}{c}{CIMAgD} & \multicolumn{2}{c}{IPD-PO} & \multicolumn{2}{c}{IPD-GF} \\
 & MSE & MAE & MSE & MAE & MSE & MAE \\
\midrule
1 & 0.0039 & 0.0497 & 0.0058 & 0.0607 & 0.0037 & 0.0483 \\
2 & 0.0021 & 0.0362 & 0.0028 & 0.0428 & 0.0019 & 0.0351 \\
3 & 0.0050 & 0.0560 & 0.0067 & 0.0649 & 0.0044 & 0.0522 \\
4 & 0.0025 & 0.0406 & 0.0035 & 0.0475 & 0.0022 & 0.0378 \\
5 & 0.0039 & 0.0502 & 0.0047 & 0.0549 & 0.0034 & 0.0469 \\
6 & 0.0022 & 0.0371 & 0.0028 & 0.0424 & 0.0018 & 0.0335 \\
7 & 0.0041 & 0.0508 & 0.0052 & 0.0584 & 0.0036 & 0.0476 \\
8 & 0.0021 & 0.0358 & 0.0028 & 0.0422 & 0.0019 & 0.0345 \\
9 & 0.0043 & 0.0530 & 0.0054 & 0.0581 & 0.0037 & 0.0484 \\
10 & 0.0020 & 0.0359 & 0.0027 & 0.0410 & 0.0017 & 0.0331 \\
11 & 0.0051 & 0.0570 & 0.0074 & 0.0686 & 0.0047 & 0.0548 \\
12 & 0.0026 & 0.0399 & 0.0031 & 0.0446 & 0.0022 & 0.0370 \\
13 & 0.0042 & 0.0523 & 0.0055 & 0.0590 & 0.0037 & 0.0492 \\
14 & 0.0022 & 0.0373 & 0.0030 & 0.0432 & 0.0018 & 0.0338 \\
15 & 0.0042 & 0.0516 & 0.0056 & 0.0609 & 0.0039 & 0.0497 \\
16 & 0.0020 & 0.0356 & 0.0026 & 0.0409 & 0.0018 & 0.0333 \\
\bottomrule
\end{tabular}%
}
\par\smallskip
\begin{minipage}{\textwidth}\footnotesize IPD-PO: pseudo-outcome regression; IPD-GF: g-formula.\end{minipage}
\end{table}

\clearpage
\begin{table}[htbp]
\centering
\small
\caption{Five small trials, binary outcome: bias and variance.}
\label{tab:5trial-binary-small-add-bias_variance}
\resizebox{\textwidth}{!}{%
\begin{tabular}{rrrrrrrrrrr}
\toprule
Scenario & \multicolumn{2}{c}{CIMAgD} & \multicolumn{2}{c}{IPD-PO} & \multicolumn{2}{c}{IPD-GF} & \multicolumn{2}{c}{Meta} & \multicolumn{2}{c}{Meta-reg} \\
 & Bias & Variance & Bias & Variance & Bias & Variance & Bias & Variance & Bias & Variance \\
\midrule
1 & 0.0016 & 0.0030 & 0.0018 & 0.0033 & 0.0021 & 0.0022 & -0.0672 & 0.0013 & 0.0394 & 0.9316 \\
2 & -0.0012 & 0.0032 & 0.0006 & 0.0036 & -0.0007 & 0.0023 & -0.0772 & 0.0012 & -0.0343 & 1.0966 \\
3 & 0.0020 & 0.0027 & -0.0016 & 0.0030 & 0.0003 & 0.0020 & 0.0691 & 0.0012 & 0.0060 & 0.8747 \\
4 & -0.0002 & 0.0034 & -0.0006 & 0.0037 & -0.0009 & 0.0025 & 0.0763 & 0.0013 & -0.0063 & 1.1560 \\
5 & -0.0035 & 0.0028 & 0.0006 & 0.0034 & -0.0017 & 0.0021 & -0.0666 & 0.0012 & -0.0413 & 1.4724 \\
6 & -0.0004 & 0.0029 & 0.0010 & 0.0036 & -0.0003 & 0.0023 & -0.0761 & 0.0013 & -0.0014 & 0.8770 \\
7 & 0.0013 & 0.0028 & 0.0006 & 0.0032 & 0.0009 & 0.0022 & 0.0696 & 0.0012 & 0.0162 & 0.8837 \\
8 & 0.0004 & 0.0034 & $-8.52\times10^{-5}$ & 0.0040 & 0.0009 & 0.0025 & 0.0778 & 0.0013 & 0.0472 & 3.8660 \\
9 & -0.0012 & 0.0032 & -0.0029 & 0.0033 & -0.0013 & 0.0022 & -0.0316 & 0.0012 & -0.0606 & 1.5701 \\
10 & 0.0033 & 0.0034 & 0.0033 & 0.0040 & 0.0037 & 0.0027 & -0.0335 & 0.0013 & 0.0068 & 0.7580 \\
11 & 0.0019 & 0.0027 & 0.0017 & 0.0031 & 0.0020 & 0.0020 & 0.0325 & 0.0011 & 0.0416 & 0.7332 \\
12 & 0.0006 & 0.0034 & 0.0004 & 0.0037 & 0.0005 & 0.0025 & 0.0345 & 0.0013 & -0.0242 & 1.1548 \\
13 & 0.0026 & 0.0031 & 0.0030 & 0.0036 & 0.0014 & 0.0023 & -0.0306 & 0.0013 & 0.0043 & 1.9722 \\
14 & -0.0001 & 0.0031 & 0.0017 & 0.0035 & -0.0003 & 0.0024 & -0.0350 & 0.0012 & 0.0210 & 1.0690 \\
15 & -0.0009 & 0.0029 & -0.0006 & 0.0032 & 0.0005 & 0.0021 & 0.0308 & 0.0013 & 0.0333 & 0.7523 \\
16 & -0.0009 & 0.0034 & -0.0009 & 0.0038 & -0.0008 & 0.0025 & 0.0361 & 0.0012 & -0.0249 & 0.8898 \\
\bottomrule
\end{tabular}%
}
\par\smallskip
\begin{minipage}{\textwidth}\footnotesize IPD-PO: pseudo-outcome regression; IPD-GF: g-formula.\end{minipage}
\end{table}

\clearpage
\begin{table}[htbp]
\centering
\small
\caption{Five small trials, binary outcome: coverage.}
\label{tab:5trial-binary-small-add-coverage}
\resizebox{\textwidth}{!}{%
\begin{tabular}{rrrrrr}
\toprule
Scenario & \multicolumn{1}{c}{CIMAgD} & \multicolumn{1}{c}{IPD-PO} & \multicolumn{1}{c}{IPD-GF} & \multicolumn{1}{c}{Meta} & \multicolumn{1}{c}{Meta-reg} \\
 & Coverage (\%) & Coverage (\%) & Coverage (\%) & Coverage (\%) & Coverage (\%) \\
\midrule
1 & 93.8 & 95.0 & 95.0 & 58.7 & 97.2 \\
2 & 94.6 & 95.4 & 94.7 & 47.8 & 96.1 \\
3 & 95.2 & 95.7 & 96.1 & 57.9 & 96.6 \\
4 & 93.6 & 95.0 & 94.8 & 50.9 & 95.8 \\
5 & 94.4 & 95.9 & 94.5 & 57.9 & 96.5 \\
6 & 95.0 & 94.3 & 94.6 & 51.0 & 96.7 \\
7 & 94.4 & 94.8 & 94.9 & 57.3 & 96.5 \\
8 & 92.2 & 92.9 & 93.7 & 47.4 & 97.3 \\
9 & 92.5 & 95.1 & 94.8 & 88.2 & 95.9 \\
10 & 94.4 & 94.1 & 93.9 & 84.6 & 97.1 \\
11 & 94.6 & 95.1 & 95.8 & 89.1 & 96.5 \\
12 & 92.5 & 95.0 & 93.8 & 86.2 & 96.9 \\
13 & 93.4 & 93.6 & 94.2 & 87.1 & 97.0 \\
14 & 93.8 & 95.5 & 95.5 & 85.6 & 95.4 \\
15 & 92.5 & 95.0 & 93.8 & 87.3 & 96.8 \\
16 & 91.7 & 94.9 & 94.1 & 85.5 & 96.4 \\
\bottomrule
\end{tabular}%
}
\par\smallskip
\begin{minipage}{\textwidth}\footnotesize IPD-PO: pseudo-outcome regression; IPD-GF: g-formula. Coverage ranges across scenarios (minimum--maximum): CIMAgD: 91.7--95.2\%; IPD-PO: 92.9--95.9\%; IPD-GF: 93.7--96.1\%; Meta: 47.4--89.1\%; Meta-reg: 95.4--97.3\%.\end{minipage}
\end{table}

\clearpage
\begin{table}[htbp]
\centering
\small
\caption{Five small trials, binary outcome: MSE and MAE.}
\label{tab:5trial-binary-small-add-mse_mae}
\resizebox{\textwidth}{!}{%
\begin{tabular}{rrrrrrrrrrr}
\toprule
Scenario & \multicolumn{2}{c}{CIMAgD} & \multicolumn{2}{c}{IPD-PO} & \multicolumn{2}{c}{IPD-GF} & \multicolumn{2}{c}{Meta} & \multicolumn{2}{c}{Meta-reg} \\
 & MSE & MAE & MSE & MAE & MSE & MAE & MSE & MAE & MSE & MAE \\
\midrule
1 & 0.0030 & 0.0430 & 0.0033 & 0.0459 & 0.0022 & 0.0376 & 0.0058 & 0.0678 & 0.9322 & 0.5581 \\
2 & 0.0032 & 0.0450 & 0.0036 & 0.0484 & 0.0023 & 0.0385 & 0.0071 & 0.0776 & 1.0967 & 0.5789 \\
3 & 0.0027 & 0.0422 & 0.0030 & 0.0436 & 0.0020 & 0.0362 & 0.0060 & 0.0697 & 0.8739 & 0.4594 \\
4 & 0.0034 & 0.0457 & 0.0037 & 0.0485 & 0.0025 & 0.0393 & 0.0071 & 0.0767 & 1.1549 & 0.5725 \\
5 & 0.0028 & 0.0423 & 0.0034 & 0.0465 & 0.0021 & 0.0365 & 0.0056 & 0.0674 & 1.4726 & 0.6060 \\
6 & 0.0029 & 0.0429 & 0.0036 & 0.0474 & 0.0023 & 0.0375 & 0.0070 & 0.0765 & 0.8761 & 0.5622 \\
7 & 0.0028 & 0.0421 & 0.0032 & 0.0450 & 0.0022 & 0.0373 & 0.0061 & 0.0702 & 0.8831 & 0.4755 \\
8 & 0.0034 & 0.0453 & 0.0040 & 0.0496 & 0.0025 & 0.0397 & 0.0073 & 0.0782 & 3.8644 & 0.6246 \\
9 & 0.0032 & 0.0451 & 0.0033 & 0.0462 & 0.0022 & 0.0375 & 0.0022 & 0.0381 & 1.5722 & 0.6251 \\
10 & 0.0034 & 0.0465 & 0.0040 & 0.0509 & 0.0027 & 0.0418 & 0.0024 & 0.0405 & 0.7573 & 0.5527 \\
11 & 0.0027 & 0.0420 & 0.0031 & 0.0446 & 0.0020 & 0.0352 & 0.0022 & 0.0390 & 0.7342 & 0.4260 \\
12 & 0.0034 & 0.0463 & 0.0037 & 0.0483 & 0.0025 & 0.0397 & 0.0025 & 0.0408 & 1.1543 & 0.5649 \\
13 & 0.0031 & 0.0446 & 0.0036 & 0.0480 & 0.0023 & 0.0383 & 0.0022 & 0.0383 & 1.9702 & 0.6086 \\
14 & 0.0031 & 0.0444 & 0.0035 & 0.0476 & 0.0024 & 0.0384 & 0.0024 & 0.0406 & 1.0684 & 0.6050 \\
15 & 0.0029 & 0.0432 & 0.0032 & 0.0459 & 0.0021 & 0.0366 & 0.0022 & 0.0381 & 0.7526 & 0.4721 \\
16 & 0.0034 & 0.0458 & 0.0038 & 0.0484 & 0.0025 & 0.0398 & 0.0025 & 0.0415 & 0.8895 & 0.5544 \\
\bottomrule
\end{tabular}%
}
\par\smallskip
\begin{minipage}{\textwidth}\footnotesize IPD-PO: pseudo-outcome regression; IPD-GF: g-formula.\end{minipage}
\end{table}

\clearpage
\begin{table}[htbp]
\centering
\small
\caption{Single trial, continuous outcome: bias and variance.}
\label{tab:1trial-continuous-add-bias_variance}
\resizebox{\textwidth}{!}{%
\begin{tabular}{rrrrrrr}
\toprule
Scenario & \multicolumn{2}{c}{CIMAgD} & \multicolumn{2}{c}{IPD-PO} & \multicolumn{2}{c}{IPD-GF} \\
 & Bias & Variance & Bias & Variance & Bias & Variance \\
\midrule
1 & 0.0026 & 0.0219 & 0.0070 & 0.0275 & 0.0031 & 0.0183 \\
2 & -0.0014 & 0.0107 & -0.0003 & 0.0130 & 0.0014 & 0.0089 \\
3 & 0.0007 & 0.0272 & 0.0009 & 0.0312 & 0.0023 & 0.0219 \\
4 & 0.0067 & 0.0151 & 0.0108 & 0.0172 & 0.0084 & 0.0123 \\
5 & -0.0056 & 0.0246 & -0.0018 & 0.0249 & -0.0022 & 0.0173 \\
6 & 0.0049 & 0.0126 & 0.0010 & 0.0129 & 0.0034 & 0.0090 \\
7 & 0.0011 & 0.0228 & 0.0052 & 0.0265 & -0.0005 & 0.0179 \\
8 & -0.0021 & 0.0118 & -0.0031 & 0.0132 & -0.0021 & 0.0093 \\
9 & 0.0037 & 0.0195 & 0.0011 & 0.0249 & 0.0030 & 0.0173 \\
10 & 0.0007 & 0.0097 & 0.0005 & 0.0125 & 0.0018 & 0.0081 \\
11 & -0.0024 & 0.0230 & -0.0053 & 0.0318 & -0.0009 & 0.0205 \\
12 & 0.0023 & 0.0125 & 0.0028 & 0.0143 & 0.0049 & 0.0108 \\
13 & -0.0013 & 0.0199 & -0.0016 & 0.0244 & -0.0015 & 0.0170 \\
14 & -0.0026 & 0.0101 & 0.0017 & 0.0132 & -0.0030 & 0.0088 \\
15 & 0.0023 & 0.0190 & 0.0053 & 0.0250 & 0.0015 & 0.0166 \\
16 & 0.0056 & 0.0098 & 0.0045 & 0.0126 & 0.0042 & 0.0084 \\
\bottomrule
\end{tabular}%
}
\par\smallskip
\begin{minipage}{\textwidth}\footnotesize IPD-PO: pseudo-outcome regression; IPD-GF: g-formula.\end{minipage}
\end{table}

\clearpage
\begin{table}[htbp]
\centering
\small
\caption{Single trial, continuous outcome: coverage.}
\label{tab:1trial-continuous-add-coverage}
\resizebox{\textwidth}{!}{%
\begin{tabular}{rrrr}
\toprule
Scenario & \multicolumn{1}{c}{CIMAgD} & \multicolumn{1}{c}{IPD-PO} & \multicolumn{1}{c}{IPD-GF} \\
 & Coverage (\%) & Coverage (\%) & Coverage (\%) \\
\midrule
1 & 95.9 & 94.7 & 95.4 \\
2 & 95.6 & 95.7 & 95.4 \\
3 & 95.5 & 94.2 & 94.9 \\
4 & 94.0 & 93.8 & 93.9 \\
5 & 94.7 & 94.5 & 95.2 \\
6 & 93.5 & 94.4 & 94.7 \\
7 & 95.0 & 95.0 & 95.1 \\
8 & 95.0 & 94.9 & 94.8 \\
9 & 95.0 & 94.4 & 94.2 \\
10 & 95.0 & 94.8 & 95.5 \\
11 & 96.1 & 94.9 & 95.5 \\
12 & 94.9 & 95.7 & 95.3 \\
13 & 94.2 & 93.9 & 95.0 \\
14 & 94.6 & 93.4 & 94.2 \\
15 & 95.3 & 94.3 & 95.8 \\
16 & 94.9 & 94.7 & 94.3 \\
\bottomrule
\end{tabular}%
}
\par\smallskip
\begin{minipage}{\textwidth}\footnotesize IPD-PO: pseudo-outcome regression; IPD-GF: g-formula. Coverage ranges across scenarios (minimum--maximum): CIMAgD: 93.5--96.1\%; IPD-PO: 93.4--95.7\%; IPD-GF: 93.9--95.8\%.\end{minipage}
\end{table}

\clearpage
\begin{table}[htbp]
\centering
\small
\caption{Single trial, continuous outcome: MSE and MAE.}
\label{tab:1trial-continuous-add-mse_mae}
\resizebox{\textwidth}{!}{%
\begin{tabular}{rrrrrrr}
\toprule
Scenario & \multicolumn{2}{c}{CIMAgD} & \multicolumn{2}{c}{IPD-PO} & \multicolumn{2}{c}{IPD-GF} \\
 & MSE & MAE & MSE & MAE & MSE & MAE \\
\midrule
1 & 0.0218 & 0.1189 & 0.0276 & 0.1344 & 0.0183 & 0.1089 \\
2 & 0.0107 & 0.0821 & 0.0130 & 0.0917 & 0.0089 & 0.0752 \\
3 & 0.0272 & 0.1306 & 0.0311 & 0.1385 & 0.0219 & 0.1168 \\
4 & 0.0151 & 0.0976 & 0.0173 & 0.1055 & 0.0123 & 0.0883 \\
5 & 0.0246 & 0.1255 & 0.0248 & 0.1253 & 0.0173 & 0.1051 \\
6 & 0.0127 & 0.0887 & 0.0128 & 0.0898 & 0.0090 & 0.0754 \\
7 & 0.0227 & 0.1206 & 0.0265 & 0.1316 & 0.0179 & 0.1071 \\
8 & 0.0118 & 0.0861 & 0.0132 & 0.0913 & 0.0093 & 0.0773 \\
9 & 0.0195 & 0.1096 & 0.0249 & 0.1252 & 0.0173 & 0.1039 \\
10 & 0.0097 & 0.0783 & 0.0125 & 0.0886 & 0.0081 & 0.0714 \\
11 & 0.0230 & 0.1198 & 0.0318 & 0.1436 & 0.0205 & 0.1142 \\
12 & 0.0125 & 0.0884 & 0.0143 & 0.0950 & 0.0108 & 0.0834 \\
13 & 0.0199 & 0.1121 & 0.0244 & 0.1246 & 0.0170 & 0.1043 \\
14 & 0.0101 & 0.0785 & 0.0132 & 0.0924 & 0.0088 & 0.0738 \\
15 & 0.0190 & 0.1103 & 0.0250 & 0.1257 & 0.0166 & 0.1033 \\
16 & 0.0098 & 0.0789 & 0.0126 & 0.0905 & 0.0084 & 0.0728 \\
\bottomrule
\end{tabular}%
}
\par\smallskip
\begin{minipage}{\textwidth}\footnotesize IPD-PO: pseudo-outcome regression; IPD-GF: g-formula.\end{minipage}
\end{table}

\end{appendices}

\end{doublespace}
\end{document}